%% file: main.tex
\documentclass[a4paper,USenglish,cleveref, autoref, thm-restate]{lipics-v2021}
\title{Generalized Bidding Games:\\ Where Bidding and Stochastic Games Meet} %TODO Please add

\titlerunning{Generalized Bidding Games} %TODO optional, please use if title is longer than one line
\author{Ali Asadi}{Institute of Science and Technology Austria, Austria \and \url{https://ali-asadi.com/}}{ali.asadi@ista.ac.at}{https://orcid.org/0009-0005-2839-953X}{}%{[funding]}

\author{Thomas A. Henzinger}{Institute of Science and Technology Austria, Austria \and \url{https://pub.ista.ac.at/~tah/} }{tah@ist.ac.at}{https://orcid.org/0000-0002-2985-7724}{}%{(Optional) author-specific funding acknowledgements}%TODO mandatory, please use full name; only 1 author per \author macro; first two parameters are mandatory, other parameters can be empty. Please provide at least the name of the affiliation and the country. The full address is optional. Use additional curly braces to indicate the correct name splitting when the last name consists of multiple name parts.

\author{Ehsan Kafshdar Goharshady}{Institute of Science and Technology Austria, Austria \and \url{https://ehsan.goharshady.com/}}{egoharsh@ist.ac.at}{https://orcid.org/0000-0002-8595-0587}{}%{[funding]}

\author{Pavol Kebis}{Institute of Science and Technology Austria, Austria}{pavol.kebis@ist.ac.at}{https://orcid.org/0000-0003-0561-1364}{}

\author{Kaushik Mallik}{IMDEA Software Institute, Spain \and \url{https://kmallik.github.io/}}{kaushik.mallik@imdea.org}{https://orcid.org/0000-0001-9864-7475}{}

\authorrunning{Asadi et al.} %TODO mandatory. First: Use abbreviated first/middle names. Second (only in severe cases): Use first author plus 'et al.'

\Copyright{Jane Open Access and Joan R. Public} %TODO mandatory, please use full first names. LIPIcs license is "CC-BY";  http://creativecommons.org/licenses/by/3.0/
\ccsdesc[500]{Theory of computation~Algorithmic game theory}
\ccsdesc[300]{Computing methodologies~Stochastic games}
\keywords{Bidding Games, Stochastic Games} %TODO mandatory; please add comma-separated list of keywords

\category{} %optional, e.g. invited paper

\relatedversion{} %optional, e.g. full version hosted on arXiv, HAL, or other respository/website
\funding{The research was partially supported 
by Austrian Science Fund (FWF) 10.55776/COE12, ERC CoG 863818 (ForM-SMArt), FWF-2022-SFB F8502 (SPyCoDe), ERC-2020-AdG 101020093 (VAMOS), and the RYC2024-049116-I grant funded by MICIU/AEI/10.13039/501100011033 and ESF+.}%optional, to capture a funding statement, which applies to all authors. Please enter author specific funding statements as fifth argument of the \author macro.

\nolinenumbers %uncomment to disable line numbering

\EventEditors{John Q. Open and Joan R. Access}
\EventNoEds{2}
\EventLongTitle{42nd Conference on Very Important Topics (CVIT 2016)}
\EventShortTitle{CVIT 2016}
\EventAcronym{CVIT}
\EventYear{2016}
\EventDate{December 24--27, 2016}
\EventLocation{Little Whinging, United Kingdom}
\EventLogo{}
\SeriesVolume{42}
\ArticleNo{23}
\usepackage{xspace}
\usepackage{mathtools}
\usepackage{tikz}
\usetikzlibrary{shapes.geometric, arrows.meta, calc, positioning}
\usepackage{wrapfig}
\usepackage{amsmath}
\usepackage{booktabs}
\input{macro}

\begin{document}

\maketitle

\begin{abstract}
    Two-player games on graphs are a classical framework for analyzing strategic decision making. In turn-based games, two players move a token along the edges of the graph, and the right to move the token is determined by the current vertex. In traditional bidding games---referred to as \textit{pure} bidding games---the right to move the token is determined at each step through bidding; here we consider Richman bidding, where the winning player of a bid pays the losing player. The winner is decided based on a temporal or quantitative specification evaluated over the resulting infinite play.

    In this work, we combine turn-based games and pure bidding games into \textit{generalized} bidding games, with player-1 vertices, player-2 vertices, and bidding vertices. This natural and simple generalization of bidding games has far-reaching consequences. First, we show that, as a model, generalized bidding games are more expressive than pure bidding games, and we provide several applications. Second, and most importantly, we show that generalized Richman bidding games are structurally equivalent to simple stochastic games, a well-studied model: they are linearly interreducible to each other. As was previously known, the special case of pure Richman bidding games corresponds to random-turn games. In other words, generalized bidding games extend pure bidding games in the same way that simple stochastic games extend random-turn games. We use this connection to solve generalized Richman bidding games for temporal (parity) and quantitative (mean-payoff and discounted-sum) specifications. From a computational perspective, we establish that generalized bidding games with parity and mean-payoff specifications retain the best known upper bounds for turn-based games and pure bidding games, namely NP$\,\cap\,$coNP.

    Finally, we study a repair problem that asks whether bidding vertices can be assigned ``owners'' so as to bring the threshold budget required to win the game below a given target. This problem has direct applications in compositional policy synthesis for multi-objective settings, and we show it to be NP-complete.
\end{abstract}

\input{intro}
%%
%% Bibliography
%%
\input{model}
% \input{hardness}
\input{SSG_short}

\input{qualitative}
\input{quantitative}

\input{repair}

\section{Conclusions and Future Directions}

We introduce generalized bidding games (\GBG) that subsume and extend both pure bidding games (\BG) and turn-based games.
We establish that thresholds exist for \GBGs for a variety of qualitative and quantitative specifications, like parity and discounted-sum, and establish a structural equivalence between the thresholds of simple \GBGs and SSGs.
As the complexity of computing thresholds for \GBGs is no more than the best-known complexity for \BGs, we hope \GBGs will become the standard model for studying bidding games in the future.

Several questions remain unresolved.
Firstly, we only considered first-price Richman bidding, but many other alternatives exist in the literature.
We conjecture that, apart from the structural equivalence result with SSGs, the existence and computation of thresholds will be extendible to other settings in a straightforward manner.
Secondly, we established that the threshold computation problem for discounted-sum specifications turns out to be equivalent to the quantitative satisfaction value problem in SSGs. Since the latter is open, resolving either problem is an interesting open question that would simultaneously advance the other.
% \todo{add one or two more}

%% Please use bibtex, 

\bibliography{bibliography}

\appendix

\input{appendix}

\end{document}

%% file: macro.tex
\newtheorem{problem}{Problem}
\newtheorem*{theorem*}{Theorem}

\newcommand{\myparagraph}[1]{\smallskip\noindent\textbf{#1}}
\newcommand{\set}[1]{\left\lbrace #1\right\rbrace}
\newcommand{\argmax}{\text{argmax}}
\newcommand{\argmin}{\text{argmin}}

\mathchardef\mhyphen="2D % Define a "math hyphen"

\newcommand{\eve}{\mathtt{E}}
\newcommand{\adam}{\mathtt{A}}
\newcommand{\R}{\mathbb{R}}

\newcommand{\GBG}{GBG\xspace}
\newcommand{\GBGs}{GBGs\xspace}
\newcommand{\SGBG}{SGBG\xspace}
\newcommand{\SGBGs}{SGBGs\xspace}
\newcommand{\BG}{BG\xspace}
\newcommand{\BGs}{BGs\xspace}
\newcommand{\SBG}{SBG\xspace}
\newcommand{\SBGs}{SBGs\xspace}

\newcommand{\G}{\mathcal{G}}
\newcommand{\game}{\G}
\newcommand{\V}{V}
\newcommand{\VZ}{V_\eve}
\newcommand{\VO}{V_\adam}

\newcommand{\VB}{V_{\mathtt{b}}}
\newcommand{\VR}{V_{\mathtt{r}}}
\newcommand{\vin}{v_{\mathtt{init}}}
\newcommand{\E}{E}
\newcommand{\out}{\mathit{out}}

\renewcommand{\S}{\mathcal{S}}
\newcommand{\val}{\mathit{val}}

\newcommand{\optval}{{\val^*}}
\newcommand{\SSG}{\Gamma}

\newcommand{\VALUE}{\mathit{VALUE}}
\newcommand{\THRESH}{\mathit{THRESH}}

\newcommand{\PZ}{Eve\xspace}
\newcommand{\PO}{Adam\xspace}

\newcommand{\polZ}{\pi_\eve}
\newcommand{\polO}{\pi_\adam}

\newcommand{\PolZ}{\Pi_\eve}
\newcommand{\PolO}{\Pi_\adam}

\newcommand{\ppath}{\mathrm{path}}
\newcommand{\finpath}[1]{\mathrm{PathsFin}^{#1}}
\newcommand{\infpath}[1]{\mathrm{PathsInf}^{#1}}
\newcommand{\last}{\mathit{last}}

\newcommand{\spec}{\varphi}
\newcommand{\DS}{\mathsf{DS}}
\newcommand{\MP}{\mathsf{MP}}
\newcommand{\discsum}{\DS}
\newcommand{\meanpayoff}{\MP}
\newcommand{\reach}{\mathsf{Reach}}
\newcommand{\safe}{\mathsf{Safe}}
\newcommand{\parity}{\mathsf{Parity}}

\newcommand{\thresh}{\mathit{Th}}
\newcommand{\rTh}{\mathit{rTh}}
\newcommand{\sTh}{\mathit{sTh}}
\newcommand{\bellman}{\mathcal{L}}
\newcommand{\bellmanSSG}{\mathcal{B}}

\newcommand{\pTh}{\mathit{pTh}}
\newcommand{\fpTh}{\mathit{fpTh}}

\newcommand{\Bop}{\mathcal{B}}

\newcommand{\valcand}{\psi}

\renewcommand{\inf}{\mathop{\mathrm{inf}\vphantom{\mathrm{sup}}}}

%% file: intro.tex
\section{Introduction}

Games on graphs are central in computer science, with applications from automata theory to the synthesis of reactive systems. In the classical turn-based setting, two players, Eve and Adam, move a token on a graph, generating an infinite path; Eve wins if the path satisfies a temporal or quantitative specification, and the goal is to compute a winning strategy for her.

An important extension is \textit{bidding games}, where turns are determined by auctions rather than fixed ownership. 
Players bid using finite budgets to gain control of the next move, with payments governed by a model such as first-price Richman bidding, where the players are initially allocated budgets whose sum is $1$, and after each bidding, the higher bidder picks the next move and pays the bid amount to the opponent, so that the total budget in the game remains $1$.
% Unless otherwise mentioned, we will use this payment model in the entire paper.
Bidding games model situations where resources must be invested while making strategic decisions in the presence of opponents.
Applications include online advertisement platforms~\cite{avni2018infinite} and decoupled control of cyber-physical systems~\cite{avni2024auction}.

Often, decision-making problems involve the interleaving of bidding and turn-based mechanisms, but there is no existing model that combines them in one framework.
To model this, we introduce \textit{generalized} bidding games (\GBG), where bidding takes place in a given subset of vertices---called the \textit{bidding vertices}, while each non-bidding vertex---called a \textit{control vertex}---is either owned by Eve or by Adam.
When the token lands on a bidding vertex, who moves the token is decided by (first-price Richman) bidding.
When the token lands on a control vertex, no bidding takes place. 
Instead, it is the owner's turn to choose where the token goes next.
This way, our model unifies bidding and turn-based games into a new class of games that is more general than each of the two.
To mark the distinction, we refer to the traditional bidding games as \textit{pure} bidding games or BG in short, because turns are assigned purely based on bidding.

We claim that \GBGs are strictly more powerful than BGs and turn-based games.
One direction is easy: \GBGs are at least as expressive as BGs and turn-based games, because these are restricted variants.
For the other direction, we show that the solution of \GBGs requires a \textit{quantitative} analysis (which is also sufficient as we will show in the subsequent sections), while qualitative analyses suffice for the two restricted cases.

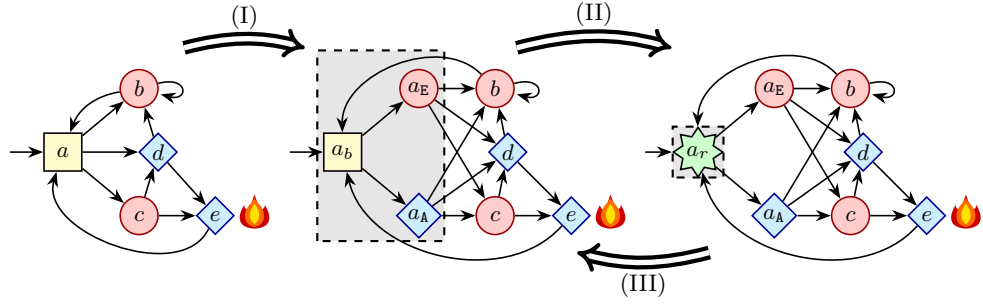
\begin{figure}
    \centering
    \input{FIGURES/expressiveness}
    \caption{An overview of our translation from generalized bidding games (\GBG) to simple stochastic games (SSG).
    LEFT: the original \GBG; CENTER: the equivalent simple \GBG (\SGBG), i.e., a \GBG whose bidding vertices have exactly two successors; RIGHT: the equivalent SSG.
    The yellow squares are the bidding vertices, the green star is a random vertex, the red circles are Eve's vertex, and the blue diamonds are Adam's vertices.
    The grey box highlights the main change in each step while going from left to right.
    Eve's specification is to avoid the unsafe vertex $e$ forever. 
    }
    \label{fig:non-trivial threshold in SCC}
\end{figure}

\begin{example}[\GBGs versus BGs and turn-based games]\label{ex:safety game on SCC} 
    Consider the \GBG shown in the left diagram in Figure~\ref{fig:non-trivial threshold in SCC}, where Eve's goal is to avoid the unsafe vertex $e$ at all times (safety).
    In this game, Eve wins from $a$ if her budget is strictly more than $\frac{1}{2}$, because Adam's budget is less than $\frac{1}{2}$ (since the sum of players' budgets is $1$), so that Eve can outbid Adam at $a$, move the token to $b$, and then produce the sequence $b^\omega$, thereby surely avoiding $e$ at all time.
    Dually, Eve loses from $a$ if her budget is strictly less than $\frac{1}{2}$, because then Adam can outbid Eve, move to $d$, and then reach $e$, violating Eve's safety specification.
    We say that the threshold budget in $a$ is $\frac{1}{2}$.
    %Another illustrative example of \GBGs is provided in Appendix~\ref{appendix:sec:bid-tac-toe}.

    If instead the game were a BG, with all vertices being bidding vertices, Eve would not be able to guarantee safety in the long-run with any amount of initial budget.
    This is because the underlying graph is strongly connected, for which the analysis of thresholds is qualitative: if the set of unsafe vertices is non-empty, Eve loses with \textit{any} budget, and if the set of unsafe vertices is empty, Eve can win with \textit{any} budget \cite{avni2020survey}.
    In general, on strongly connected graphs, a qualitative analysis of the specification suffices for solving \BGs, whereas \GBGs require a more involved quantitative analysis.
    We formalize this claim in Appendix~\ref{app:sec:distinction}.

    Similarly, if the game were a turn-based game, with all vertices being control vertices, the discussion of budgets would not be necessary, and every vertex would either be winning or losing for Eve.
    Finding the winning vertices in any turn-based game is based on a combinatorial (qualitative) analysis of the graph, which is again less involved than the quantitative analysis necessary for \GBGs.    
    \qed
\end{example}

Like \BGs, the analysis of \GBGs revolves around the concept of \textit{threshold budget}, which is the necessary and sufficient budget of Eve to win from a given initial vertex.
For instance, the threshold budget in Example~\ref{ex:safety game on SCC} was $\frac{1}{2}$.
The existence of threshold budgets is a fundamental theoretical question, which is a variant of determinacy.\footnote{The determinacy of zero-sum games implies every state is winning for one of the players.} 
We prove that threshold budgets exist for \GBGs with parity, discounted-sum, and mean-payoff specifications.
%(For quantitative specifications like discounted-sum and mean-payoff, threshold budgets are no longer constants, but rather functions of the payoff; see Section~\ref{sec:model}.)
Incidentally, we are the first to consider discounted-sum specifications for any form of bidding games.

Our main result is that \textit{simple} \SGBGs,
where every bidding vertex has exactly two outgoing edges,
are \textit{structurally equivalent} to \textit{simple stochastic games} (SSG), which have been extensively studied in the literature \cite{condon1990algorithms,ChatterjeeJH04}.
By structural equivalence we refer to a simple reinterpretation of the same game graph.
In SSGs, the game graph contains the usual control vertices---owned by Eve and Adam---and \textit{random} vertices with two outgoing edges, such that from every random vertex, the token moves to one of the two successors with equal probabilities.
The \textit{value} of an SSG is the maximal probability with which Eve can fulfill her specification from the initial vertex against any strategy of Adam.
We show that, by simply replacing each bidding vertex of each \SGBG with a random vertex, as illustrated in Step~(II) of Figure~\ref{fig:non-trivial threshold in SCC}, the value of the resulting SSG equals $1$ minus the threshold of the original \GBG.
The opposite direction also holds true: by replacing every random vertex of a given SSG with a bidding vertex, as illustrated in Step~(III) of Figure~\ref{fig:non-trivial threshold in SCC}, the threshold of the resulting \SGBG equals $1$ minus the value of the original SSG.
An important consequence is that any tool developed for either \SGBG or SSG supports the other out of the box, without involving any intermediate reduction step.

While the structural equivalence of \SGBGs and SSGs is uniform for all specifications we consider---both qualitative specifications such as parity and quantitative specifications such as discounted-sum and mean-payoff---we need to prove the equivalence for each class of specifications separately.
In particular, the different specifications require different proof techniques to establish the preservation of strategies and threshold budgets.

It was already known that a similar structural equivalence (except for discounted-sum specifications, which were not studied before in this context) exists between \BGs and \textit{random-turn games} (RTG), where RTGs are the special subclass of SSGs in which every random vertex $v$ is succeeded by an Eve's vertex $v_\eve$ and an Adam's vertex $v_\adam$, and the edges originating from $v_\eve$ and $v_\adam$ connect to identical sets of vertices.
To put our result into perspective, \SGBGs extend \BGs in the same way that SSGs extend RTGs.
However, while little is known about the exact relation between SSGs and RTGs, 
we show in our second main result, 
that every \GBG (independent of the out-degree of bidding vertices) 
can be linearly translated into a \SGBG (with bidding out-degree~2).
This establishes a linear reduction of all \GBGs to SSGs;
the reduction gadget is shown in Step~(I) in Figure~\ref{fig:non-trivial threshold in SCC}.
The translation is such that the size of the resulting stochastic game is linear in the size of the original bidding game, and moreover, for parity, discounted-sum, and mean-payoff specifications, all thresholds and winning strategies remain intact.
Incidentally, a similar translation is unlikely to exist for \BGs, because \BGs with out-degree 2 (called simple \BG or SBG) are equivalent to RTGs with out-degree 2 and reduce to Markov chains, which can be solved in polynomial time; therefore, if a reduction from \BGs to \SBGs were possible, we would have a polynomial-time algorithm for \BGs (with arbitrary out-degree)---resolving a long-standing open problem (the best known complexity for the general out-degree case is NP$\,\cap\,$coNP).

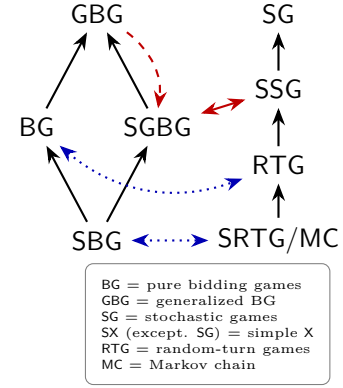
\begin{wrapfigure}{r}{5cm}
    \vspace{-0.4cm}
    \input{FIGURES/hierarchy.tex}
    \caption{Connection between bidding games and stochastic games.
    The red arrow (our result) and dotted blue arrows (known) are structural equivalences.
    The red dashed arrow is our reduction (polynomial-time). 
    The black arrows go from special cases to more general cases.
    % Abbreviations: BG = pure bidding games, GBG = generalized bidding games, RTG = random-turn games, SX = simple X, meaning two out-degree bidding/random vertex with uniform randomness in case of random vertex, MC = Markov chain, SG = stochastic games, SSG = simple stochastic games.
    }
    \label{fig:game relationships}
    \vspace{-0.4cm}
\end{wrapfigure}
Complexity-wise, for reachability, safety, parity, and mean-payoff specifications, computing values for both SSGs and RTGs are known to be in NP$\,\cap\,$coNP, and the membership in P is a long-standing open problem.
Therefore, our results imply that computing thresholds of \GBGs is in NP$\,\cap\,$coNP for all considered specification classes.
On the other hand, the best known algorithms for \BGs are also in NP$\,\cap\,$coNP, which is a corollary of their connection to RTGs.
Summarizing, not only do \GBGs subsume and extend \BGs, but also they allow a seamless transition to the extensively studied SSGs, without paying any additional computational blow-up.
The relationships obtained between different classes of \GBGs and stochastic games is shown in Figure~\ref{fig:game relationships}.

Our work is the first to consider discounted-payoff specifications in the context of any form of bidding games, and we show that the structural equivalence between SSGs and \SGBGs also holds in this case.
However, this does not provide an immediate upper complexity bound, because finding values for SSGs with discounted-sum specifications is an open problem.\footnote{Here, we are dealing with the quantitative satisfaction variant of the problem, aka the ``quantile'' problem, where the goal is to determine the maximum \textit{probability} of securing a discounted-sum above a given threshold. This should not be confused with the problem of maximizing the \textit{expected} discounted sum, for which solutions exist in the literature.}
Nonetheless, our result indicates that solving this problem for either \SGBGs or SSGs will immediately provide a solution for the other category.

\begin{figure}
    \centering
    % \input{FIGURES/tic-tac-toe.tex}
    % \vspace{0.2cm}
    \input{FIGURES/auction-based-scheduling}
    \caption{
    Application of \GBGs in multi-objective control via auction-based scheduling~\cite{avni2024auction}; the left and right subfigures show the solutions using \BG and \GBG, respectively.
    Eve's and Adam's objectives are to reach any of the red and the blue flags, respectively.
    The numbers next to the vertices represent the personal threshold budgets of Eve and Adam as if each of them were playing separately against an adversary.
    While composition fails in the \BG (since the sum of thresholds is more than $1$), it succeeds using \GBGs.
    }
    \label{fig:motivating example:auction-based scheduling}
\end{figure}
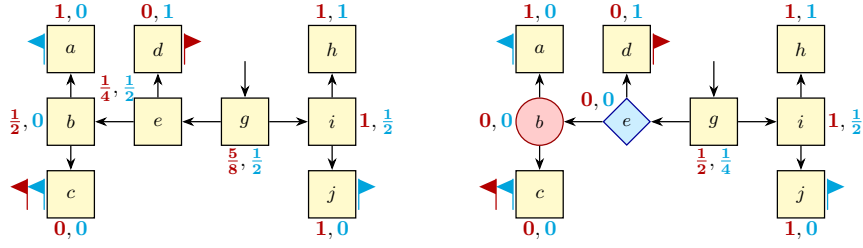

\subsection*{Motivating Example: Application in Multi-Objective Decision-Making}
% \begin{figure}
%     \centering
%     \input{FIGURES/auction-based-scheduling.tex}
%     \caption{Application: auction-based scheduling, where Eve and Adam have non-complementary reachability objectives with red and blue targets, respectively (shown using flags of respective colors, i.e., the red target set is $\set{c,d}$ and blue target set is $\set{a, c, j}$).
%     The legend of different vertices is as in Figure~\ref{fig:tic-tac-toe}, and the numbers next to the vertices represent the personal threshold budgets of Eve and Adam as if each of them were playing against an adversary.
%     \textit{LEFT, the usual setup with classical bidding games:} the sum of thresholds of the players at $g$ is $\frac{5}{8}+\frac{1}{2}$ which is more than $1$, so auction-base scheduling fails.
%     \textit{RIGHT, the new setup with generalized bidding games:} by converting a small subset of bidding vertices to control vertices, the sum of thresholds at $g$ becomes $\frac{1}{2}+\frac{1}{4}$ which is less than $1$, and therefore auction-based scheduling succeeds.
%     }
%     \label{fig:auction-based scheduling}
% \end{figure}

    Auction-based scheduling~\cite{avni2024auction} provides a decoupled approach to multi-objective path planning, where independent strategies for objectives $\varphi$ and $\varphi'$ are combined via bidding. Each objective is solved as a zero-sum bidding game, yielding local strategies and threshold budgets $\theta$ and $\theta'$. If $\theta + \theta' < 1$, the budget can be split so that both objectives are satisfied under composition.

A limitation is that treating the opponent as fully adversarial can make $\theta + \theta' \ge 1$. This can be alleviated by assigning permanent control of some vertices to players, lowering thresholds for critical objectives, though at the cost of added coordination; this is illustrated in Figure~\ref{fig:motivating example:auction-based scheduling}.
This is a direct application of our model of \GBGs.

This application leads to the \textit{threshold repair problem}: given a \GBG, a target threshold $t \in (0,1)$, and a budget $n$, can we select at most $n$ bidding vertices to assign to Eve so that her threshold drops below $t$? A dual version for Adam's threshold can be posed in the same way. We show this problem is NP-complete for reachability, safety, and parity specifications, and it can be encoded as a mixed-integer logical program (as shown in Appendix~\ref{appendix:sec:milp encoding of repair}) and solved efficiently in practice.

\subsection*{Related Work}
Bidding games have a long history, going back to the seminal work of Lazarus et al.~\cite{lazarus1999combinatorial} in the nineties.
This work studied bidding games with reachability specifications, where various different bidding mechanisms were introduced, including Richman, poorman, taxman, etc.
Later, bidding games were extended to richer classes of specifications, like temporal (parity, etc.) and quantitative specifications (mean-payoff, etc.)~\cite{avni2019infinite,avni2018infinite}, richer classes of arenas, like Markov decision processes~\cite{avni2025bidding}, and richer bidding mechanisms, like games with discrete bidding currencies~\cite{develin2008discrete,avni2025computing} and games with charging~\cite{AvniGHM24}, where the players can top-up their budgets by collecting rewards during the game.
All these extensions are orthogonal to our setting, and therefore they can be combined together.

In terms of the characteristics of our \GBGs, the closest models are bidding games with charging~\cite{AvniGHM24} and discrete bidding games~\cite{avni2025computing}.
As in our case, both of these models exhibit non-trivial (other than $0$ or $1$) thresholds for safety and parity specifications on strongly connected game graphs.
However, these games introduce additional mechanisms to the bidding step, and they do not provide a clean connection to stochastic games like our approach does.
For instance, in discrete-bidding games~\cite{avni2025computing}, resolving bidding ties play a central role while computing strategies.
In bidding games with charging~\cite{AvniGHM24}, a normalization step is required to ensure that the sum of budgets of the players remain $1$ after every round of charging.
In contrast, our model offers a very simple and clean extension of \BGs.

%% file: FIGURES/expressiveness.tex
\definecolor{biddingcol}{HTML}{FFFACD}
\definecolor{adamcol}   {HTML}{CCEEFF}
\definecolor{evecol}  {HTML}{FFCCCC}
\definecolor{rdncol}{HTML}{CCFFCC}

\tikzset{
  nd/.style  = {rectangle, draw=black, semithick,  minimum width=0.5cm, minimum height=0.5cm,
                font=\footnotesize\itshape, fill=biddingcol, inner sep=0pt},
  evn/.style={circle, draw=red!60!black, semithick, fill=evecol,
              minimum size=0.5cm, font=\footnotesize\bfseries, inner sep=1pt},
  adn/.style={diamond, draw=blue!60!black, semithick,
              fill=adamcol, minimum size=0.5cm,
              font=\footnotesize\bfseries, inner sep=1pt},
  rdn/.style = {star, star points=8, star point ratio=1.5,
               draw=black, semithick, fill=rdncol,
               minimum size=0.5cm, font=\footnotesize\itshape, inner sep=0pt},
  arr/.style = {->, >=Stealth, semithick},
  rfrac/.style = {font=\scriptsize, text=red!70!black},
  bfrac/.style = {font=\scriptsize, text=blue!60!black},
}

\def\nd{1.2}    % <── node distance: adjust this one value
\def\hw{0.39}   % half node width/height (match nd style)

\newcommand{\tikzflame}{%
  %% ── outer red blob ────────────────────────────────────────────────────────
  \fill[red!80!black]
    (0, 0)
    .. controls (-0.38, 0.02) and (-0.44, 0.22) .. (-0.38, 0.42)
    .. controls (-0.34, 0.56) and (-0.26, 0.64) .. (-0.24, 0.78)
    .. controls (-0.20, 0.68) and (-0.22, 0.56) .. (-0.16, 0.50)
    .. controls (-0.10, 0.62) and (-0.06, 0.76) .. ( 0,    0.92)
    .. controls ( 0.06, 0.76) and ( 0.10, 0.62) .. ( 0.16, 0.50)
    .. controls ( 0.22, 0.56) and ( 0.20, 0.68) .. ( 0.24, 0.78)
    .. controls ( 0.26, 0.64) and ( 0.34, 0.56) .. ( 0.38, 0.42)
    .. controls ( 0.44, 0.22) and ( 0.38, 0.02) .. ( 0,    0)
    -- cycle;
  %% ── white cutout left ─────────────────────────────────────────────────────
  \fill[white]
    (-0.16, 0.50)
    .. controls (-0.26, 0.52) and (-0.32, 0.62) .. (-0.26, 0.74)
    .. controls (-0.20, 0.80) and (-0.14, 0.76) .. (-0.12, 0.70)
    .. controls (-0.08, 0.62) and (-0.10, 0.54) .. (-0.16, 0.50)
    -- cycle;
  %% ── white cutout right ────────────────────────────────────────────────────
  \fill[white]
    ( 0.16, 0.50)
    .. controls ( 0.26, 0.52) and ( 0.32, 0.62) .. ( 0.26, 0.74)
    .. controls ( 0.20, 0.80) and ( 0.14, 0.76) .. ( 0.12, 0.70)
    .. controls ( 0.08, 0.62) and ( 0.10, 0.54) .. ( 0.16, 0.50)
    -- cycle;
  %% ── orange inner flame ────────────────────────────────────────────────────
  \fill[orange!85!red]
    ( 0,    0.04)
    .. controls (-0.26, 0.10) and (-0.30, 0.32) .. (-0.22, 0.52)
    .. controls (-0.14, 0.66) and (-0.06, 0.72) .. ( 0,    0.82)
    .. controls ( 0.06, 0.72) and ( 0.14, 0.66) .. ( 0.22, 0.52)
    .. controls ( 0.30, 0.32) and ( 0.26, 0.10) .. ( 0,    0.04)
    -- cycle;
  %% ── yellow core ───────────────────────────────────────────────────────────
  \fill[yellow!85!orange]
    ( 0,    0.12)
    .. controls (-0.14, 0.22) and (-0.16, 0.40) .. (-0.08, 0.54)
    .. controls (-0.03, 0.64) and ( 0,    0.66) .. ( 0,    0.68)
    .. controls ( 0,    0.66) and ( 0.03, 0.64) .. ( 0.08, 0.54)
    .. controls ( 0.16, 0.40) and ( 0.14, 0.22) .. ( 0,    0.12)
    -- cycle;
}

\def\fs{0.7}   % <── flame scale, adjust to taste

\begin{tikzpicture}
\def\Hgap{4.2}   % horizontal gap between scope origins

\tikzset{
  trans/.style={-{Implies}, double, double distance=2pt,
                ultra thick, shorten >=3pt, shorten <=3pt}
}

%% ── P1: left ─────────────────────────────────────────────────────────────────
\begin{scope}[shift={(0,0)}, scale=0.7]
    % \draw[dashed,thick, fill=black!10!white] (-0.4*\nd,-0.4*\nd) rectangle (0.4*\nd,0.4*\nd);
    \node[nd]  (pa)  at (0,0)           {$a$};
    \node[evn] (pb)  at (1.2*\nd,\nd)  {$b$};
    \node[evn] (pc)  at (1.2*\nd,-\nd) {$c$};
    \node[adn] (pd)  at (1.5*\nd,0)    {$d$};
    \node[adn] (pe)  at (2.4*\nd,-\nd) {$e$};
    \begin{scope}[shift={($(pe)+(\hw+0.35,-0.35*\fs)$)}, scale=\fs]
        \tikzflame
    \end{scope}
    \draw[arr] ($(pa)+(-1,0)$) -- (pa);
    \draw[arr] (pa) -- (pb);
    \draw[arr] (pa) -- (pd);
    \draw[arr] (pa) -- (pc);
    \draw[arr] (pd) -- (pb);
    \draw[arr] (pd) -- (pe);
    \draw[arr] (pc) -- (pd);
    \draw[arr] (pc) -- (pe);
    \path[arr] (pb) edge[loop right]    ();
    \path[arr] (pb) edge[bend right]    (pa);
    \path[arr] (pe) edge[bend left=3cm] (pa);
\end{scope}

%% ── P2: middle ───────────────────────────────────────────────────────────────
\begin{scope}[shift={(\Hgap-0.5,0)}, scale=0.7]
    \draw[dashed,thick, fill=black!10!white] (-0.4*\nd,-1.4*\nd) rectangle (1.6*\nd,1.6*\nd);
    % \node   at  (-0.8,1.4)   {(I)};
    % \draw[dashed,thick, fill=black!10!white] (-0.4*\nd,-0.4*\nd) rectangle (0.4*\nd,0.4*\nd);
    % \node   at  (1,0)   {(III)};
    \node[nd]  (pba)  at (0,0)           {$a_b$};
    \node[evn] (pbae) at (1.2*\nd,\nd)  {$a_\eve$};
    \node[adn] (pbaa) at (1.2*\nd,-\nd) {$a_\adam$};
    \node[evn] (pbb)  at (2.4*\nd,\nd)  {$b$};
    \node[evn] (pbc)  at (2.4*\nd,-\nd) {$c$};
    \node[adn] (pbd)  at (2.6*\nd,0)    {$d$};
    \node[adn] (pbe)  at (3.6*\nd,-\nd) {$e$};
    \begin{scope}[shift={($(pbe)+(\hw+0.35,-0.35*\fs)$)}, scale=\fs]
        \tikzflame
    \end{scope}
    \draw[arr] ($(pba)+(-1,0)$) -- (pba);
    \draw[arr] (pba)  -- (pbae);
    \draw[arr] (pba)  -- (pbaa);
    \draw[arr] (pbae) -- (pbb);
    \draw[arr] (pbae) -- (pbd);
    \draw[arr] (pbae) -- (pbc);
    \draw[arr] (pbaa) -- (pbb);
    \draw[arr] (pbaa) -- (pbd);
    \draw[arr] (pbaa) -- (pbc);
    \draw[arr] (pbd)  -- (pbb);
    \draw[arr] (pbd)  -- (pbe);
    \draw[arr] (pbc)  -- (pbd);
    \draw[arr] (pbc)  -- (pbe);
    \path[arr] (pbb) edge[loop right]       ();
    \path[arr] (pbb) edge[bend right=2.4cm] (pba);
    \path[arr] (pbe) edge[bend left=2.2cm]  (pba);
\end{scope}

%% ── P3: right ────────────────────────────────────────────────────────────────
\begin{scope}[shift={(2*\Hgap,0)}, scale=0.7]
    \draw[dashed,thick, fill=black!10!white] (-0.4*\nd,-0.4*\nd) rectangle (0.4*\nd,0.4*\nd);
    \node[rdn] (pca)  at (0,0)           {$a_r$};
    \node[evn] (pcae) at (1.2*\nd,\nd)  {$a_\eve$};
    \node[adn] (pcaa) at (1.2*\nd,-\nd) {$a_\adam$};
    \node[evn] (pcb)  at (2.4*\nd,\nd)  {$b$};
    \node[evn] (pcc)  at (2.4*\nd,-\nd) {$c$};
    \node[adn] (pcd)  at (2.6*\nd,0)    {$d$};
    \node[adn] (pce)  at (3.6*\nd,-\nd) {$e$};
    \begin{scope}[shift={($(pce)+(\hw+0.35,-0.35*\fs)$)}, scale=\fs]
        \tikzflame
    \end{scope}
    \draw[arr] ($(pca)+(-1,0)$) -- (pca);
    \draw[arr] (pca)  -- (pcae);
    \draw[arr] (pca)  -- (pcaa);
    \draw[arr] (pcae) -- (pcb);
    \draw[arr] (pcae) -- (pcd);
    \draw[arr] (pcae) -- (pcc);
    \draw[arr] (pcaa) -- (pcb);
    \draw[arr] (pcaa) -- (pcd);
    \draw[arr] (pcaa) -- (pcc);
    \draw[arr] (pcd)  -- (pcb);
    \draw[arr] (pcd)  -- (pce);
    \draw[arr] (pcc)  -- (pcd);
    \draw[arr] (pcc)  -- (pce);
    \path[arr] (pcb) edge[loop right]       ();
    \path[arr] (pcb) edge[bend right=2.4cm] (pca);
    \path[arr] (pce) edge[bend left=2.2cm]  (pca);
\end{scope}

%% ── Transition arrows ────────────────────────────────────────────────────────
%% 1 → 2 : from top-right of P1, arc upward, land above P2
\draw[trans]
    (1.5, 1.35)
    to[out=15, in=165]
    node[midway, above, font=\small] {(I)}
    (\Hgap-0.9, 1.35);

%% 2 → 3 : from top-right of P2, arc upward, land above P3
\draw[trans]
    (\Hgap+1.7, 1.35)
    to[out=15, in=165]
    node[midway, above, font=\small] {(II)}
    (2*\Hgap-0.25, 1.35);

%% 3 → 2 : from below P3, arc downward, land below P2
\draw[trans]
    (2*\Hgap+0.25, -1.3)
    to[out=-165, in=-15]
    node[midway, below, font=\small] {(III)}
    (\Hgap+2.5, -1.3);

\end{tikzpicture}

%% file: FIGURES/hierarchy.tex
\begin{tikzpicture}

%% ── Spacing parameters (adjust these) ───────────────────────────────────────
\def\Hsp{0.8}   % horizontal spacing between columns
\def\Vsp{1}   % vertical spacing between rows

%% ── Node positions (col * \Hsp, row * \Vsp) ─────────────────────────────────
%% Col:  0=GBG, 1=SGBG, 2=SSG/SG/RTG   Row: 0=SBG/SRTG, 1=BG/SGBG/RTG, 2=GBG/SSG, 3=SG
\coordinate (posBG)    at (0,          1.5*\Vsp);
\coordinate (posGBG)   at (1*\Hsp,     3*\Vsp);
\coordinate (posSGBG)  at (2*\Hsp,     1.5*\Vsp);
\coordinate (posSSG)   at (4*\Hsp,     2*\Vsp);
\coordinate (posSG)    at (4*\Hsp,     3*\Vsp);
\coordinate (posRTG)   at (4*\Hsp,     1*\Vsp);
\coordinate (posSBG)   at (1*\Hsp,     0);
\coordinate (posSRTG)  at (4*\Hsp,     0);

%% ── Arrow styles ─────────────────────────────────────────────────────────────
\tikzset{
  blk/.style  = {->, >=Stealth, thick, black},
  redarr/.style = {<->, >=Stealth, thick, red!75!black},
  bluarr/.style = {<->, >=Stealth, thick, blue!70!black,dotted},
}

%% ── Nodes ────────────────────────────────────────────────────────────────────
\node (BG)    at (posBG)   {$\mathsf{BG}$};
\node (GBG)   at (posGBG)  {$\mathsf{GBG}$};
\node (SGBG)  at (posSGBG) {$\mathsf{SGBG}$};
\node (SSG)   at (posSSG)  {$\mathsf{SSG}$};
\node (SG)    at (posSG)   {$\mathsf{SG}$};
\node (RTG)   at (posRTG)  {$\mathsf{RTG}$};
\node (SBG)   at (posSBG)  {$\mathsf{SBG}$};
\node (SRTG)  at (posSRTG) {$\mathsf{SRTG/MC}$};

%% dashed red circle around SGBG
% \draw[red!75!black, dashed, thick] (SGBG) circle (0.55cm);

%% ── Black arrows ─────────────────────────────────────────────────────────────
\draw[blk] (BG)   -- (GBG);
\draw[blk] (SGBG) -- (GBG);
\draw[blk] (SBG)  -- (BG);
\draw[blk] (SBG)  -- (SGBG);
\draw[blk] (RTG)  -- (SSG);
\draw[blk] (SRTG) -- (RTG);
\draw[blk] (SSG)  -- (SG);

%% ── Red arrows ───────────────────────────────────────────────────────────────
%% GBG <-> SG with × crossing mark in the middle
% \draw[redarr] (GBG) -- (SG)
%     node[midway, fill=white, inner sep=1pt,
%          font=\normalsize, text=red!75!black] {$\times$};

\draw[redarr] (SGBG) -- (SSG);

\draw[->,>=Stealth, thick, red!75!black, dashed]   (GBG) edge[bend left] (SGBG);

%% ── Blue arrows ──────────────────────────────────────────────────────────────
\draw[bluarr] (BG)   edge[bend right] (RTG);
\draw[bluarr] (SBG)  -- (SRTG);

%% ── Abbreviation box ─────────────────────────────────────────────────────────
\node[draw=black!50, thin, rounded corners=3pt,
      fill=white, inner sep=6pt,
      anchor=north west, align=left, font=\tiny]
  at (0.8*\Hsp, -0.3*\Vsp) {%
  $\mathsf{BG}$ = pure bidding games\\
  $\mathsf{GBG}$ = generalized BG\\
  $\mathsf{SG}$ = stochastic games\\
  $\mathsf{SX}$ (except. $\mathsf{SG}$) = simple $\mathsf{X}$\\
  $\mathsf{RTG}$ = random-turn games\\
  $\mathsf{MC}$ = Markov chain
  %\\
  % $\mathsf{SSG}$ = simple stochastic games%
};

\end{tikzpicture}

%% file: FIGURES/auction-based-scheduling.tex
\definecolor{biddingcol}{HTML}{FFFACD}
\definecolor{adamcol}   {HTML}{CCEEFF}
\definecolor{evecol}  {HTML}{FFCCCC}

\tikzset{
  nd/.style  = {rectangle, draw=black, semithick,  minimum width=0.78cm, minimum height=0.78cm,
                font=\small\itshape, fill=biddingcol, inner sep=0pt},
  evn/.style={circle, draw=red!60!black, semithick, fill=evecol,
              minimum size=0.78cm, font=\footnotesize\bfseries, inner sep=1pt},
  adn/.style={diamond, draw=blue!60!black, semithick,
              fill=adamcol, minimum size=0.78cm,
              font=\footnotesize\bfseries, inner sep=1pt},
  arr/.style = {->, >=Stealth, semithick},
  rfrac/.style = {font=\scriptsize, text=red!70!black},
  bfrac/.style = {font=\scriptsize, text=blue!60!black},
}

% \bidlabel{node}{placement}{red number}{blue number}
% placement: above, below, left, right, above left, etc.
\newcommand{\bidlabel}[4]{%
  \node[font=\normalsize, #2] at (#1)
    {$\textcolor{red!70!black}{\mathbf{#3}},\textcolor{cyan!90!black}{\mathbf{#4}}$};
}

\def\hw{0.39}   % half node width/height (match nd style)
\def\nd{1.2}    % <── node distance: adjust this one value
\def\fh{0.5}   % <── flag height: adjust this one value

%% flag: pole + triangle pennant; origin = base of pole
%% usage: \tikzflag{color}  inside a scope already shifted to position
\newcommand{\tikzflag}[1]{%
  \draw[#1, semithick, line cap=round] (0,0) -- (0,0.28);
  \fill[#1] (0,0.28) -- (0.16,0.21) -- (0,0.14) -- cycle;
}

\scalebox{0.8}{
\begin{tikzpicture}
%% ── plain nodes ─────────────────────────────────────────────────────────────
\node[nd] (a) at (     0,  \nd) {$d$};
\node[nd] (b) at (     0,   0 ) {$e$};
% \node[nd] (c) at (     0, -\nd) {$f$};
\node[nd] (d) at (  1.2*\nd, 0  ) {$g$};
\node[nd] (e) at (  2.4*\nd, \nd) {$h$};
\node[nd] (f) at (  2.4*\nd, 0  ) {$i$};
\node[nd] (i) at ( -1.2*\nd, 0  ) {$b$};
%% ── a: red flag ─────────────────────────────────────────────────────────
\begin{scope}[shift={($(a)+(\hw+0.05, -\fh/2)$)}]
  \draw[red!70!black, semithick, line cap=round] (0,0) -- (0,\fh);
  \fill[red!70!black] (0,\fh) -- (0.56*\fh, 0.75*\fh) -- (0, 0.5*\fh) -- cycle;
\end{scope}
%% ── h: Adam/blue flag ─────────────────────────────────────────────────────────
\node[nd] (h) at (-1.2*\nd, \nd) {$a$};
\begin{scope}[shift={($(h)+(-\hw-0.05, -\fh/2)$)}]
  \draw[cyan!90!black, semithick, line cap=round] (0,0) -- (0,\fh);
  \fill[cyan!90!black] (0,\fh) -- (-0.56*\fh, 0.75*\fh) -- (0, 0.5*\fh) -- cycle;
\end{scope}
%% ── g: Adam/blue flag ─────────────────────────────────────────────────────────
\node[nd] (g) at (2.4*\nd, -\nd) {$j$};
\begin{scope}[shift={($(g)+(\hw+0.05, -\fh/2)$)}]
  \draw[cyan!90!black, semithick, line cap=round] (0,0) -- (0,\fh);
  \fill[cyan!90!black] (0,\fh) -- (0.56*\fh, 0.75*\fh) -- (0, 0.5*\fh) -- cycle;
\end{scope}
%% ── j: blue flag + red flag ─────────────────────────────────────────────────
\node[nd] (j) at (-1.2*\nd, -\nd) {$c$};
\begin{scope}[shift={($(j)+(-\hw-0.05, -\fh/2)$)}]
  \draw[cyan!90!black, semithick, line cap=round] (0,0) -- (0,\fh);
  \fill[cyan!90!black] (0,\fh) -- (-0.56*\fh, 0.75*\fh) -- (0, 0.5*\fh) -- cycle;
\end{scope}
\begin{scope}[shift={($(j)+(-\hw-0.05-0.56*\fh, -\fh/2)$)}]
  \draw[red!70!black, semithick, line cap=round] (0,0) -- (0,\fh);
  \fill[red!70!black] (0,\fh) -- (-0.56*\fh, 0.75*\fh) -- (0, 0.5*\fh) -- cycle;
\end{scope}
%% ── edges ───────────────────────────────────────────────────────────────────
\draw[arr] ($(d)+(0, 1)$) -- (d);
\draw[arr] (d) -- (b);
\draw[arr] (d) -- (f);
\draw[arr] (b) -- (a);
\draw[arr] (b) -- (i);
% \draw[arr] (b) -- (c);
\draw[arr] (i) -- (h);
\draw[arr] (i) -- (j);
\draw[arr] (f) -- (e);
\draw[arr] (f) -- (g);
%% ── fraction labels ─────────────────────────────────────────────────────────
\bidlabel{a}{above=9pt}{0}{1}
\bidlabel{b}{above left=9pt}{\tfrac{1}{4}}{\tfrac{1}{2}}
% \bidlabel{c}{below=9pt}{1}{1}
\bidlabel{d}{below=9pt}{\tfrac{5}{8}}{\tfrac{1}{2}}
\bidlabel{e}{above=9pt}{1}{1}
\bidlabel{f}{right=9pt}{1}{\tfrac{1}{2}}
\bidlabel{g}{below=9pt}{1}{0}
\bidlabel{h}{above=9pt}{1}{0}
\bidlabel{i}{left=9pt}{\tfrac{1}{2}}{0}
\bidlabel{j}{below=9pt}{0}{0}
\end{tikzpicture}
}
\hspace{0.5cm}
\scalebox{0.8}{
\begin{tikzpicture}

%% ── nodes ─────────────────────────────────────────────────────────────
\node[nd] (a) at (     0,  \nd) {$d$};
\node[adn] (b) at (     0,   0 ) {$e$};
% \node[nd] (c) at (     0, -\nd) {$f$};
\node[nd] (d) at (  1.2*\nd, 0  ) {$g$};
\node[nd] (e) at (  2.4*\nd, \nd) {$h$};
\node[nd] (f) at (  2.4*\nd, 0  ) {$i$};
\node[evn] (i) at ( -1.2*\nd, 0  ) {$b$};
%% ── a: red flag ─────────────────────────────────────────────────────────
\begin{scope}[shift={($(a)+(\hw+0.05, -\fh/2)$)}]
  \draw[red!70!black, semithick, line cap=round] (0,0) -- (0,\fh);
  \fill[red!70!black] (0,\fh) -- (0.56*\fh, 0.75*\fh) -- (0, 0.5*\fh) -- cycle;
\end{scope}
%% ── h: Adam/blue flag ─────────────────────────────────────────────────────────
\node[nd] (h) at (-1.2*\nd, \nd) {$a$};
\begin{scope}[shift={($(h)+(-\hw-0.05, -\fh/2)$)}]
  \draw[cyan!90!black, semithick, line cap=round] (0,0) -- (0,\fh);
  \fill[cyan!90!black] (0,\fh) -- (-0.56*\fh, 0.75*\fh) -- (0, 0.5*\fh) -- cycle;
\end{scope}
%% ── g: Adam/blue flag ─────────────────────────────────────────────────────────
\node[nd] (g) at (2.4*\nd, -\nd) {$j$};
\begin{scope}[shift={($(g)+(\hw+0.05, -\fh/2)$)}]
  \draw[cyan!90!black, semithick, line cap=round] (0,0) -- (0,\fh);
  \fill[cyan!90!black] (0,\fh) -- (0.56*\fh, 0.75*\fh) -- (0, 0.5*\fh) -- cycle;
\end{scope}
%% ── j: blue flag + red flag ─────────────────────────────────────────────────
\node[nd] (j) at (-1.2*\nd, -\nd) {$c$};
\begin{scope}[shift={($(j)+(-\hw-0.05, -\fh/2)$)}]
  \draw[cyan!90!black, semithick, line cap=round] (0,0) -- (0,\fh);
  \fill[cyan!90!black] (0,\fh) -- (-0.56*\fh, 0.75*\fh) -- (0, 0.5*\fh) -- cycle;
\end{scope}
\begin{scope}[shift={($(j)+(-\hw-0.05-0.56*\fh, -\fh/2)$)}]
  \draw[red!70!black, semithick, line cap=round] (0,0) -- (0,\fh);
  \fill[red!70!black] (0,\fh) -- (-0.56*\fh, 0.75*\fh) -- (0, 0.5*\fh) -- cycle;
\end{scope}
%% ── edges ───────────────────────────────────────────────────────────────────
\draw[arr] ($(d)+(0, 1)$) -- (d);
\draw[arr] (d) -- (b);
\draw[arr] (d) -- (f);
\draw[arr] (b) -- (a);
\draw[arr] (b) -- (i);
% \draw[arr] (b) -- (c);
\draw[arr] (i) -- (h);
\draw[arr] (i) -- (j);
\draw[arr] (f) -- (e);
\draw[arr] (f) -- (g);

%% ── fraction labels ─────────────────────────────────────────────────────────
\bidlabel{a}{above=9pt}{0}{1}
\bidlabel{b}{above left=3pt}{0}{0}
% \bidlabel{c}{below=9pt}{1}{1}
\bidlabel{d}{below=9pt}{\tfrac{1}{2}}{\tfrac{1}{4}}
\bidlabel{e}{above=9pt}{1}{1}
\bidlabel{f}{right=9pt}{1}{\tfrac{1}{2}}
\bidlabel{g}{below=9pt}{1}{0}
\bidlabel{h}{above=9pt}{1}{0}
\bidlabel{i}{left=9pt}{0}{0}
\bidlabel{j}{below=9pt}{0}{0}

\end{tikzpicture}
}

%% file: model.tex
\section{Generalized Bidding Games} \label{sec:model}

We introduce generalized bidding games as a novel model for strategic decision-making involving auctions.
The two opposing players will be referred to as \PZ and \PO, respectively.

\begin{definition}[Game graphs with bidding and control vertices]
A game graph with bidding and control vertices is of the form $\G = (\V,\VZ,\VO,\VB,\E)$, where $\V$ is the finite set of vertices, $\VZ$, $\VO$, and $\VB$ form a partition of $\V$, and $\E\subseteq \V\times \V$ is the set of edges.
The vertices $\VZ$ and $\VO$ are controlled by \PZ and \PO, respectively, and they will be collectively called the \emph{control} vertices.
On the other hand, the vertices $\VB$ are called the \emph{bidding} vertices. 
\end{definition}

We will use $\E$ interchangeably as a relation (like defined) and as a function, where $\E(v)\coloneqq \set{v'\in \V\mid (v,v')\in \E}$ for every $v\in \V$.

\myparagraph{Semantics.} 
Initially, a token is placed on a vertex $v^0$, and some non-negative monetary budgets are allocated to the players, where Eve's initial budget is $B_\eve^0$, and Adam's initial budget is $B_\adam^0=1-B_\eve^0$, so that $B_\eve^0 + B_\adam^0 = 1$.
We say $(v^t,B_\eve^t)$ is the \emph{configuration} at time $t\geq 0$, where $v^t$ and $B_\eve^t$ are the token's location and Eve's budget, respectively, and it will be guaranteed by the auction mechanism that Adam's budget at time $t$ is implicitly $B_\adam^t = 1-B_\eve^t$.

At configuration $(v^t,B_\eve^t)$, if $v^t\in \VZ$ then Eve picks $v^{t+1}\in \E(v^t)$, if $v^t\in \VO$ then Adam picks $v^{t+1}\in \E(v^t)$, and otherwise, if $v^t\in \VB$, an auction takes place, where the players use their available budgets to bid for the privilege of choosing the successor, and whoever bids higher pays the bid amount to the opponent and picks $v^{t+1}\in \E(v^t)$.

Towards formalizing the above, we introduce some notation.
A \emph{path} (in $\G$) is a vertex sequence $v^0v^1\ldots$, such that for every $t\geq 0$, $(v_t,v_{t+1})\in \E$.
Paths are either finite or infinite, and $\finpath{\G}_v$ and $\infpath{\G}_v$ respectively denote the sets of all finite and infinite paths in $\G$ starting at $v$.
For every finite path $\rho = v^0\ldots v^k$, $\last(\rho)=v^k$.
A \emph{trace} is a (finite or infinite) sequence of configurations $(v^0,B_\eve^0)(v^1,B_\eve^1)\ldots$ such that $v^0v^1\ldots$ is a path in $\G$, and moreover for every $t\geq 0$, if $v^t$ is a control vertex then the budgets remain unchanged, i.e., $B_\eve^t = B_\eve^{t+1}$.
The path \textit{induced} by the trace $(v^0,B_\eve^0)(v^1,B_\eve^1)\ldots$ is simply $v^0v^1\ldots$.

\textit{Strategies} are recipes describing the next moves of the players, defined as follows.

\begin{definition}[Strategies]
A \textit{strategy} of \PZ is a function $\polZ$ such that for every finite path $\rho$ and current budget $B_\eve\in [0,1]$ of \PZ, 
if $\last(\rho)\in \VZ$ then $\polZ\colon (\rho,B_\eve)\mapsto v\in \E(\last(\rho))$, and
if $\last(\rho)\in \VB$ then $\polZ\colon (\rho,B_\eve)\mapsto (b,v)\in [0,B_\eve]\times \E(\last(\rho))$.
% \begin{align*}
%     \polZ\colon (\rho,B_\eve)\mapsto
%     \begin{cases}
%          v\in \E(\last(\rho))   &   \last(\rho)\in \VZ,\\
%        (b,v)\in [0,B_\eve]\times \E(\last(\rho))  &   \last(\rho)\in \VB.
%     \end{cases}
% \end{align*}
Strategies of \PO are defined analogously.
Let $\PolZ$ and $\PolO$ be the sets of all strategies of \PZ and \PO, respectively.
\end{definition}
Intuitively, when $\last(\rho)\in \VZ$ (no bidding), $\polZ(\rho,B_\eve)$ prescribes the next vertex, and when $\last(\rho)\in \VB$, $\polZ(\rho,B_\eve)$ prescribes a bid value $b\in [0,B_\eve]$ and the next vertex to pick if the bidding is won.

Suppose $\polZ\in \PolZ$ and $\polO\in \PolO$ are arbitrary policies of the players.
Given an initial vertex $v^0$ and initial budget $B_\eve^0\in [0,1]$ of \PZ, $\polZ$ and $\polO$ generate a unique infinite trace as follows:
\begin{itemize}
    \item The initial configuration is $(v^0,B_\eve^0)$.
    \item At time $t\geq 0$, let the configuration be $(v^t,B_\eve^t)$. There are three possibilities for obtaining the next configuration $(v^{t+1},B_\eve^{t+1})$:
    \begin{enumerate}
        \item if $v^t\in \VZ$, and if $\polZ(v^0\ldots v^t,B_\eve^t) = v'$, then $v^{t+1} = v'$ and $B_\eve^{t+1} = B_\eve^t$;
        \item if $v^t\in \VO$, and if $\polO(v^0\ldots v^t,1-B_\eve^t) = v'$, then $v^{t+1} = v'$ and $B_\eve^{t+1}=B_\eve^t$;
        \item if $v^t\in \VB$, then an auction takes place: Let $\polZ(v^0\ldots v^t,B_\eve^t) = (b_\eve,v')$ and $\polO(v^0\ldots v^t,1-B_\eve^t) = (b_\adam,v'')$.
        If $b_\eve\geq b_\adam$, then \PZ wins the bidding, pays \PO the bid amount $b_\eve$, so that $B_\eve^{t+1}=B_\eve^t-b_\eve$, and the next vertex is the one chosen by \PZ, i.e., $v^{t+1} = v'$.
        Otherwise, if $b_\adam > b_\eve$, \PO wins the bidding, pays \PZ the bid amount $b_\adam$, so that $B_\eve^{t+1} = B_\eve^t+b_\adam$, and the next vertex is chosen by $\PO$, i.e., $v^{t+1} = v''$.
    \end{enumerate}
\end{itemize}
Note that bidding ties are resolved in favor of \PZ; later we will point out (see Remark~\ref{rem:bidding ties}) that how the ties are resolved is immaterial for the technical results that we will present in this paper.
The infinite path generated by $\polZ$ and $\polO$ and the initial budget $B_\eve^0$ is the path induced by the unique trace generated by $\polZ$ and $\polO$, and it is written as $\ppath^{\G}(\polZ,\polO,B_\eve^0) = v^0v^1\ldots$.

\myparagraph{Specifications.} A \textit{specification} is a given \textit{set} of infinite paths of $\G$. We categorized the specifications into \textit{qualitative} and \textit{quantitative} specifications. We will consider the following qualitative specifications: 
(1)~\textit{reachability} specifications contain all paths visiting a given target $T\subseteq \V$, written as: $\reach^T\coloneqq \set{v^0v^1\ldots\in \infpath{\G}_{v^0}\mid \exists t\geq 0\;.\;v^t\in T}$;
(2)~\textit{safety} specifications contain all paths staying inside a given set of safe vertices $T\subseteq \V$, written as: $\safe^T\coloneqq \set{v^0v^1\ldots\in \infpath{\G}_{v^0} \mid \forall t\geq 0\;.\; v^t\in T}$;
(3)~\textit{parity} specifications assume each vertex $v$ is assigned a color $C(v)\in \mathbb{N}$, and it contains all paths in which the maximum infinitely occurring color is even, written as: $\parity^C\coloneqq \set{v^0v^1\ldots \in \infpath{\G}_{v^0} \mid \max\set{c\in \mathbb{N}\mid \forall s\geq 0\;.\;\exists t > s\;.\;C(v^t)=c} \text{ is even}}$. 

For quantitative specifications, we assume a reward assignment $R(v)\in \mathbb{R}$ to each vertex $v$ and $r \in \mathbb{R}$.
We will consider the following quantitative specifications:
(1)~\textit{discounted sum} specifications, given a discount factor $\lambda\in (0,1)$, contain all paths whose discounted sum of rewards is at least $r$, i.e., $\discsum^{\lambda,R}_{\ge r} \coloneq \set{v^0v^1\ldots\in \infpath{\G}_{v^0} \mid \sum_{t=0}^\infty \lambda^tR(v^t) \ge r}$;
(2)~\textit{mean-payoff} specifications contain all paths whose mean-payoff is at least $r$, i.e. $\meanpayoff^{R}_{\ge r} \coloneqq \set{v^0v^1\ldots\in \infpath{\G}_{v^0} \mid \limsup_{k\to \infty} \frac{1}{k}\sum_{t=0}^{k-1} R(v^t) \ge r}$.

\myparagraph{Generalized bidding games.}
We formalize \GBGs and \SGBG; examples are in Figure~\ref{fig:non-trivial threshold in SCC}.

\begin{definition}[Generalized bidding game]
    A generalized bidding game (\GBG) is a pair $(\G,\spec)$, where $\G$ is a game graph, and $\spec$ is a specification on $\G$. 
\end{definition}

\begin{definition}[Simple generalized bidding game]
    A simple GBG (\SGBG) is a \GBG $(\G,\spec)$ where for every bidding vertex $v$ of $\G$, it holds that $|\E(v)|=2$.
\end{definition}

The traditional bidding games---which we refer to as \emph{pure} bidding games or just \BG in short---are special cases of \GBGs with only bidding vertices.
Similarly, a simple pure bidding game (\SBG) is a special case of \BG whose every (bidding) vertex has two outgoing edges.

Like \BGs, the central concept in \GBGs is the \emph{threshold budget} of Eve, which is the amount of initial budget that is necessary and sufficient for \PZ to fulfill the specification.

\begin{definition}[Threshold budgets]
    Suppose $(\G,\spec)$ is a \GBG. The threshold budget is a function $\thresh_{\G,\spec}\colon \V \to [0,1]$ such that for every vertex $v\in \V$, the following hold:
    % If $\spec$ is a qualitative specification, then the threshold budget is a constant $\thresh_{\G,\spec} \in [0,1]$ such that:
    \begin{align*}
        \forall B_\eve^0 > \thresh_{\G,\spec}(v)\;.\;
        \exists \polZ\in \PolZ\;.\;
        \forall \polO\in \PolO\;.\;
        \ppath_v^\G(\polZ,\polO,B_\eve^0) \in \spec,\\
        \forall B_\eve^0 < \thresh_{\G,\spec}(v)\;.\;
        \exists \polO\in \PolO\;.\;
        \forall \polZ\in \PolZ\;.\;
        \ppath_v^\G(\polZ,\polO,B_\eve^0) \notin \spec.
    \end{align*}
    % When $\spec$ is a quantitative specification, for each $r \in \mathbb{R}$, the budget threshold is a constant $\thresh_{\G,\spec}^{\geq r}$ which is the necessary and sufficient budget for achieving a reward of at least $r$:
    % \begin{align*}
    %     \forall r\in \mathbb{R}\;.\;
    %     \forall B_\eve^0 > \thresh_{\G,\spec}^{\geq r}\;.\;
    %     \exists \polZ\in \PolZ\;.\;
    %     \forall \polO\in \PolO\;.\;
    %     \spec(\ppath^{\G}(\polZ,\polO,B_\eve^0)) \geq r,\\
    %     \forall r\in \mathbb{R}\;.\;
    %     \forall B_\eve^0 < \thresh_{\G,\spec}^{\geq r}\;.\;
    %     \exists \polO\in \PolO\;.\;
    %     \forall \polZ\in \PolZ\;.\;
    %     \spec(\ppath^{\G}(\polZ,\polO,B_\eve^0)) < r.
    % \end{align*}
    % We drop the subscript of $\thresh_{\G,\spec}$ whenever both $\G$ and $\spec$ are clear from the context.
\end{definition}
Intuitively, Eve wins from $v$ if her initial budget $B_\eve^0$ is above her threshold budget $\thresh_{G,\varphi}(v)$ at $v$, and loses if her budget is lower than $\thresh_{G,\varphi}(v)$.
When $B_\eve^0 = \thresh_{G,\varphi}(v)$, the outcome of the game (i.e., who has a winning strategy) depends on the tie-breaking mechanism; we will point it out in Remark~\ref{rem:bidding ties}.

We first establish the following connection between \GBGs and \SGBGs:
\begin{restatable}[Equivalence of \GBGs and \SGBGs, Proof in Appendix~\ref{app:sec:GBG-to-SGBG}]{proposition}{propEquivalenceGBGSGBG} \label{prop:equivalence-GBG-SGBG}
    Computing thresholds for \GBGs reduces in linear time to computing thresholds for \SGBGs. %$(\game,\spec)$ is linearly reducible to computing $\thresh_{\game',\spec'}$ where $(\game',\spec')$ is an \SGBG. 
\end{restatable}
Since \SGBGs are special cases of \GBGs, the other direction trivially holds, implying that \SGBGs and \GBGs are equivalent to each other modulo a linear blow-up.

\myparagraph{Decision problems.} For each specification type $M \in \{\reach,\allowbreak \safe,\allowbreak \parity,\allowbreak \discsum, \meanpayoff\}$, the question of whether the threshold is below $0.5$ is the main problem we study, and it is equivalent to the membership of the set $\THRESH_M^{\leq 0.5} \coloneqq \set{(\game,\spec,v) \mid (\game,\spec) \text{ is a \GBG and } v \text{ is a vertex of } \game \land \spec \text{ is of type } M \wedge \thresh_{\game,\spec}(v) \leq 0.5}$.
The two fundamental questions related to threshold budgets are:
\begin{enumerate}
    \item \textit{Does threshold budget exist for every generalized bidding game with reachability, safety, parity, discounted sum, and mean-payoff specifications?} We show that the answer is ``yes.'' 
    \item \textit{What is the complexity of deciding membership in $\THRESH_M^{\leq 0.5}$, for various specification types $M$?} 
    We will answer this question by establishing bridges between \GBGs and simple stochastic games, a class of extensively studied games which we summarize next.
\end{enumerate}

% \begin{align*}
% &\THRESH_M^{\leq 0.5} \coloneqq \\
% &\hspace{-0.5cm}\begin{cases}
%     \set{(\game,\spec) \mid (\game,\spec) \text{ is a \GBG } \land \spec \text{ of type } M \wedge \thresh_{\game,\spec} \leq 0.5}    &   M \in \set{\reach,\safe,\parity}\\
%     \set{(\game,\spec,r) \mid (\game,\spec) \text{ is a \GBG } \land \spec \text{ of type } M \wedge \thresh_{\game,\spec}^{\geq r} \leq 0.5}    &   M \in \set{\discsum,\meanpayoff}
% \end{cases}
% \end{align*}
% \KM{I redefined $\thresh_{(\game,\reach^T)}(r) $ to $\thresh_{(\game,\reach^T)}^{\geq r} $ in the quantitative case to parallel with the value function. And I changed $\thresh_{(\game,\reach^T)}(r) $ to $\thresh_{\game,\reach^T}(r) $ (without parenthesis in the subscript), like defined earlier.}

%% file: SSG_short.tex
\section{Bridging Generalized Bidding Games and Simple Stochastic Games}
\label{sec:ssg}

\subsection{Preliminaries on Simple Stochastic Games}
We recall some basic concepts on simple stochastic games (SSGs) from the literature~\cite{Condon92,ChatterjeeH12}.
\begin{definition}[Simple stochastic games]
    A simple stochastic game (SSG) is a pair $(\S,\spec)$, where $\S = (\V, \VZ, \VO, \VR, \E)$ is the game graph, in which $\V$ is a finite set of vertices, $\VZ, \VO$, and $\VR$ form a partition of $\V$, with $\VZ$ and $\VO$ being the control vertices owned respectively by Eve and Adam and $\VR$ being the random vertices, $\E \subseteq \V \times \V$ is the set of edges where each random vertex $r \in \VR$ has exactly two successors,\footnote{The more standard definition of SSGs from the literature imposes the two out-degree restriction on \textit{all} vertices, not just the random ones. Our formalism allows a more succinct representation.} i.e., $|\E(r)|=2$, and  $\spec$ is a specification.
    % An SSG is a random-turn game (RTG) if the random vertices alternate with the control vertices, i.e., $\E\subseteq \VR\times (\VZ\cup \VO) \cup (\VZ\cup \VO)\times\VR$, and every random vertex $v\in \VR$ has exactly two successors $v_\eve\in \VZ$ and $v_\adam\in \VO$, such that $\E(v_\eve)=\E(v_\adam)$ (i.e., $v$ assigns the turn randomly, but afterwards the players have the same options).
\end{definition}

SSGs have similar gameplay as \GBGs for the control vertices---Eve chooses successors from $\VZ$ and Adam chooses successors from $\VO$---whereas at the random vertices in $\VR$, each of the two successors is chosen with probability $0.5$.

\emph{Random-turn games} (RTG) are played between Eve and Adam on finite graphs. In each turn, who moves the token is determined by a (fair) coin toss, i.e., each player gets to move the token with probability $0.5$. RTGs can be modeled by SSGs where random vertices alternate with the control vertices, i.e., $\E\subseteq \left(\VR\times (\VZ\cup \VO)\right) \cup \left((\VZ\cup \VO)\times\VR\right)$, and every random vertex $v\in \VR$ has exactly two successors $v_\eve\in \VZ$ and $v_\adam\in \VO$, such that $\E(v_\eve)=\E(v_\adam)$ (i.e., $v$ assigns the turn randomly, but afterwards the players have the same options).

\textit{Paths} in SSGs are defined as in \GBGs, and a \textit{strategy} $\sigma_\eve$ of Eve maps finite paths only ending at Eve's vertices to their successors (i.e., strategies involve no bidding); strategies of Adam are defined analogously, and the sets of their strategies are $\Sigma_\eve$ and $\Sigma_\adam$, respectively.
Every pair of strategies $\sigma_\eve\in \Sigma_\eve$ and $\sigma_\adam\in \Sigma_\adam$ in an SSG $\S$ resolves all non-determinism and induces a Markov chain $\S^{\sigma_\eve,\sigma_\adam}$ (with possibly infinite state space).
For a given measurable set of paths $\Delta$, we write $P^{\sigma_\eve,\sigma_\adam}(\Delta\mid v)\in [0,1]$ to denote the probability that a randomly sampled path of $\S^{\sigma_\eve,\sigma_\adam}$ starting at $v$ belongs to $\Delta$.
With this, for a given specification $\varphi$, we define the optimal \textit{value} of every vertex $v\in \V$ as 
$ \val_{\S,\varphi}(v)\coloneqq
    \sup_{\sigma_\eve \in \Sigma_\eve} \inf_{\sigma_\adam \in \Sigma_\adam} P^{\sigma_\eve, \sigma_\adam}(\varphi\mid v)$,
    %when $\varphi\in \set{\reach,\safe,\parity}$, and
%$\val^{\geq r}_{\S,\varphi}(v)                                                                                       \coloneqq    \sup_{\sigma_\eve \in \Sigma_\eve} \inf_{\sigma_\adam \in \Sigma_\adam} P^{\sigma_\eve, \sigma_\adam}(\varphi \geq r \mid v)$ when $\varphi\in \set{\discsum,\meanpayoff}$,
% \begin{flalign*}
%     \val_{\S,\varphi}(v)                                                                                                 & \coloneqq
%     \sup_{\sigma_\eve \in \Sigma_\eve} \inf_{\sigma_\adam \in \Sigma_\adam} P^{\sigma_\eve, \sigma_\adam}(\varphi\mid v) & \varphi\in \set{\reach,\safe,\parity},                                                                                                                                             \\
%     \val^{\geq r}_{\S,\varphi}(v)                                                                                        & \coloneqq    \sup_{\sigma_\eve \in \Sigma_\eve} \inf_{\sigma_\adam \in \Sigma_\adam} P^{\sigma_\eve, \sigma_\adam}(\varphi \geq r \mid v) & \varphi\in \set{\discsum,\meanpayoff},
% \end{flalign*}
where the measurability for different types of $\varphi$ that we consider is established~\cite{DBLP:journals/tcs/Chatterjee07,Puterman94,modelcheckingbook}.
% \KM{I removed the ``$*$'' from the qualitative value, it did not match with the threshold notation for bidding games. Also, I added the stochastic game $\S$ in the suffix to match threshold notation.}

Deciding whether the value is above $ 0.5$ is a fundamental problem for SSGs, and, for a given specification type $M \in \{\reach,\safe,\parity,\discsum,\meanpayoff\}$, it is equivalent to the membership of the set 
$\VALUE_M^{\geq 0.5} \coloneqq \set{(\S,\varphi,v) \mid (\S,\varphi) \text{ is an SSG and } v \text{ is a vertex of } \S \land \varphi \text{ is of type } M \wedge \val_{\S,\varphi}(v) \geq 0.5}$.
% \begin{align*}
%      & \VALUE_M^{\geq 0.5} \coloneqq                                                                                                                                                                      \\
%      & \quad\begin{cases}
%                 \set{(\S,\varphi) \mid (\S,\varphi) \text{ is an SSG } \land \varphi \text{ of type } M \wedge \val_{\S,\varphi} \geq 0.5}            & M \in \set{\reach,\safe,\parity}  \\
%                 \set{(\S,\varphi,r) \mid (\S,\varphi) \text{ is an SSG } \land \varphi \text{ of type } M \wedge \val_{\S,\varphi}^{\geq r} \geq 0.5} & M \in \set{\discsum,\meanpayoff}.
%             \end{cases}
% \end{align*}

Central to the value estimation for SSGs is the \textit{Bellman operator} $\bellmanSSG^S \colon [0,1]^\V \to [0,1]^\V$, which maps a given value assignment $X\in [0,1]^\V$ to an updated assignment as follows:
\[
    \bellmanSSG^S(X)(v) = \begin{cases}
        X(v)                                                        & v \in S                \\
        \max_{u \in \E(v)} X(u)                                     & v \in \VZ \setminus S  \\
        \min_{u \in \E(v)} X(u)                                     & v \in \VO \setminus S  \\
        \frac{\min_{u \in \E(v)} X(u) + \max_{u \in \E(v)} X(u)}{2} & v \in \VR \setminus S.
    \end{cases}
\]

\subsection{Exchanging Bidding Vertices with Random Vertices}
\label{sec:structural equivalence}

% It is a classical result in the literature of bidding games that the threshold $\thresh_{\game,\spec}$ of a pure BG $\game$ and a reachability/safety specification $\spec$ is equal to $1-\val_{RT(\game),\spec}$ where $RT(\game)$ is random-turn game resulted by replacing each (bidding) vertex of $\game$ with a random-turn vertex~\cite{avni2020survey}. 

Given a BG $(\game,\varphi)$, its corresponding random-turn game $(RT(\game),\varphi)$ is obtained by reinterpreting each vertex of $\game$ as the vertex in a  random-turn game. The following theorem is a well-known result in the literature that relates the thresholds of BGs to the values of RTGs.

\begin{theorem}[Structural equivalence between BGs and RTGs~\cite{avni2020survey}] \label{thm:pure-bgs-rtgs}
    Let $(\game,\spec)$ be a BG where $\spec$ is a $\reach, \safe, \parity,$ or $\meanpayoff$ specification. Then, the threshold $\thresh_{\game,\spec}$ exists, and $(\game,\spec)$ is structurally equivalent to the RTG $(RT(\G),\spec)$, i.e., $\thresh_{\game,\spec}(v) = 1 - \val_{RT(\game),\spec}(v)$. 
\end{theorem}
We extend this connection in the context of \SGBGs and SSGs.
Given an \SGBG $(\G,\varphi)$, let $(\SSG(\G),\varphi)$ be the identical SSG that reinterprets the bidding vertices of $\G$ as random vertices. Conversely, for an SSG $(\S,\varphi)$, let $(\SSG^{-1}(\S),\varphi)$ be the identical \SGBG obtained by reinterpreting the random vertices of $\S$ as bidding vertices. It is easy to see that, for every \SGBG $(\game,\varphi)$, $\SSG^{-1}(\SSG(\game))=\game$.
%Given a generalized bidding game $\game=(\V,\VZ,\VO,\VB,\vin,\E)$, we define its corresponding SSG $\SSG(\game)$ as an SSG with the exact same components as $\game$ where vertices in $\VB$ are treated as random vertices of the SSG. Similarly, given an SSG $\game=(\V,\VZ,\VO,\VR,\vin,\E)$, we use $\SSG^{-1}(\game)$ for the game graph of a generalized bidding game where the random vertices of $\game$ are treated as bidding vertices. Clearly, $\SSG^{-1}(\SSG(\game))=\game$ for every generalized bidding game $\game$. 

Our main result is Theorem~\ref{thm:equivalence between SGBG and SSG}, proved in Section~\ref{sec:qual} and Section~\ref{sec:quant}. It is a generalization of Theorem~\ref{thm:pure-bgs-rtgs} to \SGBGs with parity, discounted-sum, and mean-payoff specifications.
\begin{theorem}[Structural equivalence between \SGBGs and SSGs]\label{thm:equivalence between SGBG and SSG}
Let $(\game, \spec)$ be an \SGBG where $\spec$ is a $\reach, \safe, \parity,\discsum$, or $\meanpayoff$ specification. Then, the threshold $\thresh_{\game,\spec}$ exists, and $(\game, \spec)$ is structurally equivalent to the SSG $(\SSG(\game), \spec)$, i.e., $\thresh_{\game, \spec}(v) = 1 - \val_{\SSG(\game), \spec}(v)$.% if $\spec$ is a reachability, safety, parity, discounted-sum, or mean payoff specification.
\end{theorem}
Proposition~\ref{prop:equivalence-GBG-SGBG} then gives rise to a structural equivalence between \GBGs and SSGs, modulo a linear blow-up.
This result yields immediately a computational equivalence between the problems $\THRESH_\spec^{\leq 0.5}$ and $\VALUE_\spec^{\geq 0.5}$, where $\spec$ is among the listed types.

% As a result of Proposition~\ref{prop:equivalence-GBG-SGBG}, computing the thresholds in any \BG $(\G,\spec)$, is linearly reducible to computing thresholds of a \SGBG $(\G',\spec)$. It is a classical result in the literature of bidding games~\cite{avni2020survey}, that for every reachability, safety, parity specification $\spec$, it holds that $\thresh_{\game,\spec} = 1- \val_{\SSG(\G'),\spec'}$. Similarly, for mean-payoff $\spec$, $\thresh_{\game,\spec}^{\geq r} = 1- \val_{\SSG(G'),\spec'}^{\geq r}$. Section~\ref{sec:qual} and Section~\ref{sec:quant} extend these results to \GBGs with parity, discounted-sum, and mean-payoff specifications.

%A classical result shows that for pure bidding games with reachability/safety $\spec$, $\SSG(\game)$ is a random-turn game and $\thresh_{\game,\spec} = 1-\val_{\SSG(\game),\spec}$~\cite{avni2020survey}. Furthermore, for mean-payoff $\spec$ on strongly connected pure games, there is an $r>0$ such that $\thresh_{\game,\spec}(r')=1$ for $r'<r$ and $0$ for $r' \geq r$~\cite{avni2020survey}; analogously in random-turn games at the same $r$, $\val_{\SSG(\game),\spec}(r')=0$ for $r' \leq r$ and $1$ for $r' \geq r$. Section~\ref{sec:qual} and Section~\ref{sec:quant} extend these results to generalized bidding games with parity, discounted-sum, and mean-payoff specifications. 

%% file: qualitative.tex
\section{Qualitative Specifications} \label{sec:qual}

% \EKG{define decision problems as set membership. add a theorem that says for each specification $\phi$, the thresholds exist and the decision problems are reducible}

% We show the existence of thresholds in \GBGs with reachability, safety, and parity specifications. This immediately provides a positive answer to our first fundamental question in the qualitative setting. Furthermore, we show that in both settings, the threshold $\thresh$ in the \GBG is equal to $1-\optval(\vin)$ where $\optval(\vin)$ is the value of the corresponding simple stochastic game. As a corollary of both results above, we obtain complexity bounds for deciding whether the threshold budget is above 0.5. 

In this section, we study the existence and computation of threshold budgets in \GBGs with reachability, safety, and parity specifications. In particular, for all these specifications, we show that threshold budgets exist, and moreover, we establish a computational procedure by showing the structural equivalence between \GBGs and SSGs, as claimed in Section~\ref{sec:structural equivalence}.

Although parity specifications subsume both reachability and safety specifications, we first state and prove our results for reachability and safety (Section~\ref{sec:qual:reach-safe}) because they are more intuitive. Building on top of this, we will extend our results to the more general case of parity specifications in Section~\ref{sec:qual:parity}.

\subsection{Reachability and Safety Specifications} \label{sec:qual:reach-safe}

Suppose $\G = (\V,\VZ,\VO,\VB,\E)$ is a game graph and $\reach^T$ is a reachability specification on $\G$, i.e. \PZ's goal is to reach $T$ and \PO's goal is to prevent $T$ from being reached. We define $\rTh\colon \V \to [0,1]$ where if \PZ's budget exceeds $\rTh(v)$ when the game is in vertex $v$, she has a winning strategy. We then take the point-of-view of \PO and define $\sTh\colon \V \to [0,1]$ where if \PO's budget exceeds $\sTh(v)$, he can prevent the reachability of $T$. We will then show that $\rTh(v)=1-\sTh(v)$, which implies $\thresh_{\G,\reach^T} = \rTh(v)= 1-\thresh_{\G,\safe^{\V\setminus T}}=1-\sTh(v)$.

%In particular, this shows that $\thresh_{\game,\reach^T} = \rTh(\vin)$.

We first introduce the bidding operator $\bellman^S\colon \R^\V \to \R^\V$ as follows:
\[
  \bellman^S(X)(v) \coloneqq \begin{cases}
    X(v)                                                       & v \in S               \\
    \min_{u \in \E(v)} X(u)                                    & v \in \VZ \setminus S \\
    \max_{u \in \E(v)} X(u)                                    & v \in \VO \setminus S \\
    \frac{\min_{u \in \E(v)} X(u)+ \max_{u \in \E(v)} X(u)}{2} & v \in \VB \setminus S
  \end{cases}
\]
Lemma~\ref{lem:bellman:property} in Appendix~\ref{app:qual:proofs} shows three useful properties of the bidding operator: (i)~$\bellman^S$ is monotone in $X$, (ii)~if for $v \in \V \setminus S$, \PZ's budget is above $\bellman^S(X)(v)$ when the game is in $v$, then she can play so that her budget is above $X(u)$ if $u$ is reached next, and (iii)~$\bellman^S(\mathbf{1} - X) = \mathbf{1} - \bellmanSSG^S(X)$ where $\bellmanSSG$ is the bellman operator of $\SSG(\G)$.

Next, we illustrate the finite-horizon reachability thresholds $\rTh^i \colon \V \to [0,1]$, where if \PZ's budget exceeds $\rTh^i(v)$ when the token is on $v$, she has a strategy to reach $T$ in at most $i$ steps.
The threshold $\rTh^i$ is defined recursively, with
$ \rTh^0(v) \coloneqq 0$ if $v\in T$,
$ \rTh^0(v) \coloneqq 1$ if $v\notin T$, and
$\rTh^i(v) \coloneqq \bellman^T(\rTh^{i-1})(v)$ for every $i>0$.
% \begin{align*}
%     \rTh^0(v) & = \begin{cases}
%                       0 & v \in T    \\
%                       1 & v \notin T
%                   \end{cases}           &
%     \rTh^i(v) & = \bellman^T(\rTh^{i-1})(v)
% \end{align*}
% where $\bellman$ is the bidding operator defined as follows: 

% \[
% \rTh^i(v) = \begin{cases}
%     0 & v \in T \\ 
%     1 & v \notin T \wedge i =0 \\ 
%     \argmin_{u \in E(v)} \rTh^{i-1}(u) & v \notin T \wedge i>0 \wedge v \in \VZ \\
%     \argmax_{u \in E(v)} \rTh^{i-1}(u) & v \notin T \wedge i>0 \wedge v \in \VO \\
%     \frac{\argmin_{u \in E(v)} \rTh^{i-1}(u)+\argmax_{u \in E(v)} \rTh^{i-1}(u)}{2} & v \notin T \wedge i>0 \wedge v \in \VB
% \end{cases}
% \]

Intuitively, \PZ can win in 0 steps, if and only if the game is already in $T$. Moreover, for $i>0$, she can guarantee to win the game in $i$ steps, if she can guarantee to have enough budget in the next step to win in $i-1$ steps. We then define $\rTh(v) \coloneqq \lim_{i \to \infty} \rTh^i(v)$ and show that it is well-defined. We also show how \PZ can win when her budget is above $\rTh(v)$.
%and equal to the reachability threshold of the \GBG, which immediately provides a positive answer to our first fundamental question regarding the existence of reachability thresholds. This is formalized in the following theorem.
\begin{restatable}[Existence of $\reach$ Thresholds]{theorem}{thmReachWellStrat} \label{thm:reach:well-strat}
  For any \GBG $(\G,\reach^T)$ it holds that:
  \begin{itemize}
    \item The limit $\rTh(v) = \lim_{i \to \infty} \rTh^i(v)$ is well-defined.
    \item If \PZ's budget exceeds $\rTh^i(v)$ when the token is on $v$, she has a strategy to force reachability of $T$ in at most $i$ steps.
    \item If \PZ's budget exceeds $\rTh(v)$ when the token is on $v$, she has a strategy to win $\reach^T$.
  \end{itemize}
\end{restatable}
\begin{proof}
  We provide a proof sketch for each statement. The full proofs are in Appendix~\ref{app:qual:proofs}.

  \textbf{Well-Definedness of $\rTh$:} The first statement follows by the monotonicity of the operator $\bellman^T$
  %(Lemma~\ref{lem:bellman:property} in Appendix~\ref{app:qual:proofs}) 
  and the fact that the sequence is bounded in $[0,1]$.
  \\
  The second statement is proved by induction on $i$. The full proof is available in Appendix~\ref{app:qual:proofs}.
  \\
  \textbf{Reachability Winning Strategy:} Suppose the game starts at $v$ and \PZ's budget $B_\eve$ strictly exceeds $\rTh(v)$. Since $\rTh^i(v)$ is a non-increasing sequence converging to $\rTh(v)$, there exists some finite $k \geq 0$ where $B_\eve > \rTh^k(v)$. Due to the second statement, \PZ has a finite-horizon strategy to reach $T$ in at most $k$ steps. Therefore, following this strategy, \PZ ensures reachability of $T$ and wins the qualitative specification $\reach^T$.
\end{proof}

Similarly, taking the point-of-view of \PO, we can define the finite-horizon safety thresholds $\sTh^i\colon V \to [0,1]$ as the amount of budget he needs in order to prevent $T$ from being reached for at least $i$ steps (i.e. $\safe^{\V \setminus T}$ for $i$ steps). Formally, $\sTh^i$ is defined recursively, with
$\sTh^0(v) = 1$ if $v\in T$,
$\sTh^0(v) = 0$ if $v\notin T$, and
$\sTh^i(v) = \bellman^T(\sTh^{i-1})(v)$ for every $i>0$.
The next theorem shows the correctness of the $\sTh^i$ as the threshold budget for finite-horizon safety. Moreover, it shows that the limit $\sTh(v) \coloneqq \lim_{i \to \infty} \sTh^i(v)$ exists and \PO can win when his budget exceeds $\sTh(v)$.
% \begin{align*}
%     \sTh^0(v) & = \begin{cases}
%                       1 & v \in T    \\
%                       0 & v \notin T
%                   \end{cases}           &
%     \sTh^i(v) & = \bellman^T(\sTh^{i-1})(v)
% \end{align*}

\begin{restatable}[Existence of $\safe$ Thresholds]{theorem}{thmSafeWellStrat} \label{thm:safe:well-strat}
  The following hold for any \GBG $(\G,\reach^T)$:
  \begin{itemize}
    \item If \PO's budget exceeds $\sTh^i(v)$ when the token is on $v$, he can prevent $T$ from being reached for at least $i$ steps.
    \item The limit $\sTh(v) = \lim_{i \to \infty} \sTh^i(v)$ is well-defined.
    \item If \PO's budget exceeds $\sTh(v)$ when the token is on $v$, he can prevent reachability of $T$.
  \end{itemize}
\end{restatable}
\begin{proof}
  The first two statements are shown in Appendix~\ref{app:qual:proofs}. We prove the third statement:\\
  \textbf{Safety Winning Strategy:} Suppose \PO's initial budget satisfies $B_\adam > \sTh(v)$. Since $\sTh^i(v) = \bellman^{T}(\sTh)$, it follows by Lemma~\ref{lem:bellman:property} that he has a strategy to ensure his budget exceeds $\sTh(u)$ if $u$ is reached in the next step. Thus, he can ensure that his budget is always above $\sTh(v)$, for any $v$ that is reached during the gameplay. Given that $\sTh(v)=1$ for $v \in T$, and Adam's budget cannot exceed 1, it follows that $T$ is never reached.
\end{proof}

% By definition, it can be proven inductively that $\rTh^i(v) = 1- \sTh^i(v)$ for all $i$. Thus, the reachability and safety thresholds $\rTh$ and $\sTh$ are indeed tight:
% \begin{restatable}[Determinacy]{corollary}{corRthEqSth} \label{cor:rth-eq-sth}
%     When the token is on $v$, if \PZ's budget exceeds $\rTh(v)$, she has a winning strategy. If her budget is less than $\rTh(v)$, then \PO has a winning strategy.
% \end{restatable}

\begin{remark}\label{rem:bidding ties}
  Like pure bidding games~\cite{lazarus1999combinatorial,AvniGHM24}, only when Eve's budget coincides with the threshold budget $\rTh$, the tie-breaking rule matters. Assuming ties are always resolved in favor of \PZ, she has a winning strategy iff her initial budget is at least $\rTh$. If ties were resolved in favor of \PO, then \PZ would require more budget than $\rTh$ to win.
\end{remark}
By induction, we can show that $\rTh^i(v) = 1- \sTh^i(v)$ for all $i$ and $v \in \V$, which leads to: %Thus, $\rTh(v)=1-\sTh(v)$, which shows the following proposition: 
\begin{proposition}
  For any $T \subseteq \V$ we have $\thresh_{\game,\reach^T} = \rTh = 1-\sTh =  1-\thresh_{\game,\safe^{\V \setminus T}} $.
\end{proposition}
Thus, computing threshold budgets for reachability and safety in \GBGs are equivalent problems. Moreover, considering the complete lattice $([0,1]^V, \leq)$ with the pointwise definition of $\leq$, the above definition of $\rTh$ shows that $\rTh$ is the greatest fixed-point $X$ of the operator $\bellman^T(X)$ where $X_{|T}=\mathbf{0}$ is fixed. It is also a classical result for SSGs~\cite{AlfaroM04} that the optimal value function $\val$ in an SSG with reachability specification is the least fixed-point of the operator $\bellmanSSG^T(X)$ where $X_{|T}=\mathbf{1}$ is fixed. This implies the following:

%Therefore, we focus our subsequent analysis exclusively on reachability thresholds, keeping in mind that the derived results carry over to safety specifications.

%To bridge this with classical results, we recall the counterpart SSG Bellman operator $\bellmanSSG^T$ defined earlier in Section~\ref{sec:ssg}.
% Indeed, 
\begin{restatable}[Equivalence to SSGs]{theorem}{thmReachSSG} \label{thm:reach:ssg}
  Let $(\G,\reach^T)$ be an \SGBG. Then for all $v \in \V$, it holds that $\rTh(v) = 1-\val_{\SSG(\game),\reach^T}(v)$. Similarly, $\sTh(v) = 1- \val_{\SSG(\game),\safe^{\V \setminus T}}$.
\end{restatable}
\begin{proof}
  The dual relationship $\bellman^T(\mathbf{1} - X) = \mathbf{1} - \bellmanSSG^T(X)$ holds (See Lemma~\ref{lem:bellman:property} in Appendix). Consequently, by applying the Knaster-Tarski fixed-point theorem, the bijection $X \mapsto \mathbf{1} - X$ maps the least fixed-point of $\bellmanSSG^T$ to the greatest fixed-point of $\bellman^T$.

  It is established~\cite{AlfaroM04} that the maximal reachability value function $\val_{\SSG(\game),\reach^T}$ acts as the \emph{least} fixed-point of the operator $\bellmanSSG^T(X)$ subject to the target condition $X_{|T} = \mathbf{1}$. Applying the complementing map, the value vector $\mathbf{1} - \val_{\SSG(\game),\reach^T}$ calculates the corresponding \emph{greatest} fixed-point of $\bellman^T(X)$ under the inverted target condition $X_{|T} = \mathbf{0}$. Thus, $\rTh(v) = 1- \val_{\SSG(\game),\reach^T}(v)$ for all $v \in \V$. The last is due to $\rTh = 1-\sTh$ for $\reach^T$.
\end{proof}

% \begin{restatable}{corollary}{corReachEquivalence} \label{cor:reach:equivalence}
%     The membership problem of $\THRESH_{\reach}^{\leq 0.5}$ and the membership problem of $\VALUE_{\reach}^{\geq 0.5}$ are linearly reducible to each other.
% \end{restatable}

% On the algorithmic side, Theorem~\ref{thm:reach:ssg} implies that the decision problem of whether $\rTh(\vin) \geq 0.5$ in $(\game,\reach^T)$ is equivalent to the problem of deciding whether $\optval_{\reach^T}(\vin) \leq 0.5$ in $\SSG(\game)$, which is in $\mathit{NP} \cap \mathit{coNP}$~\cite{Condon92}. 

Proposition~\ref{prop:equivalence-GBG-SGBG} then implies that the reachability threshold computation in \GBGs is linearly equivalent to that of \SGBGs, which is structurally equivalent to computing values in SSGs. On the algorithmic side, it is known that the membership problem of $\VALUE_{\reach}^{\geq 0.5}$ in SSGs is in $\mathit{NP} \cap \mathit{coNP}$~\cite{Condon92}. Theorem~\ref{thm:reach:ssg} implies that membership in $\THRESH_\reach^{\leq 0.5}$ is decidable in $\mathit{NP} \cap \mathit{coNP}$ as well. Moreover, if any or both of $\VZ$ or $\VO$ are empty in an SSG, the SSG becomes an MDP. %(an SSG where one player has no choice). 
Computing the value function in MDPs is possible in polynomial time using linear programming~\cite{modelcheckingbook}. Hence, if one of $\VZ$ or $\VO$ is empty, $\THRESH_{\reach^T}^{\leq 0.5}$ is decidable in polynomial time. The following proposition summarises our complexity results:
\begin{restatable}[Threshold verification complexities for $\reach$ and $\safe$]{corollary}{propReachComplexity} \label{prop:reach:complexity}
  The following hold:
  \begin{itemize}
    \item The membership problems of $\THRESH_{\reach}^{\leq 0.5}$ and $\THRESH_{\safe}^{\leq 0.5}$ are in $\mathit{NP} \cap \mathit{coNP}$.
    \item If $\VZ=\emptyset$ or $\VO=\emptyset$, then both $\THRESH_\reach^{\leq 0.5}$ and $\THRESH_{\safe}^{\leq 0.5}$ membership are in $P$.
  \end{itemize}
\end{restatable}

% \KM{This should be a theorem. Moreover, where do you state that \GBGs are at least as hard as simple stochastic games? }

\subsection{Parity Specifications} \label{sec:qual:parity}

% While reachability, safety, B\"uchi, and coB\"uchi specifications have been studied in previous models of bidding games such as classical bidding games~\cite{AvniH22} and bidding games with charge~\cite{AvniGHM24}, parity specifications have never been studied in the context of bidding games. 
% \KM{This is not correct, parity has been studied for discrete bidding game, see the TheoretiCS paper of Guy and Suman. Actually, we should discuss what the differences are between our and their algorithms.}

In this section, we show the existence of parity thresholds in \GBGs by constructing a series of recursive sequences whose limit converges to the parity threshold. This will later be used to show how the parity thresholds in \GBGs relate to parity values in stochastic games.

Let $(\game,\parity^C)$ be a \GBG where $C \colon \V \to \{0,\dots, d\}$ is a function that assigns a color $C(v)$ to each vertex $v$. For $k \leq d$ let $T_k = \{v \in \V | C(v) = k\}$ be the set of vertices with color $k$.
%\KM{We should not call it ``generalized parity bidding game.'' It could be read as (classical) bidding game + generalized parity specification (conjunction of multiple parity). The right way to say it would be ``\GBG with parity specification.'' We could discuss if this is used too often, in which case we should shorten the name.}
Moreover, for $\trianglelefteq \in \{<,\leq,>,\geq, \neq\}$, let $T_{\trianglelefteq k} = \{v \in \V | C(v) \trianglelefteq k\}$.

% For an infinite sequence $\pi = [v_0,v_1,\dots] \in \V^\omega$, define 
% \begin{align*}\inf(\pi) = \{v \in \V | v \text{ appears infinitely many times in }\pi\}.\end{align*} 
% An infinite sequence $\pi$ is winning for \PZ if $\max(C(\inf(\pi)))$ is even. 
% \KM{Parity was already defined in Section 2.}

We show the existence of the parity threshold function $\pTh\colon\V \to [0,1]$ where $\pTh(v)$ is the amount of budget \PZ needs at $v$ to win the game.
Let $Z_k \in [0,1]^{T_k}$ be a vector of {\em candidate} thresholds for vertices in $T_k$. We first define the {\em{frugal}} parity threshold function $\fpTh_k(Z_{k+1},Z_{k+2}, \dots, Z_d) \colon T_{\leq k} \to [0,1]$ where for $v \in T_{\leq k}$, if \PZ's budget exceeds $\fpTh_k(Z_{k+1},Z_{k+2}, \dots, Z_d)(v)$ when the token is on $v$, she has a strategy to ensure either:
\begin{itemize}
  \item \textbf{Escape:} The game reaches a vertex $u \in T_{>k}$ while her budget exceeds $Z_{C(u)}$, or,
  \item \textbf{Internal Win:} The game stays in $T_{\leq k}$ and the greatest color visited infinitely is even.
\end{itemize}
This way, $\fpTh_d(\cdot)(v)$ will be equal to the actual parity threshold $\pTh(v)$ since the escape condition is impossible when $k=d$.
The recursive construction of $\fpTh_k$ is in Appendix~\ref{app:qual:parity}.
%The following lemma shows the well-definedness and correctness of $\fpTh$:
\begin{restatable}[Well-definedness and correctness of $\fpTh$]{lemma}{lemFpthCorrectness} \label{lem:fpth:correctness}
  Suppose $k \in \{0,\dots,d\}$ and $Z_{k+1}, \dots Z_d$ are candidate thresholds. The following statements hold:
  \begin{itemize}
    \item $\fpTh_k(Z_{k+1}, \dots, Z_d)(v)$ is well-defined for all $v \in \V$.
    \item If \PZ has budget more than $\fpTh_{k}(Z_{k+1},\dots, Z_d)(v)$ when the token is on $v$, she has a strategy to ensure either (i) the game stays in $T_{\leq k}$ and she wins, or (ii) the game reaches a vertex $u \in T_{>k}$ while her budget exceeds $Z_{C(u)}(u)$.
  \end{itemize}
\end{restatable}

Lemma~\ref{lem:fpth:correctness} uses the point-of-view of \PZ and shows the existence of winning strategies for her whenever her budget is above $\fpTh_k$. One could add one unit to the color of every vertex, i.e. have $C'(v) = C(v)+1$, to obtain a new coloring $C'$. This switches the role of players and by the definition of $\fpTh$, it follows that $\fpTh^C_k = 1-\fpTh^{C'}_{k+1}$. This immediately results in the first main result of this section, stated below. Consequently, $\thresh_{\game,\parity^C}(v) = \pTh(v)$.

\begin{restatable}[Existence of $\parity$ thresholds]{theorem}{thmPthCorrectness} \label{thm:pth:correctness}
  Consider the \GBG $(\game,\parity^C)$.  When the token is on $v$, if \PZ's budget exceeds $\pTh(v) = \fpTh_{d}(\cdot)(v)$, she has a strategy to win. If her budget is less than $\pTh(v)$, \PO has a winning strategy.
\end{restatable}

%  \noindent\textbf{Parity Winning Strategy.} Suppose the game token is on vertex $v$ and \PZ's budget is above $\pTh(v)$. We show how she can play in order to win the parity specification. The idea is that she maintains a set of vectors $Z_0, \dots, Z_d$ where $Z_k \in [0,1]^{C(k)}$ and $Z_{C(u)}(u)$ is intially $\pTh(u)$. During the game, \PZ plays such that her budget is always above the threshold of the current vertex. Whenever a vertex $u$ is reached, every $Z_i$ with $i < C(u)$ is discarded (i.e. assigned $\bot$). If $C(u)$ is even, she plays the one-step strategy provided by $\bellman^{T_{>k}}(Z_{k+1},\dots, Z_d)$ to ensure that her budget exceeds  

The recursive definition of $\pTh(v) = \fpTh_d(\cdot)(v)$ reveals that the parity threshold is computed as a \emph{nested quantitative fixed-point} of the operator $\bellman^\varnothing$. Let $\mu$ and $\nu$ denote the least and greatest fixed-point operators, respectively. At each color $k$, the frugal threshold $\fpTh_k$ is defined as $\mu$ of $\bellman^{T_{\geq k}}$ when $k$ is even and as $\nu$ of $\bellman^{T_{\geq k}}$ when $k$ is odd. The full parity threshold $\pTh$ is thus given by alternating least and greatest fixed-points, nested from the innermost color $0$ up to the outermost color $d$, instantiated with the operator $\bellman$. Formally,
$
  \pTh = \Theta_d X_{T_d} \dots \Theta_0 X_{T_0} \bellman^\varnothing(X),
$
where $\Theta_i = \mu$ for all even $i$, and $\Theta_i=\nu$ for all odd $i$.

This nested fixed-point structure has a precise counterpart in the theory of simple stochastic games, utilizing the SSG Bellman operator $\bellmanSSG$ defined earlier. It was shown in~\cite{AlfaroM04} that the optimal value function $\val$ of $\SSG(\G)$ with parity specification $\parity^C$ is computed as an alternating nested fixed-point of $\bellmanSSG$, using $\nu$ at even priority levels and $\mu$ at odd priority levels—the exact mirror of the alternation used in the threshold computation. Formally,
$
  \optval = \Omega_d X_{T_d} \dots \Omega_0 X_{T_0} \bellmanSSG^\varnothing(X),
$
where $\Omega_i = \nu$ for all even $i$, and $\Omega_i=\mu$ for all odd $i$.

The bridge between the two nested fixed-point computations is the pointwise duality between the two operators. This duality effectively exchanges least fixed-points $(\mu)$ and greatest fixed-points $(\nu)$ under the complementation map $X \mapsto \mathbf{1} - X$.
%Since $\bellman^S(X)$ and $\bellmanSSG^S(X)$ are exact duals of each other (Lemma~\ref{lem:bellman:property}), a least fixed-point ($\mu$) of $\bellman^S$ corresponds to a greatest fixed-point ($\nu$) of $\bellmanSSG^S$, and vice versa. 
As the threshold computation uses $\mu$ exactly where the SSG value computation uses $\nu$ (and vice versa), this equivalence propagates through the entire alternating sequence from level $0$ up to $d$:

\begin{restatable}[Equivalence to SSGs, Proof in Appendix~\ref{app:qual:parity}]{theorem}{thmParityThresholdSSG}\label{thm:parity-threshold-ssg}
  Let $(\game, \parity^C)$ be an \SGBG. Then for all $v \in \V$, it holds that $\pTh(v)=1-\val_{\SSG(\game),\parity^C}(v)$.
  % Consider the \GBG $(\game, \parity^C)$. Let $\pTh \colon \V \to [0,1]$ be its parity threshold function and let $\optval \colon \V \to [0,1]$ be the optimal value function of $\SSG(\G)$ under the parity specification $\parity^C$. Then, for every $v \in \V$:
  % \[
  %     \pTh(v) = 1 - \optval(v).
  % \]
\end{restatable}
Combining Theorem~\ref{thm:parity-threshold-ssg} with Proposition~\ref{prop:equivalence-GBG-SGBG} it is implied that computing parity thresholds in \GBGs is linearly equivalent to that of \SGBGs, which is structurally equivalent to computing parity values in SSGs.
%
% Theorem~\ref{thm:parity-threshold-ssg} directly implies the equivalence between $\THRESH_{\parity}^{\leq 0.5}$ and $\VALUE_{\parity}^{\geq 0.5}$:
%
% \begin{restatable}{corollary}{corParityEquivalence} \label{cor:parity:equivalence}
%     The membership problem of $\THRESH_{\parity}^{\leq 0.5}$ and the membership problem of $\VALUE_{\parity}^{\geq 0.5}$ are linearly reducible to each other.
% \end{restatable}
%
In particular, given that $\VALUE_{\parity}^{\geq 0.5}$ is in $NP \cap coNP$~\cite{BerthonKZ25,Condon92}, it follows that $\THRESH_{\parity}^{\leq 0.5}$ is also in $NP \cap coNP$. Moreover, many algorithms have been developed for computing the optimal values of parity stochastic games \cite{ChatterjeeH06,ChatterjeeJH04}, and all of these algorithms can be used to compute the parity thresholds of \GBGs. As stated earlier, in case $\VZ=\emptyset$ or $\VO=\emptyset$, the corresponding SSG $\SSG(\game)$ is an MDP. Computing parity values in MDPs reduces to finding maximal end-components (MECs) and computing the maximum probability of reaching winning MECs~\cite{modelcheckingbook}, which totally takes polynomial time. Thus, in such cases, the parity thresholds of \GBGs can be computed in polynomial time as well.

% \begin{restatable}[Threshold verification complexity for $\parity$]{corollary}{propParityComplexity} \label{prop:parity:complexity}
\begin{corollary}[Threshold verification complexity for $\parity$]\label{prop:parity:complexity}
  The following results hold:
  \begin{itemize}
    \item The membership problem of $\THRESH_{\parity}^{\leq 0.5}$ is in $\mathit{NP} \cap \mathit{coNP}$.
    \item If $\VZ=\emptyset$ or $\VO=\emptyset$, then $\THRESH_\parity^{\leq 0.5}$ membership is in $P$.
  \end{itemize}
\end{corollary}
% \end{restatable}

% \myparagraph{Strategy complexity.}
% In \BGs with reachability specifications, it is known that Eve's bids must account for the available budget, whereas the choice of the successor vertex, upon winning the bidding, only depends on the 

\myparagraph{Strategy complexity.}
The optimal strategies we constructed above need to account for various kinds of information from the past:
For reachability, only the current configuration suffices, similar to the case of \BGs.
For safety, only the current vertex suffices, irrespective of the budget; recall that for \BGs, safety has trivial solutions.
For parity, the path seen so far and the current budget suffices; recall that for \BGs, parity reduces to reachability.
We conjecture, only the current configuration would suffice for parity on \GBGs.

%% file: quantitative.tex
\section{Quantitative Specifications} \label{sec:quant}
In this section, we consider \GBGs with the quantitative specifications, namely mean payoff and discounted-sum. We show that threshold budgets exist, and prove the structural equivalence between \SGBGs and SSGs, as claimed in Section~\ref{sec:structural equivalence}. For mean payoff, this gives us the complexity of computing the thresholds in \GBGs, due to known results from the SSG literature. For discounted-sum, the complexity question remains open.

\subsection{Mean Payoff Specifications}
\label{sec:mean payoff}

In this subsection, we consider GBGs with the mean payoff specification, that is, given a reward function $R$ and a number $r$, a path $\rho = v^0v^1\cdots$ is winning for \PZ iff $\limsup_{k \to \infty} \frac{1}{k}\sum_{t=0}^{k-1} R(v^t) \geq r$. This specification has been studied in the literature~\cite{DBLP:conf/fsttcs/BrazdilBE10,DBLP:conf/lics/Chatterjee016}. It is noteworthy that the specification of \PZ is not to maximize the expected payoff, but rather to maximize the probability of having mean payoff at least $r$.
%This specification is Borel, submixing and prefix-invariant \cite{TODO}.
Since this specification is Borel, sub-mixing, and prefix-independent, positional strategies suffice:
\begin{restatable}[Memoryless determinacy of mean-payoff SSGs~{\cite[Proposition~1]{DBLP:conf/fsttcs/BrazdilBE10}}]{theorem}{theoMeanpayoffPositional}
  \label{theo:meanpayoff:positional}
  SSGs with mean payoff specification admit optimal positional strategies for \PZ and \PO.
\end{restatable}

The connection between random-turn games and pure bidding games has already been proved \cite{avni2019infinite}. Their result relies on the fact that strongly connected components in pure bidding games have threshold budget either 0 or 1.
We generalize this claim to any \SGBG for which every vertex of the game has the same expected mean payoff value. Combining this claim with the results in \cite{avni2019infinite}, we obtain the following theorem. More details and the proof are given in \Cref{app:sec:mean payoff proofs}.

\begin{restatable}[Existence of $\meanpayoff$ thresholds and equivalence to SSGs]{theorem}{theMeanpayoffMain}
  \label{the:meanpayoff:main}
  Let $(\G, \meanpayoff^{R}_{\geq r})$ be an \SGBG with a mean payoff specification and a reward function $R: V \rightarrow \mathbb{Q}$. The threshold budget function exists.
  Furthermore, for all $r \in \mathbb{R}$ and all $v \in V$,
  $
    \thresh_{\G, \meanpayoff^R_{\geq r}}(v) = 1 - \val_{\SSG(\G), \meanpayoff^R_{\geq r}}(v)
  $.
\end{restatable}

By \Cref{theo:meanpayoff:positional}, SSGs with mean payoff specification admit positional strategies. Computing the value for an MDP can be done in polynomial time~\cite{DBLP:conf/fsttcs/BrazdilBE10}. Thus we get the following:
% \begin{restatable}[Threshold verification complexity for $\meanpayoff$]{corollary}{propMeanComplexity} \label{prop:parity:complexity}
\begin{corollary}[Threshold verification complexity for $\meanpayoff$]\label{prop:MP:complexity}
  The following results hold:
  \begin{itemize}
    \item The membership problem of $\THRESH_{\meanpayoff^R_{\geq r}}^{\leq 0.5}$ is in $\mathit{NP} \cap \mathit{coNP}$.
    \item If $\VZ=\emptyset$ or $\VO=\emptyset$, then $\THRESH_{\meanpayoff^R_{\geq r}}^{\leq 0.5}$ membership is in $P$.
  \end{itemize}
\end{corollary}
% \end{restatable}

\subsection{Discounted-Sum Specifications}
\label{sec:discounted sum}

We are the first to consider discounted-sum specifications in the context of bidding games. We prove that the thresholds admit the same equivalence with SSGs as shown for other specifications.
Towards this goal, first, we describe a fixed-point characterization of the values of SSGs with discounted-sum specifications using a Bellman operator. Similar characterizations have been observed for Markov Decision Processes \cite{white1993minimizing}. Notice that the goal of Eve is to maximize the probability of achieving payoff at least $r$. This should not be confused with the goal of maximizing the expected payoff commonly studied in the literature. Without loss of generality, we assume all rewards in the game are non-negative. We can assume this since adding $c$ to every reward changes the payoff of an infinite play by $\frac{c}{1 - \lambda}$.

\begin{definition}[Bellman operator for discounted-sum SSGs]
  Let $(\game=(\V, \VZ, \VO, \VR, E),\allowbreak \discsum^{\lambda, R}_{\geq r})$ be an SSG with $\lambda \in (0, 1)$. For $ \valcand \in V\times\mathbb{R}  \rightarrow [0, 1]$, we define
  $\Bop: (V\times \mathbb{R} \to [0, 1]) \to (V\times \mathbb{R} \rightarrow [0, 1])$ such that for every $(v,r)\in V\times \mathbb{R}$, and $r' = \tfrac{r-R(v)}{\lambda}$,
  \begin{align*}
    \Bop[\valcand](v,r) \;=\;
    \begin{cases}
      1                                                                   & r \leq R(v),                   \\
      \max_{v'\in E(v)}\valcand\!\left(v',r'\right)                       & v\in \VZ,                      \\
      \min_{v'\in E(v)}\valcand\left(v',r'\right)                         & v\in \VO,                      \\
      \frac{\valcand\left(v_1,r'\right) + \valcand\left(v_2,r'\right)}{2} & v\in \VR, \{v_1, v_2\} = E(v).
    \end{cases}
    % & \quad\quad\text{For } r' = \tfrac{r-R(v)}{\lambda}.
  \end{align*}
  %
  %The substitution $r\mapsto(r-R(s))/\lambda$ is derived from $\discsum^{\lambda,R}(\rho) = R(v)+\lambda\discsum^{\lambda,R}(\rho[1\cdots])$.
\end{definition}
The value of discounted-sum SSGs is characterized by taking the left-continuous closure of the limit of iterates of $\Bop$.
The proof is provided in Appendix~\ref{app:sec:discounted sum proofs}. To the best of our knowledge, it is not known whether the function $(v, r) \to \val_{\SSG(\G), \discsum^{\lambda, R}_{\geq r}}(v)$ is computable. The left-continuous closure is necessary, as it is demonstrated in Example~\ref{ex:left continuous closure is needed} in Appendix~\ref{app:sec:discounted sum proofs}.
\begin{restatable}[Values in discounted-sum SSGs]{theorem}{theDiscValues}
  \label{the:disc:values}
  For every $v \in \V$ and every $r \in \mathbb{R}$, it holds that
  $
    \val_{\SSG(\G), \discsum^{\lambda, R}_{\geq r}}(v) \;=\; \lim_{\varepsilon\to 0^+} \lim_{k\to \infty}\Bop^k[\valcand_0](v,r - \varepsilon)
  $,
  where $\valcand_0(v, r) = 0$. %\1[r\le R_{\min}/(1 - \lambda)]$ where $R_{\min}$ is the minimal weight in $R$.
\end{restatable}
Next, we prove the correspondence between the thresholds of the \SGBGs and the values of the simple stochastic game. The proof uses the following idea: Let $\rho = v^0v^1\cdots v^k$ be a finite history of a play and inductively define $r^0 = r$ and $r^{n+1} = \frac{r^n - R(v^n)}{\lambda}$ for any $n \in \mathbb{N}$. If Eve starts the game in a vertex $v$ with budget more than $1 - \val_{\SSG(\G), \discsum^{\lambda, R}_{\geq r}}(v)$, then she can maintain the following invariant:
For any $k$, Eve has strictly more budget than $1 - \val_{\SSG(\G), \discsum^{\lambda, R}_{\geq r^k}}(v^k)$.
Thanks to this invariant, she can guarantee overall payoff at least $r$. The situation is almost symmetric for Adam if Eve has budget less than $1 - \val_{\SSG(\G), \discsum^{\lambda, R}_{\geq r}}(v)$ but we need to use the left-continuity property of the value function.

\begin{restatable}[Existence of $\discsum$ thresholds and equivalence to SSGs]{theorem}{theDiscConnection}
  \label{the:disc:connection}
  Let $(\G, \discsum^{\lambda, R}_{\geq r})$ be a \SGBG with a discounted-sum specification where $\lambda \in (0, 1)$. The threshold budget function exists.
  Furthermore,
  $\thresh_{\G, \discsum^{\lambda, R}_{\geq r}}(v) = 1 - \val_{\SSG(\G), \discsum^{\lambda, R}_{\geq r}}(v) $
  for all $r \in \mathbb{R}$ and all $v \in V$.
\end{restatable}
We remark that the computational complexity of the $\THRESH_{\discsum^{\lambda, R}_{\geq r}}^{\leq 0.5}$ problem is open since the corresponding decision problem for SSGs is still open.

\myparagraph{Strategy complexity.}
The optimal strategies for discounted-sum specifications may need to account for the path seen so far. Intuitively, since the specification is not prefix-independent, the player needs to know the accumulated payoff. We illustrate this in \Cref{ex:disc:memory} in \Cref{app:sec:discounted sum proofs}. This property also translates to SSGs, which is a curious observation.

%% file: repair.tex
\section{The Threshold Repair Problem}
\label{sec:repair}

Inspired by the potential application of \GBGs in auction-based scheduling for multi-objective decision-making (see the motivating example in the Introduction), we introduce the following threshold repair problem, where we want to change a small set of bidding vertices to Eve's vertices in order to reduce the threshold to below a given target value. 
\begin{problem}[Threshold repair]\label{prob:repair}
    Suppose we are given an \SGBG $(\G = (\V,\VZ,\VO,\VB,\E),\varphi)$, a set of bidding vertices $C\subseteq \VB$ called the constraint set, an integer $k>0$, and a target threshold $\gamma\in (0,1)$ for some vertex $v \in \V$.
    Does there exist a new \SGBG $(\G'= (\V,\VZ',\VO,\VB',\E),\varphi)$, only differing in Eve's and bidding vertices, such that:
    \begin{enumerate}[(i)]
        \item only some \emph{bidding} vertices from $\VB$ got converted to Eve's vertices in $\VZ'$, i.e., $\VB'\subseteq \VB$, implying $\VZ\subseteq\VZ'$ and $\VZ'\subseteq \VZ\cup \VB$;\label{cond:repair:1}
        \item all bidding vertices in $\VB$ that got converted to Eve's vertices in $\VZ'$ are in the constraint set $C$, i.e., $C\cap \VZ'=\emptyset$;
        \item the number of vertices converted to Eve's vertices is at most $k$, i.e., $|\VZ'\setminus \VZ|\leq k $; 
        \item the threshold of the new game is at most $t$, i.e., $\thresh_{\G',\varphi}(v)\leq t$.\label{cond:repair:4}
    \end{enumerate}
\end{problem}
From the equivalence between \SGBGs and SSGs, the problem has a direct analogue in the SSG setting, namely how can we increase the value of an SSG by changing a small set of random vertices with Eve's vertices? 
This problem could be of independent interest in applications where the randomness in the SSG could be replaced by deterministic control for a price.
Below we state that the repair problem is NP-complete; the complexity proof and an MILP-based algorithm can be found in Appendix~\ref{appendix:sec:repair}.

\begin{restatable}[Threshold repair complexity, proof in Appendix~\ref{appendix:sec:repair np-complete}]{theorem}{thmRepairNPC} \label{thm:repair:npc}
    The threshold repair problem is NP-complete for GBGs with reachability, safety, and parity specifications.
\end{restatable}

%% file: appendix.tex
\section{Motivating Example: The Game Bid-Tac-Toe}\label{appendix:sec:bid-tac-toe}
Consider a bidding variant of the game tic-tac-toe, which we call \emph{bid}-tac-toe; the game can be modeled as a \GBG as shown in Figure~\ref{fig:motivating example:bid-tac-toe}.
Let Eve's symbol be circles and Adam's symbol be crosses.
Unlike the standard variant where the players take turns for placing their symbols, at each step, a bidding takes place, and the winner of the bidding decides who plays next.
To model the bidding, we consider the first-price Richman variant, where the players initially start with budgets whose sum is $1$.
In each round, both players submit their bids simultaneously and concurrently, the highest bidder chooses who plays next, and pays the bid amount to the opponent.
All bidding games in the rest of the paper will use the first-price Richman mechanism.
The game ends if either all cells are filled, or if one of the players has put three symbols in a line.
Eve wins if she manages to form a line (see vertex $k$), and loses otherwise (see vertex $l$), regardless of whether Adam forms a line.
While bid-tac-toe can be easily modeled using \GBGs, as shown in Figure~\ref{fig:motivating example:bid-tac-toe}, it cannot be modeled using a classical bidding game.
This is because in classical bidding games, every vertex is a bidding vertex, and the outgoing edges are independent of the identity of the highest bidder.
This would mean that, in the classical encoding of bid-tac-toe, Eve could place crosses and Adam could place circles, which is not allowed by the rule of the game.\footnote{While existing literature did use the classical model to represent bid-tac-toe~\cite{avni2020all}, strictly speaking, it did not conform with the syntax of the game.
  In particular, Eve was allowed to place crosses and Adam was allowed to place circles.
  Incidentally, it did not matter, because the specification of Eve implicitly discourages her to put crosses, and vice versa.}

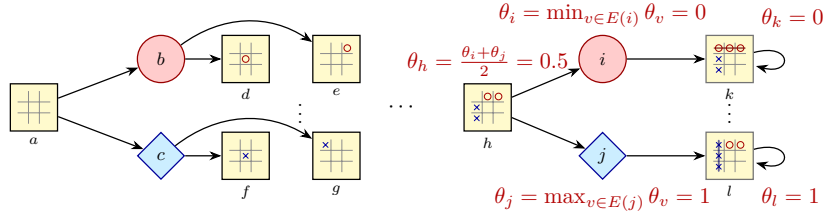
\begin{figure}[h]
  \centering
  \input{FIGURES/tic-tac-toe.tex}
  \caption{
    Snippets of the \GBG modeling bid-tac-toe; see Figure~\ref{fig:non-trivial threshold in SCC} for vertex legends.
    The bidding vertices show the current board configuration.
    Eve and Adam respectively win by reaching vertices like $k$ and $l$.
    For vertices $h$--$l$, we show the threshold budget computation (details in Section~\ref{sec:qual}), where  $E(x)$ and $\theta_x$ respectively denote the set of successors and the threshold of a given vertex $x$.
  }
  \label{fig:motivating example:bid-tac-toe}
\end{figure}

\section{Distinction between \GBGs and Pure Bidding Games} \label{app:sec:distinction}

Our main goal in this section is to show that \GBGs are strictly more general than pure bidding games. We first briefly recall the syntax of Alternating Temporal Logic (ATL) and its extension, ATL$^*$ and then show an ATL$^*$ formula that is satisfied by some \GBG but cannot be satisfied by any pure bidding game.

\myparagraph{ATL.} Introduced to reason about multi-agent systems, ATL replaces the path quantifiers of CTL with strategic path quantifiers $\langle\langle A \rangle\rangle$, meaning ``the coalition of agents $A$ has a joint strategy to ensure''.
To directly connect this framework with \GBGs, where initial budgets uniquely determine the players' abilities to enforce outcomes, we extend the standard ATL structures to incorporate an initial budget allocation.
Focusing on the two-player setting where $\Sigma=\{\PZ, \PO\}$ and the players' budgets always sum to 1, we augment the strategic quantifier with a budget parameter $b \in [0,1]$. We write $\langle\langle A \rangle\rangle^{>b}$ to denote that when \PZ's budget is above $b$, coalition $A$ has a strategy to ensure the specified behavior against any strategy of the opposing coalition.

Given a set of atomic propositions $AP$, the extended \emph{Budgeted ATL} syntax is defined inductively as follows:
\[ \varphi ::= p \mid \neg\varphi \mid \varphi \vee \varphi \mid  \varphi \wedge \varphi \mid \langle\langle A \rangle\rangle^{>b} X \varphi \mid \langle\langle A \rangle\rangle^{>b} G \varphi \mid \langle\langle A \rangle\rangle^{>b} F \varphi \mid \langle\langle A \rangle\rangle^{>b} \varphi \mathcal{U} \varphi \]
where $p \in AP$, $A \subseteq \Sigma$, and $b \in [0,1]$. The temporal operators $X$, $G$, $F$, and $\mathcal{U}$ represent ``next'', ``globally'', ``eventually'', and ``until'', respectively.

In the context of \GBGs, satisfying a formula such as $\langle\langle \PZ \rangle\rangle^{>b} \psi$ at a given vertex $v$ maps directly to our exact definition of threshold budgets (see Section~\ref{sec:model}): it evaluates to true if and only if the threshold budget required for \PZ to enforce the path specification $\psi$ from $v$ satisfies $\thresh_{\G,\psi}(v) \leq b$.

\emph{Budgeted ATL$^*$} generalizes this naturally by decoupling the strategic quantifiers from the temporal operators, allowing unstructured nesting and boolean combination of path properties. It recursively distinguishes between state formulas ($\varphi$) and path formulas ($\psi$):
\begin{align*}
  \varphi & ::= p \mid \neg\varphi \mid \varphi \vee \varphi \mid \langle\langle A \rangle\rangle^{>b} \psi           \\
  \psi    & ::= \varphi \mid \neg\psi \mid \psi \vee \psi \mid X\psi \mid G\psi \mid F\psi \mid \psi \mathcal{U} \psi
\end{align*}
where $p \in AP$, $A \subseteq \Sigma$, and $b \in [0,1]$.

Our main goal in this section is to show that \GBGs are strictly more general than pure bidding games. It is easy to see that pure bidding games are special cases of \GBGs. We show an ATL$^*$ property that is satisfied by some \GBGs but not by any pure bidding game. This shows that pure bidding games a strict subclass of \GBGs.

Let $AP = \{p_1,p_2,p_3\}$ be the set of atomic propositions. Consider the following ATL$^*$ formulas:
\begin{align*}
  \varphi_1 & = \langle\langle \PZ,\PO \rangle \rangle^{>0} \bigwedge_{1 \leq i,j \leq 3} G(p_i \Rightarrow F p_j) \\
  \varphi_2 & = \neg \langle\langle \PZ,\PO \rangle \rangle^{>0} FG(\neg p_1 \wedge \neg p_2 \wedge \neg p_3)      \\
  \varphi_3 & = \langle \langle \PZ \rangle \rangle^{>0.5} G(\neg p_1)
\end{align*}

Intuitively, $\varphi_1$ holds when every $p_i$ vertex has a path to a $p_j$ vertex for every $i,j$. On the other hand, $\varphi_2$ is satisfied when \PZ and \PO cannot collaborate to ensure that no $p_i$ is reached from some point on. Consequently, every SCC of the game graph must contain at least one $p_i$ vertex for some $i$. $\varphi_1$ and $\varphi_2$ together imply that every bottom SCC of the game graph has at least one $p_i$ vertex for each $i$. Lastly, $\varphi_3$ states that when \PZ's initial budget is above $0.5$, she has a strategy to ensure that $p_1$ is never reached, i.e. $\thresh_{\game,\safe(p_1)} \leq 0.5$.

We show a \GBG that satisfies $\Phi = \varphi_1 \wedge \varphi_2 \wedge \varphi_3$. We also show that no pure bidding game can satisfy $\Phi$.

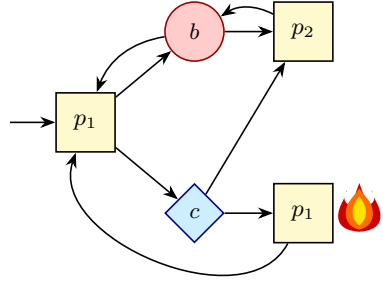
\begin{figure}
  \centering
  \input{FIGURES/distinction}
  \caption{A \GBG that satisfies $\Phi = \varphi_1 \wedge \varphi_2 \wedge \varphi_3$.}
  \label{fig:app:distinction}
\end{figure}

Consider the \GBG in Figure~\ref{fig:app:distinction}. It satisfies $\varphi_1$, because the whole game graph is strongly connected. It satisfies $\varphi_2$, because the only SCC of the game graph contains all $p_1,p_2$ and $p_3$. Finally, it satisfies $\varphi_3$, because if \PZ's budget is above $0.5$, she can bid all in at $p_1$, win the first bidding (because \PO's budget is less than 0.5), and move the game to vertex $b$, where she can trap the game in $\{b,p_2\}$ indefinitely and satisfy $G(\neg p_1)$.

On the other hand, assume for the same of contradiction that a pure bidding game $\G$ satisfies $\Phi$. As shown above, every bottom SCC of $\G$ must contain all $p_1,p_2$ and $p_3$. This means that \PO can always force the game to reach $p_1$, if his budget is non-zero~\cite{avni2020survey}. This is in direct contradiction with $\varphi_3$ being satisfied, which shows no such pure bidding game $\G$ can exist.

\section{Proof of Proposition~\ref{prop:equivalence-GBG-SGBG}} \label{app:sec:GBG-to-SGBG}

\propEquivalenceGBGSGBG*
\begin{proof}
  Obviously, any \SGBG is a \GBG, and the identity function is a linear reduction from \SGBGs to \GBGs regardless of the specification.

  Now suppose $(\game,\spec)$ is a \GBG where the out-degree of its bidding vertices is arbitrary. Our goal is to find a function $f$ where (i) $f(\game,\spec) = (\game',\spec')$ is a \SGBG, (ii) $f$ is computable in linear time and (iii) $\thresh_{\game,\spec} = \thresh_{\game',\spec'}$. We do a case analysis on the type of $\spec$ and find $f(\game,\spec)$ in each case:
  \begin{itemize}
    \item \textbf{Parity.} Suppose $\game = (\V,\VZ,\VO,\VB,\E)$ and $\spec=\parity^C$ for $C\colon \V \to [0,\dots,d]$. Define $\game' = (\V',\VZ',\VO',\VB',\E')$ and $C'\colon \V' \to [0,\dots, d]$ as follows:
          \begin{itemize}
            \item $\VB' = \VB$
            \item $\VZ' \supseteq \VZ$
            \item $\VO' \supseteq \VO$
            \item For every $v \in \VZ \cup \VO$, $\E'(v) = \E(v)$
            \item For all $v \in \V$, $C'(v)=C(v)$.
            \item For every $v \in \VB$, there are two additional vertices $u^v_\eve \in \VZ'$ and $u^v_\adam \in \VO'$ where $E'(v_b) = \{u^v_\eve,u^v_\adam\}$ and $E'(u^v_\eve)=E'(u^v_\adam) = E(v)$. Moreover, $C(u^v_\eve) = C(u^v_\adam) = C(v)$.
          \end{itemize}
          An example of the above transformation is depicted in Figure~\ref{fig:non-trivial threshold in SCC}.

          Clearly, $f(\game,\parity^C) = (\game',\parity^{C'})$ is linearly computable since the only change in $\game$ is the introduction of new $u_\eve$ and $u_\adam$ vertices, which can explode the size of $\game$ by a factor of at most 2. Moreover, $f(\game,\parity^C)$ is an \SGBG simply because we are replacing every bidding vertex in $\game$ with a bidding vertex whose out-degree is 2.

          Given the structure of $\game'$ relative to $\game$, any finite path $\rho' = v_0, \dots, v_k$ of $\game'$ can be transformed into a finite path $\rho$ of $\game$ simply by discarding all the $u^v_\eve$ and $u^v_\adam$ vertices from $\rho'$. We call $\rho$ the path induced by $\rho'$.

          Now consider an \PZ strategy $\polZ$ in $\game$. This strategy can be directly translated into a strategy $\polZ'$ in $\game'$. Let $\rho' \in \finpath{\G'}$ be a finite path in $\game'$, $\rho$ be the path in $\game$ induced by $\rho'$ and $B_\eve \in [0,1]$ be the budget of \PZ after taking $\rho'$. The strategy translation is as follows:
          \begin{itemize}
            \item If $\last(\rho') \in \VB$, suppose $\last(\rho')=v$ and $\polZ(\rho,B_\eve) = (b,x)$ for some $b \in [0,B_\eve]$ and $x \in \E(v)$. Then, we define $\polZ'(\rho',B_\eve) = (b,u_\eve^v)$ and $\polZ'(\rho' \cdot u_\eve^v,B_\eve) = x$.
            \item If $\last(\rho') \in \VZ$, then $\polZ'(\rho',B_\eve) = \polZ(\rho,B_\eve)$.
          \end{itemize}
          In other words, whenever \PZ was winning a bid in $\rho$ to move the token to $x$, she would win the same bid using $\polZ'$ and would move the token first to $u_\eve^v$ and then to $x$. Because the sequence of winning bids is preserved, the colors along the generated path in $\game'$ correspond directly to the colors in $\game$, with the colors of bidding vertices duplicated sequentially. Since repeating the same color finitely many times does not change the maximum color appearing infinitely often, $\polZ'$ guarantees the parity condition in $\game'$ if and only if $\polZ$ does in $\game$. Symmetrically, any winning strategy for \PO in $\game$ translates to one in $\game'$. Thus, $\thresh_{\game,\parity^C} = \thresh_{\game',\parity^{C'}}$.

    \item \textbf{Discounted Sum.} Suppose $\game = (\V,\VZ,\VO,\VB,\E)$ and $\spec$ is a discounted-sum specification with reward function $R\colon \V \to \mathbb{Q}$ and discount factor $\lambda \in (0,1)$. We define $\game' = (\V',\VZ',\VO',\VB',\E')$ where every vertex $v \in \V$ is replaced by three vertices: the vertex $v$ itself, an Eve choice vertex $v_\eve \in \VZ'$, and an Adam choice vertex $v_\adam \in \VO'$. The types of the original vertices are preserved, meaning $\VB' = \VB$, while $\VZ'$ and $\VO'$ contain the original control vertices along with all newly added choice vertices. For the edges, the original vertex $v$ is connected to the two choice vertices, i.e., $\E'(v) = \{v_\eve, v_\adam\}$. The choice vertices are then connected to the successors of the original vertex $v$, meaning $\E'(v_\eve) = \E'(v_\adam) = \E(v)$. This ensures that every bidding vertex in $\game'$ has an out-degree of exactly 2, satisfying the definition of \SGBGs.

          The discount factor in the new game is set to $\lambda' = \sqrt{\lambda}$. The reward function $R'$ is defined as $R'(v) = R(v)$ for all $v \in \V$, and the rewards of the new choice vertices are set to zero, i.e., $R'(v_\eve) = R'(v_\adam) = 0$. This preserves the values because every step in $\game$ corresponds to two steps in $\game'$, and in every other step, the incurred reward is zero. The reward obtained at step $i$ in $\game$ is obtained at step $2i$ in $\game'$, subjected to a discount of $(\lambda')^{2i} = (\sqrt{\lambda})^{2i} = \lambda^i$. Because the choice vertices provide identical options for both players, the strategic values and threshold functions remain equivalent.

    \item \textbf{Mean-Payoff.} Suppose $\game = (\V,\VZ,\VO,\VB,\E)$ and $\spec$ is a mean-payoff specification with reward function $R\colon \V \to \mathbb{Q}$. The transformation from $\game$ to $\game'$ is done similar to the discounted-sum case. The reward function $R'$ is defined as $R'(v) = 2R(v)$ for all $v \in \V$, and the rewards of the new choice vertices are set to zero, i.e., $R'(v_\eve) = R'(v_\adam) = 0$. This preserves the values because every step in $\game$ corresponds to two steps in $\game'$, and in every other step, the incurred reward is zero. The mean payoff of a path $\rho'$ in $\game'$ is the same as the mean payoff of the path $\rho$ in $\game$ induced by $\rho'$. Because the choice vertices provide identical options for both players, the strategic values and threshold functions remain equivalent.
  \end{itemize}
\end{proof}

\section{Omitted Proofs of Section~\ref{sec:qual:reach-safe}} \label{app:qual:proofs}

\begin{restatable}{lemma}{lemBellmanProperty} \label{lem:bellman:property}
  Let $\G = (\V,\VZ,\VO,\VB,\E)$ be a game graph, $S \subseteq \V$ and $X \colon \V \to [0,1]$ be arbitrary. The following statements hold:
  \begin{enumerate}
    \item $\bellman^S$ is monotone in $X$, i.e. if $X \geq X'$, then $\bellman^S(X) \geq \bellman^S(X')$.
    \item If the game is in vertex $v \in \V \setminus S$, and \PZ's budget exceeds $\bellman^S(X)(v)$, then she has a strategy to ensure that if $u$ is reached in the next step, then her budget exceeds $X(u)$.
    \item $\bellman^S(\mathbf{1} - X) = \mathbf{1} - \bellmanSSG^S(X)$, where $\bellmanSSG^S$ is the SSG Bellman operator.
  \end{enumerate}
\end{restatable}
\begin{proof} We prove each statement separately:
  \begin{itemize}
    \item\textbf{Proof of Statement 1.} The first statement (monotonicity) holds because the operators $\min$ and $\max$ over a set, as well as their average, are monotone. That is, if $X \geq X'$ pointwise, then for any $v \in \V$, $\min_{u \in \E(v)} X(u) \geq \min_{u \in \E(v)} X'(u)$ and $\max_{u \in \E(v)} X(u) \geq \max_{u \in \E(v)} X'(u)$, which implies $\bellman^S(X) \geq \bellman^S(X')$.

    \item\textbf{Proof of Statement 2.} For the second statement, suppose the game is in vertex $v \in \V \setminus S$, and \PZ's budget $B_\eve$ strictly exceeds $\bellman^S(X)(v)$. Let $v^+ = \argmax_{u \in \E(v)} X(u)$ and $v^- = \argmin_{u \in \E(v)} X(u)$. We do a case analysis on the control of $v$:
          \begin{itemize}
            \item If $v \in \VZ$, then $\bellman^S(X)(v) = X(v^-)$. Since $B_\eve > X(v^-)$, \PZ can choose $v^-$ as the successor vertex. As budgets remain unchanged at control vertices, in the next step the vertex is $v^-$ where her budget $B_\eve$ still strictly exceeds $X(v^-)$.
            \item If $v \in \VO$, then $\bellman^S(X)(v) = X(v^+)$. No matter which successor vertex $u \in \E(v)$ is chosen by \PO, it holds that $X(v^+) \geq X(u)$. Thus $B_\eve > X(v^+) \geq X(u)$. Her budget remains unchanged and strictly exceeds $X(u)$ upon reaching $u$.
            \item If $v \in \VB$, then $\bellman^S(X)(v) = \frac{X(v^+)+X(v^-)}{2}$. \PZ can submit a bid equal to $b = \frac{X(v^+) - X(v^-)}{2}$.
                  If she wins the bidding, she selects $v^-$ as the next vertex. Her remaining budget is $B_\eve - b$, which is strictly greater than $\frac{X(v^+)+X(v^-)}{2} - \frac{X(v^+)-X(v^-)}{2} = X(v^-)$.
                  If she loses the bidding, it is due to \PO bidding $b_{\adam}$ where $b_{\adam} > b$ (since ties are resolved in favor of \PZ). \PO chooses some successor $u \in \E(v)$ and pays $b_{\adam}$ to \PZ. \PZ's updated budget becomes $B_\eve + b_{\adam} > B_\eve + b$. This is strictly greater than $\frac{X(v^+)+X(v^-)}{2} + \frac{X(v^+)-X(v^-)}{2} = X(v^+)$. Since $X(v^+) \geq X(u)$ for all successors, her new budget strictly exceeds $X(u)$ upon reaching $u$.
          \end{itemize}
          In all cases, \PZ has a strategy that guarantees her budget in the subsequent vertex $u$ will strictly exceed $X(u)$.

    \item\textbf{Proof of Statement 3.} For the third statement ($\bellman^S(\mathbf{1} - X) = \mathbf{1} - \bellmanSSG^S(X)$), evaluating the complement of the input vector exchanges every discrete $\min$ operator with a $\max$ operator for vertices in $\VZ \setminus S$ and $\VO \setminus S$, while maintaining the average evaluation at probabilistic vertices. This corresponds to the SSG definition mapping $\bellmanSSG^S(X)$.
  \end{itemize}
\end{proof}

\thmReachWellStrat*
\begin{proof}
  \textbf{Well-definedness:} We show that the sequence $[\rTh^i(v)]_{i=0}^\infty$ is non-increasing and bounded in $[0,1]$. By definition, $\rTh^0(v) \in \{0,1\}$, and $\rTh^i(v)$ recursively applies the minimum, maximum, or average over previous sequence values, ensuring that $\rTh^i(v) \in [0,1]$ for all $i$. To show monotonicity, we prove $\rTh^i \leq \rTh^{i-1}$ pointwise by induction on $i$. For $i=1$, since $\rTh^0(v) = 1$ for all $v \notin T$, it holds that $\rTh^1(v) \leq \rTh^0(v)$. Since the operator $\bellman^T$ is monotone (Lemma~\ref{lem:bellman:property}), $\rTh^i \leq \rTh^{i-1}$ implies that $\rTh^{i+1} = \bellman^T(\rTh^i) \leq \bellman^T(\rTh^{i-1}) = \rTh^i$. The sequence is therefore non-increasing and bounded, making its pointwise limit $\rTh(v) = \lim_{i \to \infty} \rTh^i(v)$ well-defined.

  \textbf{Finite-Horizon Thresholds $\rTh^i$:} We use induction on $i$.
  For $i=0$, since \PZ's budget satisfies $B_\eve \leq 1$, it is impossible to have budget strictly more than $1$. Thus $v \in T$, and the game is already won in $0$ steps.

  For the inductive step, suppose the lemma holds for $i-1$, and consider the case for $i$. Suppose \PZ's budget at $v$ strictly exceeds $\rTh^i(v)$. If $v \in T$, she has already won. Otherwise ($v \notin T$), by definition $\rTh^i(v) = \bellman^T(\rTh^{i-1})(v)$. Since her budget $B_\eve > \bellman^T(\rTh^{i-1})(v)$, Lemma~\ref{lem:bellman:property} implies she has a strategy to ensure her budget in the next step strictly exceeds $\rTh^{i-1}(u)$ upon reaching the next vertex $u$. By the induction hypothesis, starting from $u$ with budget exceeding $\rTh^{i-1}(u)$, she can guarantee reachability of $T$ in at most $i-1$ additional steps, ensuring reachability in at most $i$ total steps.

\end{proof}

\thmSafeWellStrat*
\begin{proof}
  \begin{itemize}
    \item \textbf{Finite-Horizon Safety:} We proceed by induction on $i$. For $i=0$, if $v \notin T$, $\sTh^0(v) = 0$, and \PO's budget $B_\adam \geq 0$ achieves preventing $T$ for $0$ steps. If $v \in T$, it is impossible for \PO's budget to exceed $\sTh^0(v)=1$.
          For the inductive step, suppose for $j+1$ that $B_\adam > \sTh^{j+1}(v) = \bellman^T(\sTh^j)(v)$. By symmetry and duality of the $\bellman$ operators, \PO has a strategy to ensure that for any reached successor vertex $u$, his remaining budget maintains $B_\adam > \sTh^j(u)$. By the induction hypothesis, starting from vertex $u$, \PO can prevent $T$ from being visited for at least $j$ more steps, securing safety for at least $j+1$ steps.
    \item \textbf{Well-definedness:} By the definition of $\bellman^T$ applied iteratively from the base case $\sTh^0$, the sequence $[\sTh^i(v)]_{i=0}^\infty$ is bounded within $[0,1]$ and monotonic (analogous to Theorem 13). Thus, its limit $\sTh(v) = \lim_{i \to \infty} \sTh^i(v)$ converges and is well-defined.
  \end{itemize}
\end{proof}

\section{Omitted Proofs of Section~\ref{sec:qual:parity}} \label{app:qual:parity}
\subsection{Construction of $\fpTh$}
We show that $\fpTh_k$ can be computed recursively. We first show how $\fpTh_0$ is computed (base) and then show how $\fpTh_{k}$ is computed from $\fpTh_{k-1}$ (recursive step):

\begin{itemize}
  \item \textbf{(Base)} Suppose $Z_1, \dots Z_d$ are candidate threshold functions for $T_1, \dots T_d$ vertices. Define the sequence $Y^0_i$ as follows:
        \begin{align*}
          Y^0_0(v) = \begin{cases}
                       Z_{C(v)}(v) & v \in T_{>0} \\
                       0           & v \in T_0
                     \end{cases} &  & Y^0_i(v) = \bellman^{T_{>0}}(Y^0_{i-1})
        \end{align*}
        % \[
        % Y_i(v) = \begin{cases}
        %     Z_{C(v)}(v) & v \in T_{>0} \\ 
        %     0 & v \in T_0 \wedge i=0 \\
        %     \min_{u \in E(v)}Y_{i-1}(u) & v \in T_0 \cap \vpo \wedge i>0 \\ 
        %     \max_{u \in E(v)}Y_{i-1}(u) & v \in T_0 \cap \vpt \wedge i>0 \\ 
        %     \frac{\min_{u \in E(v)}Y_{i-1}(u) + \max_{u \in E(v)}Y_{i-1}(u)}{2} &  v \in T_0 \cap \VBid \wedge i>0
        % \end{cases} 
        % \]
        Intuitively, since $Y^0_0 = 0$ on $T_0$ and $\bellman^{T_{>0}}$ is monotone, the sequence $Y^0_i$ is non-decreasing and converges to a fixed point of $\bellman^{T_{>0}}$. At this fixed point, if \PZ's budget exceeds the limit value at $v \in T_0$, Lemma~\ref{lem:bellman:property} allows her to maintain her budget above the limit at every subsequent step within $T_0$. Thus, either the game stays in $T_0$ forever (satisfying the internal win, since color~$0$ is even), or it escapes to some $u \in T_{>0}$ with her budget exceeding $Z_{C(u)}(u)$.

        Let $\fpTh_0(Z_1, \dots, Z_d) = \lim_{i \to \infty} Y^0_i$. Given the monotonicity of $\bellman$, the limit is well-defined (formally proven in Lemma~\ref{lem:fpth:correctness}).
  \item \textbf{(Recursive Step)} Suppose $Z_{k+1}, \dots, Z_d$ are candidate threshold functions for $T_{k+1}, \dots T_d$. Define the sequence of vectors $Y^k_i \colon T_{\leq k} \to [0,1]$ as follows.
        \begin{align*}
          \hspace{-3em}
          Y^k_0(v) = \begin{cases}
                       Z_{C(v)}(v) & v \in T_{>k}                             \\
                       0           & v \in T_{\leq k} \wedge k\text{ is even} \\
                       1           & v \in T_{\leq k} \wedge k\text{ is odd}  \\
                     \end{cases} &  & Y^k_i(v) = \bellman^{T_{\neq k}}(\fpTh_{k-1}({Y^{k}_{i-1}}|_{T_k},Z_{k+1}, \dots, Z_d))
        \end{align*}

        Intuitively, the interpretation of $Y^k_i$ depends on the parity of $k$. If $k$ is even, $Y^k_i$ is the budget \PZ needs to guarantee either (i)~winning in $T_{<k}$, (ii)~escaping to $T_{>k}$ with budget exceeding $Z_{C(u)}(u)$, or (iii)~visiting $T_k$ at least $i$ times. If $k$ is odd, $Y^k_i$ is the budget she needs to guarantee (i) or (ii) as above, or (iii) visiting $T_k$  \emph{at most} $i$ times. In both cases, the mechanism is the same: when \PZ's budget exceeds $Y^k_i(v)$ at $v \in T_k$, she can play so that after one step her budget exceeds $\fpTh_{k-1}({Y^{k}_{i-1}}|_{T_k},Z_{k+1}, \dots, Z_d)(u)$ (by Lemma~\ref{lem:bellman:property}), which by the induction hypothesis on $k-1$ ensures she either achieves~(i), (ii), or returns to $T_k$ with budget above $Y^k_{i-1}$. At the limit, for even $k$ the sequence is non-decreasing (starting from~$0$), so \PZ's budget is self-sustaining at $T_k$ vertices and she wins if $T_k$ is visited infinitely often. For odd $k$, the sequence is non-increasing (starting from~$1$), so budget above the limit exceeds $Y^k_I$ for some finite $I$, bounding the visits to $T_k$ and forcing \PZ to eventually escape or win in $T_{<k}$.
        %her budget is either enough to win the game in $T_{k-1}$ (provided by $\fpTh_{k-1}$), or reach some vertex $u \in T_{>k}$ with a budget exceeding $Z_{C(u)}(u)$ (provided by $\fpTh_{k-1})$, or reach a vertex $u \in T_k$ while having budget above $Y_{i-1}^0$ which is enough for reaching $T_k$ at least $i-1$ more times.

        Let $\fpTh_k(Z_{k+1}, \dots, Z_d) = \lim_{i \to \infty} Y^k_i$. Given the definition of $Y^k_i$, it can be shown that $Y^k_i$ is a monotonic sequence, hence the limit is well-defined (See Lemma~\ref{lem:fpth:correctness}).
\end{itemize}

\lemFpthCorrectness*
\begin{proof}
  We prove the lemma using an outer induction on $k$.
  To aid the proof, we first state that $\fpTh_{k}$ is non-decreasing in all its arguments: a straightforward induction on $k$ and $i$ confirms that if $Z'_l \geq Z_l$, the resulting sequence $\lim_{i \to \infty} Y'^k_i \geq \lim_{i \to \infty} Y^k_i$ due to Lemma~\ref{lem:bellman:property}.
  \begin{itemize}
    \item \textbf{(Outer Induction Base)} For $k=0$, suppose $Z_1, \dots, Z_d$ are given. We first show $Y^0_i$ is non-decreasing in $i$. For $i=0$, if $v \in T_0$, $Y^0_1(v) = \bellman^{T_{>0}}(Y^0_0)(v) \geq 0 = Y^0_0(v)$; if $v \in T_{>0}$, $Y^0_1(v) = Y^0_0(v) = Z_{C(v)}(v)$. By monotonicity of $\bellman$ (Lemma~\ref{lem:bellman:property}), since $Y^0_1 \geq Y^0_0$, it implies inductively that $Y^0_i \geq Y^0_{i-1}$ for all $i$. Hence $Y^0_i$ is a non-decreasing, bounded sequence and $\fpTh_0 = \lim_{i\to\infty} Y^0_i$ is well-defined.

          Now suppose \PZ's budget $B_\eve = \fpTh_0(Z_1, \dots, Z_d)(v) + \epsilon$ for some $\epsilon>0$. By definition, for $v \in T_0$, her budget is strictly above $\bellman^{T_{>0}}(\fpTh_0(Z_1,\dots,Z_d))(v)$. By Lemma~\ref{lem:bellman:property}, she has a strategy to ensure that in the next step, her budget strictly exceeds $\fpTh_0(Z_1,\dots,Z_d)(u)$ for the reached vertex $u$. By repeatedly applying this strategy as long as the game remains in $T_0$, \PZ guarantees her budget remains above $\fpTh_0(Z_1,\dots,Z_d)$. Consequently, either the game remains in $T_0$ forever (satisfying internal win), or it eventually reaches some $u \notin T_0$ (i.e., $u \in T_{>0}$). In the latter case, her budget upon reaching $u$ will strictly exceed $\fpTh_0(Z_1,\dots,Z_d)(u) = Z_{C(u)}(u)$, satisfying the escape condition.

    \item \textbf{(Outer Induction Step)} Suppose the lemma holds for $k-1$. We show it holds for $k$. First we show $\fpTh_k$ is well-defined. We do case analysis on the parity of $k$:
          \begin{itemize}
            \item Suppose $k$ is even. For $i=0$, $Y^k_0(v) \leq Y^k_1(v)$ trivially because $Y^k_0=0$ on $T_{\leq k}$ and $Y^k_0(v) = Y^k_1(v) = Z_{C(v)}(v)$ on $T_{>k}$. Assuming $Y^k_{i-2} \leq Y^k_{i-1}$, since $\fpTh_{k-1}$ is non-decreasing in its arguments, $\fpTh_{k-1}(Y^k_{i-1}|_{T_k}, \dots) \geq \fpTh_{k-1}(Y^k_{i-2}|_{T_k}, \dots)$. Monotonicity of $\bellman$ then implies $Y^k_i \geq Y^k_{i-1}$. Thus $Y^k_i$ is a non-decreasing, bounded sequence and its limit $\fpTh_k$ is well-defined.
            \item Suppose $k$ is odd. An analogous argument using $Y^k_0 \geq Y^k_1$ (since $Y^k_0=1$ on $T_{\leq k}$) shows $Y^k_i$ is a non-increasing sequence, so $\fpTh_k$ is well-defined.
          \end{itemize}

          Now, suppose \PZ has budget $B_\eve = \fpTh_k(Z_{k+1}, \dots, Z_d)(v) + \epsilon$ at vertex $v$. We analyze the parity of $k$:
          \begin{itemize}
            \item \textbf{(Even $k$)} If $v \in T_{\leq k}$, her budget exceeds $\bellman^{T_{\neq k}}(\fpTh_{k-1}(\fpTh_{k}|_{T_k}, \dots))(v)$. By Lemma~\ref{lem:bellman:property}, \PZ can play to ensure that in the next step, her budget strictly exceeds $\fpTh_{k-1}(\fpTh_k|_{T_k}, \dots)(u)$ for the reached vertex $u$. If $u \in T_{<k}$, she uses the strategy given by the outer induction hypothesis to ensure she either: (1) wins in $T_{<k}$, (2) visits a vertex $w \in T_{>k}$ with budget $>Z_{C(w)}(w)$, or (3) visits a vertex $w \in T_k$ with budget $>\fpTh_k(w)$. If she repeatedly visits $T_k$, returning to $T_k$ with enough budget allows her to continue this process inductively. If either (1) or (2) happens along the way, the lemma is satisfied. Otherwise, the game visits $T_k$ infinitely often without ever visiting $T_{>k}$. Since $k$ is the highest parity visited infinitely often, and $k$ is even, \PZ wins the parity specification inside $T_{\leq k}$.
            \item \textbf{(Odd $k$)} Since $Y^k_i \downarrow \fpTh_k$, \PZ's budget $B_\eve > Y^k_I(v)$ for some $I$. A similar argument shows that \PZ uses the strategy for $Y^k_i$ to ensure reaching $T_{>k}$ with $>Z_{C(w)}(w)$, winning in $T_{<k}$, or visiting $T_k$ at most $I$ times. In all cases, she successfully satisfies the specification by eventually ceasing visits to the odd parity $k$ or escaping.
          \end{itemize}
  \end{itemize}
\end{proof}

\thmPthCorrectness*
\begin{proof}
  If \PZ's budget strictly exceeds $\pTh(v) = \fpTh_d^C()(v)$, Lemma~\ref{lem:fpth:correctness} implies she has a winning strategy because for the maximum priority level $d$, $T_{>d} = \emptyset$. Consequently, the escape condition (ii) is impossible, forcing her to an internal win (condition (i)) which satisfies the parity specification.
  Otherwise, if \PZ's budget is strictly less than $\pTh(v)$, it implies \PO's budget exceeds $1 - \pTh(v)$. Let $C'(v) = C(v)+1$ be an inverted coloring function. Swapping the roles of the players demonstrates \PO plays an identical zero-sum game acting as \PZ, aiming to satisfy the parity specification using $C'$. As established in the text above, $\fpTh^C_d() = 1 - \fpTh^{C'}_{d+1}()$. Therefore, \PO's budget $1 - \pTh(v) = 1 - \fpTh^C_d()(v) = \fpTh^{C'}_{d+1}()(v)$. Because his active budget exceeds the threshold for his winning objective (achieving an even maximum recurring priority under $C'$, corresponding to an odd failure priority under $C$), \PO inherits a winning strategy by consequence of Lemma~\ref{lem:fpth:correctness}.
\end{proof}

\begin{restatable}[Knaster-Tarski Duality]{lemma}{lemTarskiDuality} \label{lem:tarski-duality}
  Let $\mathcal{L} = ([0, 1]^\V, \le)$ be a complete lattice and $F \colon \mathcal{L} \to \mathcal{L}$ be a monotonic operator. For its exact dual operator defined as $F^d(X) = \mathbf{1} - F(\mathbf{1} - X)$, the least ($\mu$) and greatest ($\nu$) fixed-points satisfy:
  \begin{itemize}
    \item $\mu Y \!.\, F^d(Y) = \mathbf{1} - \nu X \!.\, F(X)$
    \item $\nu Y \!.\, F^d(Y) = \mathbf{1} - \mu X \!.\, F(X)$
  \end{itemize}
\end{restatable}
\begin{proof}
  Direct result of the Knaster-Tarski's fixed-point theorem~\cite{tarski1955lattice}.
\end{proof}

\thmParityThresholdSSG*
\begin{proof}
  By definition, $\pTh = \Theta_d X_{T_d} \dots \Theta_0 X_{T_0} \bellman^\varnothing(X)$ and $\optval = \Omega_d X_{T_d} \dots \Omega_0 X_{T_0} \bellmanSSG^\varnothing(X)$. Because $\bellman^S(X)$ and $\bellmanSSG^S(X)$ are duals (Lemma~\ref{lem:bellman:property}), applying the Knaster-Tarski duality (Lemma~\ref{lem:tarski-duality}) exchanges each fixed-point type under the operation $X \mapsto \mathbf{1} - X$. For every priority level $i \in \{0, \dots, d\}$, the operators satisfy $\Theta_i \in \{\mu, \nu\}$ while $\Omega_i$ is the opposing operation in $\{\nu, \mu\}$. Modifying the innermost function via the complement map propagates these identities outward. Applying this from level $0$ to level $d$, we obtain $\pTh = \mathbf{1} - \optval$.
\end{proof}

\section{Omitted Proofs of Section~\ref{sec:mean payoff}}\label{app:sec:mean payoff proofs}

\begin{theorem}[Generalization of Theorem 4.2. from \cite{avni2019infinite}]
  \label{the:meanpayoff:scc-mec}
  Consider an \SGBG with mean payoff specification $(G, \meanpayoff^{R}_{\geq r})$ where $\game = (\V, \VZ, \VO, \VB, E)$ for which there exists $c \in \mathbb{R}$ such that
  \begin{align*}
  \sup_{\sigma_\eve \in \Sigma_\eve}\inf_{\sigma_\adam \in \Sigma_\adam}\mathbb{E}^{\sigma_\eve, \sigma_\adam}[\limsup_{k \to \infty} \frac{1}{k}\sum_{t=0}^{k-1} R(v^t) ~|~ v] = c
  \end{align*}
  for all $v \in \V$ (that is, in the corresponding SSG, the maximal expected mean payoff Eve can achieve against any Adam's strategy is the same in every vertex).
  Then, the threshold exists and
  \begin{align*}
  \thresh_{\game, \meanpayoff^R_{\geq r}}(v) =
    \begin{cases}
      0 & \text{if }c \geq r \\
      1 & \text{if }c < r.
    \end{cases}
  \end{align*}
\end{theorem}
\begin{proof}
  The main difference compared to the original statement is that they considered the arena is a strongly connected component without any control vertices. Therefore, we have to reformulate and prove every lemma that uses the assumption of strongly connected components and extend the strategies to consider control vertices.

  We define the following notions:
  \begin{definition}[Definition 4.5. from \cite{avni2019infinite}]
    Let $\game$ be a \SGBG satisfying the condition above for some $c$. We define a potential of a vertex using the successor vertices as
    \begin{align*}
      Po(v) = R(v) - c +
      \begin{cases}
        \max_{v' \in E(v)}Po(v')    & \text{if } v \in \VZ                    \\
        \min_{v' \in E(v)}Po(v')    & \text{if } v \in \VO                    \\
        \frac{Po(v_1) + Po(v_2)}{2} & \text{if } v \in \VB; v_1, v_2 \in E(v)
      \end{cases}
    \end{align*}
    For a bidding vertex $v \in \VB$, let $v^+$ and $v^-$ refer to the successors of $v$ with the higher and lower potential, respectively, i.e., $Po(v^-) \leq Po(v) \leq Po(v^+)$. Then,
    \begin{align*}
      St(v) =
      \begin{cases}
        0                           & \text{if } v \in \VZ \cup \VO \\
        \frac{Po(v^+) - Po(v^-)}{2} & \text{if } v \in \VB
      \end{cases}
    \end{align*}
    Let $\rho = v^0v^1\cdots v^k$ be a finite path in $\game$. We define energy, gains, and investments of $\rho.$ Energy is $E(\rho) = \sum_{t=0}^{k-1} R(v^t)$. Let $W(\rho) \subseteq \{0, \ldots, k\}$ refer to the indices $t$ where $v^t \in \VB$ and Eve wins the bidding. Similarly, let $L(\rho) \subseteq \{0, \ldots, k\}$ refer to the indices $t$ such that $v^t \in \VB$ and Adam wins the bidding. We define gains as $G(\rho) = \sum_{t \in L(\rho)}St(v^t)$ and investments as $I(\rho) = \sum_{t \in W(\rho)}St(v^t).$
  \end{definition}

 It was shown in \cite{DBLP:journals/corr/abs-1208-0446} that the potential function exists and there exist optimal\footnote{Optimal with respect to the expected payoff specifications.} positional strategies $\sigma_\eve, \sigma_\adam$ for Eve and Adam in $\SSG(\game)$ such that $\sigma_\eve(v) = \argmax_{v' \in E(v)}Po(v')$ and $\sigma_\adam(v) = \argmin_{v' \in E(v)}Po(v')$.
  \begin{lemma}[Generalisation of Lemma 4.6. from \cite{avni2019infinite}]
    Consider a \SGBG as described in the statement of the theorem such that $c = 0$. Let $\rho = v^0v^1\cdots v^k$ be a finite path from $v$ to $u$ where Eve follows the strategy $\sigma_\eve$. Then,
    \begin{align*}
    Po(v) - Po(u) \leq E(\rho) + G(\rho) - I(\rho).
    \end{align*}
  \end{lemma}
  \begin{proof}
    We prove it by an induction on the length of the play. For length 0, the claim is trivial as both sides are equal to 0. Let us have a play $\rho = v^0v^1\cdots v^k$ of length $k+1$.

    If $v^0 \in \VZ$, then $Po(v^0) - Po(v^1) = R(v^0) + \max_{v' \in E(v^0)}Po(v') - Po(v^1) = R(v^0)$ where the last equality follows from the assumption on Eve's strategy. Furthermore, $E(\rho) = E(v^1\cdots v^k) + R(v^0)$, $G(\rho) = G(v^1\cdots v^k)$, and $I(\rho) = I(v^1\cdots v^k).$ Putting everything together with the inductive assumption $Po(v^1) - Po(u) \leq E(v^1\cdots v^k) + G(v^1\cdots v^k) - I(v^1\cdots v^k)$, we get the desired inequality.

    If $v^0 \in \VO$, then $E(\rho) = E(v^1\cdots v^k) + R(v^0)$, $G(\rho) = G(v^1\cdots v^k)$, and $I(\rho) = I(v^1\cdots v^k)$ as above. Also,
    $Po(v^0) - Po(v^1) = \min_{v' \in E(v)}Po(v') + R(v^0)  - Po(v^1) \leq R(v^0)$ where the last equality follows from the fact that $v^1 \in E(v)$. Putting everything together, we get the desired inequality.

    If $v^0 \in \VB$, then we have two options. Let us assume Eve wins the first bidding. Then, $E(\rho) = E(v^1\cdots v^k) + R(v^0)$, $G(\rho) = G(v^1\cdots v^k)$, and $I(\rho) = I(v^1\cdots v^k)+ St(v^0)$. Furthermore, $Po(v^0) - Po(v^1) = \frac 12 Po((v^0)^+) + \frac 12 Po((v^0)^-) + R(v^0) - Po(v^1) = R(v^0) - \frac 12 Po((v^0)^+) + \frac 12 Po((v^0)^-) = R(v^0) - St(v^0)$ where the equality follows from $v^1 = (v^0)^+$.

    If $v^0 \in \VB$ and Adam wins the first bidding, then $E(\rho) = E(v^1\cdots v^k) + R(v^0)$, $G(\rho) = G(v^1\cdots v^k) + St(v^0)$, and $I(\rho) = I(v^1\cdots v^k)$. Furthermore, $Po(v^0) - Po(v^1) = \frac 12 Po((v^0)^+) + \frac 12 Po((v^0)^-) + R(v^0) - Po(v^1) \leq  \frac 12 Po((v^0)^+) + \frac 12 Po((v^0)^-) + R(v^0) - \min_{v' \in E(v^0)}Po(v') = \frac 12 Po((v^0)^+) + \frac 12 Po((v^0)^-) + R(v^0) - Po((v^0)^-)= R(v^0) + \frac 12 Po((v^0)^+) - \frac 12 Po((v^0)^-) = R(v^0) + St(v^0)$.
  \end{proof}

  Let us now assume $c < r$. Then, we use the above lemma in combination with \cite[Lemma 4.11, Lemma 4.21, Lemma 4.22, and Lemma 4.23]{avni2019infinite} to obtain a winning strategy for Adam in $\game$. In bidding vertices, he bids according to the prescribed strategy and he always chooses the successor with the lowest potential.

  Let us assume $c \geq r$. Then, we use the dual statement of the above lemma in combination with \cite[Lemma 4.4. and Lemma 4.9]{avni2019infinite} to obtain a winning strategy for Eve in $\game$. In bidding vertices, she bids according to the prescribed strategy and she always chooses the successor with the highest potential.

  Note that the setting is \cite{avni2019infinite} is dual: there, the objective of Eve is to obtain an overall payoff that is less than or equal to $r$ and the payoff of a path $\rho = v^0v^1\cdots$ is $\liminf_{k \to \infty} \frac{1}{k}\sum_{t=0}^{k-1} R(v^t)$.
\end{proof}

Before we proceed to prove the following theorem, we introduce some new concepts. Assume a \SGBG where all the control vertices are owned by Adam. Then, the arena of this game can be decomposed into maximal end components (MECs).
A MEC is a subgraph $\V' \subseteq \V$ such that Adam has a strategy $\sigma_\adam \in \Sigma_\adam$ in the $\SSG(\game)$ that visits every vertex from $\V'$ almost surely and never leaves the subgraph $\V'$.

For every subgraph $\V' \subseteq \V$ that is a MEC, there exists $c \in \mathbb{R}$ such that
\begin{align*}
\inf_{\sigma_\adam \in \Sigma_\adam}\mathbb{E}^{\sigma_\adam}[\limsup_{k \to \infty} \frac{1}{k}\sum_{t=0}^{k-1} R(v^t) ~|~ v] = c
\end{align*}
for every $ v \in \V'$. Therefore, when we talk about a value of a MEC, we refer to $c$.
The definition also works if the control vertices are only owned by Eve. In that case, we replace infimum with supremum.

\theMeanpayoffMain*
\begin{proof}
  Let us fix $r \in \mathbb{R}$.

  Assume Eve has greater budget than $1 - \val_{\SSG(\game), \meanpayoff^R_{\geq r}}(v)$ in a vertex $v$.
  We claim that Eve can guarantee winning. First, we take an optimal positional strategy $\sigma_\eve$ for Eve from $\SSG(\game)$ and make Eve play this strategy in her control vertices. Effectively, we reduce the \SGBG into an \SGBG where all the control vertices are owned by Adam. We decompose the arena into MECs and identify MECs with the value of at least $r$.
  We claim that if Eve reaches a MEC with value at least $r$, she wins the bidding game.
  If Eve follows the strategy $\sigma_\eve$ in $\SSG(\game)$, then the probability of reaching a MEC with value at least $r$ is at least $\val_{\SSG(\game), \meanpayoff^{R}_{\geq r}}(v)$ regardless of the strategy of Adam. Therefore, Eve implements a reachability strategy to reach one of these MECs and to have nonzero budget on entering them. This is possible as it was shown above in the section for reachability.
  We showed in \Cref{the:meanpayoff:scc-mec} that Eve can win in a MEC with value at least $r$ with nonzero budget. Therefore, Eve has a winning strategy.

  The proof for Adam is analogous. If we assume Eve has budget smaller than $1 - \val_{\SSG(\game), \meanpayoff^R_{\geq r}}(v)$, then Adam can force visiting a MEC with value less than $r$ and having nonzero budget. Then, he can win as it is described by \Cref{the:meanpayoff:scc-mec}.

\end{proof}

\section{Omitted Proofs of Section~\ref{sec:discounted sum}}\label{app:sec:discounted sum proofs}

\theDiscValues*
\begin{proof}

  Let us denote $\theta(v, r) := \lim_{\varepsilon\to 0^+} \lim_{k\to \infty}\Bop^k[\valcand_0](v,r - \varepsilon)$. By $\val^*$, we denote the function $(v, r) \to \val_{\SSG(\game), \discsum^{\lambda, R}_{\geq r}}(v)$ and by $\theta^*$, we denote $(v, r) \to \theta(v, r).$

  The following claims are immediate:
  \begin{itemize}
    \item If $\phi\le\psi$ pointwise then $\Bop[\phi]\le\Bop[\psi]$ pointwise. It follows from the fact that $\Bop$ uses a linear combination, $\min$, and $\max$ that all preserve pointwise order.
    \item The following is trivial as $\valcand_0$ is 0 everywhere: 
    \begin{align*}
    \valcand_0(v,r) \;\le\; \Bop[\valcand_0](v,r)
    \end{align*}
    \item The sequence
          $(\Bop^k[\valcand_0])_{k\ge 0}$ is non-decreasing in $k$ and bounded
          above by $1$, hence, the pointwise limit $\lim_{k \to \infty}\Bop^k[\valcand_0]$
          exists.
    \item From the definition of the value of SSG, we have $\Bop[\val^*]=\val^*$.
    \item Since $\valcand_0 \leq \val^*$ trivially, we have $\Bop^k[\valcand_0] \leq \val^*$ for any $k$ using the above facts.
    \item Moreover, we have $\lim_{k \to \infty}\Bop^k[\valcand_0]\leq \val^*$.
    \item The values of the SSG are left-continuous, i.e., $\lim_{\epsilon \to 0^+} \val_{\SSG(\game), \discsum^{\lambda, R}_{\geq r - \varepsilon}}(v) = \val_{\SSG(\game),\discsum^{\lambda, R}_{\geq r}}(v)$. Let $r \in \mathbb{R}$ and $v \in \V$. Take a sequence $(r_k)_{k \in \mathbb{N}} \in [-\infty, r)$ such that $\lim_{k \to \infty}r_k = r$ and a sequence of strategies of Eve $(\sigma_k)_{k \in \mathbb{N}}$ for which $w_k := \inf_{\sigma_\adam \in \Sigma_\adam} P^{\sigma_k, \sigma_\adam}[\rho \in \discsum_{\geq r^k}^{\lambda, R} ~|~ v] \geq \val_{\SSG(\game), \discsum_{\geq r^k}^{\lambda, R}}(v) - \frac{1}{k}$. That is, $\sigma_k$ are $\frac{1}{k}$-optimal witnesses of the values for specification $\discsum_{\geq r^k}^{\lambda, R}$. Taking $k$ into the limit, the sequence $w_k$ tends to $\lim_{k \to \infty}\val_{\SSG(\game), \discsum_{\geq r^k}^{\lambda, R}}(v)$. We define a new strategy $\sigma_\eve$ and we prove that $\inf_{\sigma_\adam \in \Sigma_\adam} P^{\sigma_\eve, \sigma_\adam}[\rho \in \discsum_{\geq r}^{\lambda, R} ~|~ v] = \lim_{k \to \infty}\val_{\SSG(\game), \discsum_{\geq r^k}^{\lambda, R}}(v)$. That is, we prove $\sigma_\eve$ is a witness that the value for specification $\discsum_{\geq r}^{\lambda, R}$ is equal to the limit of the values from the left.
          We construct $\sigma_\eve$ such that for every $K \in \mathbb{N}$, there are infinitely many $k \in \mathbb{N}$ such that $\sigma_\eve$ and $\sigma_k$ follow the same strategy for the first $K$ steps.
          This is possible since there are only finitely many strategies for a finite-horizon game of length $K$. Now, if $\inf_{\sigma_\adam \in \Sigma_\adam} P^{\sigma_\eve, \sigma_\adam}[\rho \in \discsum_{\geq r}^{\lambda, R} ~|~ v] < \lim_{k \to \infty}\val_{\SSG(\game), \discsum_{\geq r^k}^{\lambda, R}}(v)$, then there exists $k$ such that strategy $\sigma_k$ is not optimal enough for the specification $\discsum_{\geq r^k}^{\lambda, R}$ which is a contradiction with what we assumed about $\sigma_k$.
    \item Since $\val^*$ is closed under limits from the left, we have $\theta^* \leq \val^*$.
  \end{itemize}

  It remains to show that $\theta^* \geq \val^*$. For every $v \in \V, r \in \mathbb{R}$, we construct a strategy for Adam $\sigma_\adam$ such that for any strategy of Eve $\sigma_\eve$, the probability of obtaining a payoff of at least $r$ is less than or equal to $\theta(v, r)$, i.e., $\theta(v, r) \geq P^{\sigma_\eve, \sigma_\adam}[\rho \in \discsum^{\lambda, R}_{\geq r} ~|~ v]$.

  Let $\varepsilon > 0$. Let $\rho = v^0v^1\cdots v^k$ be a history of a play. We define inductively: $r^0 = r-\varepsilon$ and $r^{n+1} = \frac{r^n - R(v^n)}{\lambda}$. Intuitively, $r^n$ represents the amount of remaining payoff Eve needs to collect in order to win.
  Define the strategy
  $\sigma_\adam^{(r-\varepsilon)}$ for Adam as follows:
  at each time~$t$ with $v^t\in\VO$, play
  \[
    u^*_t\in\operatorname*{arg\,min}_{u\in E(v^t)} \theta(u,r^{t+1}).
  \]
  %The key idea is to fix $\varepsilon0$ and run a supermartingale argument at the \emph{shifted} threshold $r-\varepsilon$; the left-continuous regularisation then closes the gap as $\varepsilon\to 0^+$.
  \iffalse
    \medskip\noindent\textbf{Remaining-threshold process.}
    Fix $v_0\in V$, threshold $r\in\mathbb{R}$, and $\varepsilon>0$.
    Set $\rho_0=r-\varepsilon$ and define recursively along any play
    $\rho=v^0v^1\cdots$:
    \begin{equation}\label{eq:rprocess}
      \rho_{t+1}=\frac{\rho_t-R(v^t)}{\lambda},
    \end{equation}
    By~\eqref{eq:telescope}?? and induction,
    \begin{equation}\label{eq:ridentity}
      \rho_t = \DS_t(\rho)+\frac{(r-\varepsilon)-\DS(\rho)}{\lambda^t},
    \end{equation}
    which gives us
    \begin{equation}\label{eq:equiv}
      \DS(\rho)\ge r-\varepsilon \;\iff\; \DS_t(\rho)\ge\rho_t
      \quad\forall\,t\ge 0.
    \end{equation}
  \fi

  For any $\PZ$ strategy $\sigma_\eve$, the joint strategy $\sigma = (\sigma_\eve,\sigma_\adam^{(r-\varepsilon)})$ induces a sequence of vertices $v^0, v^1, \ldots$. We define $M^t = \theta(v^t, r^t)$ where $r^t$ is defined as above. Let $\mathcal{H}^t = (v^0, v^1, \ldots, v^t)$ be a history of a play.
  We prove that
  \begin{align*}
  \E[M^{t+1}~|~ \mathcal{H}^t] \leq M^t.
  \end{align*}
  Since $\Bop[\theta^*]=\theta^*$, we have the following at each step:
  \begin{itemize}
    \item If $v^t\in\VO$: $\sigma_\adam^{(r-\varepsilon)}$ plays the minimizer, so
          $\mathbb{E}[M^{t+1}\mid\mathcal{H}^t]
            =\min_{u\in E(v^t)}\theta(u,r^{t+1})
            =\Bop[\theta(v^t,r^t)] = \theta(v^t,r^t)=M^t$.
    \item If $v^t\in\VZ$: $\sigma_\eve$ plays some $u_t\in E(v^t)$.
          Since $\theta(v^t,r^t)=\max_{u}\theta(u,r^{t+1})$,
          \[
            \mathbb{E}[M^{t+1}\mid\mathcal{H}^t]
            =\theta(u_t,r^{t+1})
            \le\max_{u}\theta(u,r^{t+1})
            =\Bop[\theta(v^t,r^t)] = \theta(v^t,r^t)=M^t.
          \]
    \item If $v^t\in\VR$ with successors $v_1,v_2$:
          $\mathbb{E}[M^{t+1}\mid\mathcal{H}^t]
            =\tfrac12 \theta(v_1,r^{t+1})+\tfrac12 \theta(v_2,r^{t+1})
            =\Bop[\theta(v^t,r^t)] = \theta(v^t,r^t)=M^t$.
  \end{itemize}
  In every case, $\mathbb{E}[M^{t+1}\mid\mathcal{H}^t]\le M^t$,
  so $(M^t)_{t\ge 0}$ is a $[0,1]$-valued supermartingale.
  Therefore, we have
  $M^0\ge\mathbb{E}[M^t]$ for every $t\ge 0$.
  We decompose:
  \[
    \mathbb{E}[M^t]
    =\mathbb{E}[M^t]\,P[\rho \in \discsum^{\lambda, R}_{> r - \varepsilon}\  |\  v^0]
    +\mathbb{E}[M^t]\,P[\rho \in \discsum^{\lambda, R}_{= r - \varepsilon}\  |\  v^0]
    +\mathbb{E}[M^t]\,P[\rho \in \discsum^{\lambda, R}_{< r - \varepsilon}\  |\  v^0].
  \]
  If $\rho \in \discsum^{\lambda, R}_{> r - \varepsilon}(v^0)$, then $r^t \to -\infty$ so eventually $\theta(v^t, r^t) = 1$ and $M^t = 1$.
  If $\rho \in \discsum^{\lambda, R}_{< r - \varepsilon}(v^0)$, then $r^t \to +\infty$ so eventually, $\theta(v^t, r^t) = 0$ and $M^t = 0$.
  Hence,
  \[
    \lim_{t\to\infty}\mathbb{E}[M_t]P[\rho \in \discsum^{\lambda, R}_{> r - \varepsilon}\  |\  v^0]
    =P[\rho \in \discsum^{\lambda, R}_{> r - \varepsilon}\  |\  v^0],
    \qquad
    \lim_{t\to\infty}\mathbb{E}[M_t]P[\rho \in \discsum^{\lambda, R}_{< r - \varepsilon}\  |\  v^0]=0.
  \]
  Taking $t\to\infty$, since $M^t \in [0, 1]$,
  \begin{align*}
    \theta(v^0,r-\varepsilon)
    =M^0
     & \ge P[\rho \in \discsum^{\lambda, R}_{> r - \varepsilon}\  |\  v^0] + \E[M^t]P[\rho \in \discsum^{\lambda, R}_{= r - \varepsilon}\  |\  v^0] \\
     & \geq
    P[\rho \in \discsum^{\lambda, R}_{> r - \varepsilon}\  |\  v]^0 \\
     & \geq P[\rho \in \discsum^{\lambda, R}_{\geq r}\  |\  v^0]
  \end{align*}
  where the last inequality follows from the fact that the event $\discsum^{\lambda, R}_{\geq r}(v)$ is a subset
  of $\discsum^{\lambda, R}_{> r - \varepsilon}(v)$.
  This holds for every initial vertex $v$ and every strategy of Eve $\sigma_\eve$
  (the strategy $\sigma_\adam^{(r-\varepsilon)}$ is fixed and does not depend on
  $\sigma_\eve$), so taking the supremum over $\sigma_\eve$ and the infimum over all strategies of Adam:
  \begin{align*}
    \theta(v,r-\varepsilon)
    \; & \ge\;\sup_{\sigma_\eve}P_{\sigma_\eve,\sigma_\adam^{(r-\varepsilon)}}[\rho \in \discsum^{\lambda, R}_{\geq r} ~|~ v] \\
       & \geq \inf_{\sigma_\adam}\sup_{\sigma_\eve}P_{\sigma_\eve,\sigma_\adam^{r}}[\rho \in \discsum^{\lambda, R}_{\ge r} ~|~ v] = \val_{\SSG(\game), \discsum^{\lambda, R}_{\geq r}}(v).
  \end{align*}
  This inequality holds for every $\varepsilon>0$. Since $\theta^*$ is closed under limits from the left, we obtain
  \[
    \theta(v,r)
    =\lim_{\varepsilon\to 0^+}\theta(v,r-\varepsilon)
    \;\ge\;\val_{\SSG(\game), \discsum^{\lambda, R}_{\geq r}}(v).
  \]
\end{proof}

\begin{example}\label{ex:left continuous closure is needed}
  Consider a single-state discounted-sum SSG with discount $\lambda$ and a single action with payoff $1$. Let $c = \frac{1}{1 - \lambda}$. The Bellman operator is
  $\Bop[\psi](v, r) = 1$ when $r \leq 1$, and
  $\Bop[\psi](v, r) = \psi(v, (r - 1)/\lambda)$ otherwise.
  Notice that $\Bop^k[\psi](v, c) = 0$ for all $k \in \mathbb{N}$. Therefore, $\lim_{k \to \infty}\Bop^k[\psi](v, c) = 0$. However, $\val_{\SSG(\game), \discsum^{\lambda, R}_{\geq c}}(v) = 1$ since the payoff of a play is $c$ with probability 1.
\end{example}

\theDiscConnection*
\begin{proof}
  Let us fix $r \in \R$. If Eve has a budget greater than $1 - \val_{\SSG(\game), \discsum^{\lambda, R}_{\geq r}}(v)$ budget in a vertex $v$, she can win. We construct the strategy as follows.
  Let $\rho = v^0v^1\cdots v^k$ be a finite history of the play and inductively define $r^0 = r$ and $r^{n+1} = \frac{r^n - R(v^n)}{\lambda}$ for any $n \in \mathbb{N}$. In step $k$, the strategy of Eve is as follows:
  In $v^k \in \VZ$, she moves the game to a vertex $\argmax_{v' \in E(v^k)}\val_{\SSG(\game), \discsum^{\lambda, R}_{\geq r^{k+1}}}(v')$. For bidding vertices, Eve bids $\frac{1}{2}\val_{\SSG(\game), \discsum^{\lambda, R}_{\geq r^{k+1}}}(u_1) + \frac{1}{2}\val_{\SSG(\game), \discsum^{\lambda, R}_{\geq r^{k+1}}}(u_2)$ where $u_1, u_2$ are successors of $v^k$. If she wins, she moves the token to the vertex $\argmax_{v' \in \{v_1, v_2\}} \val_{\SSG(\game), \discsum^{\lambda, R}_{\geq r^{k+1}}}(v')$. Notice that the strategy depends on $r$ and requires memory.

  It is easy to verify that the following invariant is maintained throughout the play (independent from the strategy of Adam):
  For any $k$, Eve has a greater budget than $1 - \val_{\SSG(\game), \discsum^{\lambda, R}_{\geq r^k}}(v^k)$.

  Now, we prove that she actually wins the game with this strategy, i.e., she accumulates payoff at least $r$. For the sake of contradiction, assume that there exists a strategy of Adam such that Eve only accumulates payoff $r - \delta$ where $\delta > 0$.
  There exists sufficiently large $K$ such that
  $\frac{R_{\max}\lambda^K}{1 - \lambda} < \delta$. Intuitively, if Adam follows his strategy for $K$ turns, then he wins the game regardless of what happens afterwards since the gap between the target value and the accumulated payoff is too big to be recovered (we assume the weights $R$ are nonnegative which is possible since shifting all weights by $c$ shifts the overall payoff of every infinite sequence by $\frac{c}{1 - \lambda}$).
  Hence, $\val_{\SSG(\game), \discsum^{\lambda, R}_{\geq r^K}}(v^K) = 0$. However, this is a contradiction with the invariant we proved above. Therefore, Adam cannot force the overall payoff to be less than $r$ which proves Eve wins.

  Consider the situation where Eve has a budget smaller than $1 - \val_{\SSG(\game), \discsum^{\lambda, R}_{\geq r}}(v)$. The situation is slightly different here as Adam cannot always take the dual strategy to the one described above. Below, we explain why. Instead, take $\varepsilon > 0$ such that Eve has a budget smaller than $1 - \val_{\SSG(\game), \discsum^{\lambda, R}_{\geq r - \varepsilon}}(v)$. This is always possible since $\val^*$ is left-continuous (see \Cref{the:disc:values}).

  The strategy of Adam is as follows: Let $\rho = v^0v^1\cdots v^k$ be a finite history of the play and inductively define $r^0 = r - \varepsilon$ (the $\varepsilon$ is crucial here) and $r^{n+1} = \frac{r^n - R(v^n)}{\lambda}$ for any $n\in \mathbb{N}$. In step $k$: if $v^k \in \VO$, he moves the game to the vertex $\argmin_{v' \in E(v^k)}\val_{\SSG(\game), \discsum^{\lambda, R}_{\geq r^{k+1}}}(v')$; if $v^k \in \VB$, he bids $\frac{1}{2}\val_{\SSG(\game), \discsum^{\lambda, R}_{\geq r^{k+1}}}(u_1) + \frac{1}{2}\val_{\SSG(\game), \discsum^{\lambda, R}_{\geq r^{k+1}}}(u_2)$ where $u_1, u_2$ are successors of $v^k$. If he wins, he moves the vertex to the vertex $\argmin_{v' \in \{v_1, v_2\}} \val_{\SSG(\game), \discsum^{\lambda, R}_{\geq r^{k+1}}}(v')$.
  Similarly, as above, one can prove that by following this strategy, Adam maintains the following invariant:
  For any $k$, Eve has strictly smaller budget than $1 - \val_{\SSG(\game), \discsum^{\lambda, R}_{\geq r^k}}(v^k)$.
  Similarly as above, we can prove that if she gains a payoff $r$, it would violate the invariant. Notice that if $r^0 = r$, this contradiction could not be proved as it would be possible that Eve gains the payoff $r$ in the limit.
\end{proof}

\begin{example}
  \label{ex:disc:memory}
  \textbf{Discounted-sum specifications do not admit memoryless strategies.}
  Consider a discounted-sum SSG with discount $\lambda = \frac 12$, a single state owned by Eve and two actions: action (a) gives a payoff either $-1$ or $3$ with probability $\frac 12$ and $\frac 12$, respectively; action (b) gives a payoff $0$ with probability $1$.
  The objective of Eve is to maximize the probability of gaining payoff at least $1$.
  Taking action (b) cannot satisfy the specification.
  Instead, taking action (a) has a non-zero value.
  If the payoff outcome is $3$, she switches to play action (b) which guarantees overall payoff $3$.
  If the payoff outcome is $-1$, she plays action (a) again as she needs to compensate for the loss and there is non-zero probability of obtaining an overall payoff at least $1$; playing action (b) cannot compensate for the loss.
  Therefore, the optimal strategy is not memoryless. The corresponding \SGBG exhibits the same behaviour, i.e., chosen actions depend on the bidding outcomes.
\end{example}

\section{Omitted Technical Details of Section~\ref{sec:repair}}
\label{appendix:sec:repair}

\subsection{The Proof of Theorem~\ref{thm:repair:npc}}
\label{appendix:sec:repair np-complete}

\thmRepairNPC*

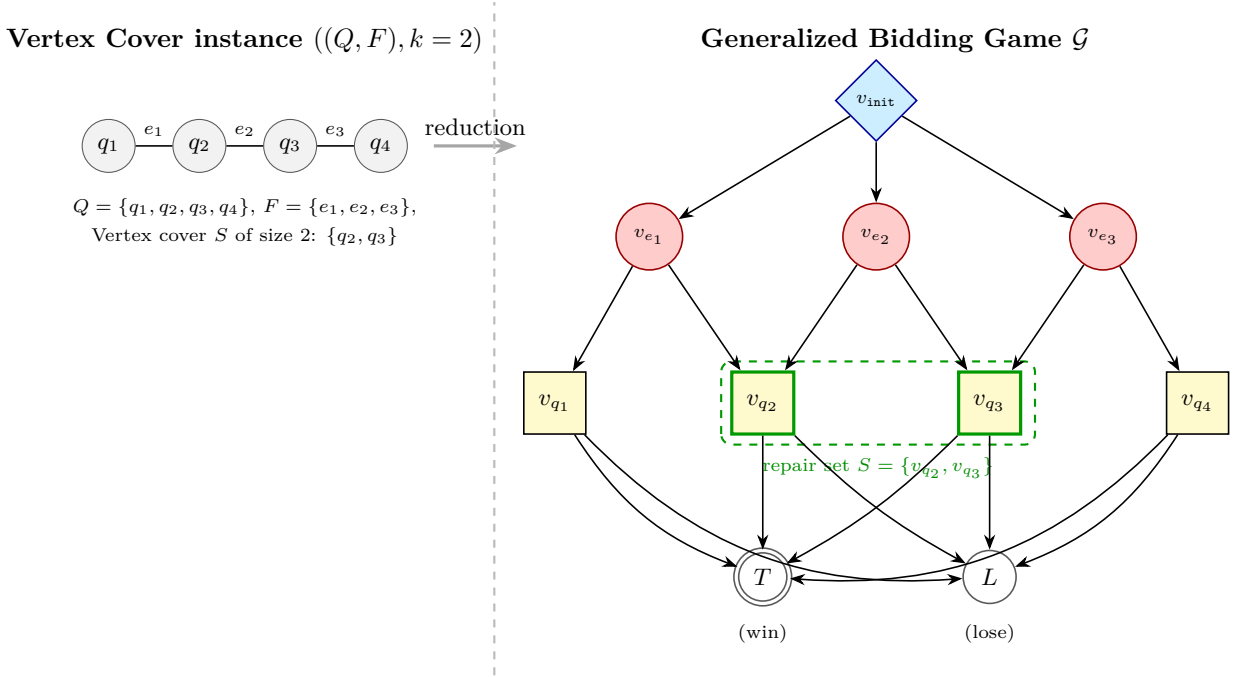
\begin{figure}
  \centering
  \input{FIGURES/repair_hardness_reduction}
  \caption{An illustration of the reduction from the vertex cover problem to the threshold repair problem for generalized bidding games.}
  \label{fig:repair hardness}
\end{figure}

\begin{proof}
  \noindent\textbf{NP-membership:}
  The NP-membership relies the equivalence of these games to simple stochastic games, which follows from Proposition~\ref{prop:equivalence-GBG-SGBG} and Theorem~\ref{thm:equivalence between SGBG and SSG}.
  It is known that in stochastic games with reachability, safety, and parity specifications, memoryless strategies suffice for Eve~\cite[Theorem~1.2]{ChatterjeeJH04}.
  With this observation, a certificate for the repair problem is the set of vertices $|\VZ'|$, and a memoryless policy $\sigma_\eve$ for Eve in $\G'$: it is easy to see that the size of the certificate is polynomial in the size of the input.
  Fixing the strategy $\sigma_\eve$ in $\G'$ gives us a Markov decision process (MDP) $\G'_{\sigma_\eve}$ whose only non-deterministic transitions are due to Adam's moves.
  Using standard linear programming encoding~\cite{modelcheckingbook}, we can compute the value of the initial state of $\G'_{\sigma_\eve}$ in polynomial time and check if it is at most $t$.

  \noindent\textbf{NP-hardness:} We interpret the threshold repair problem as the membership problem for the following set: $\mathit{TH\_REPAIR}\coloneqq \set{((\G,\varphi,v),C,k,\gamma) \mid \exists (\G',\varphi,v) \text{ such that conditions \eqref{cond:repair:1}--\eqref{cond:repair:4} of Problem~\ref{prob:repair} are fulfilled}}$.
  We reduce the NP-complete problem called vertex cover to the threshold repair problem.
  Formally, the vertex cover problem asks: given an undirected graph $(Q,F)$, with vertices $Q$ and edges $F\subseteq Q\times Q$, and an integer $k > 0$, does there exists a set $S\subseteq Q$ with $|S|\leq k$ such that for every edge $(q,r)\in F$, either $q\in S$ or $r\in S$? We say $S$ is a vertex cover of size at most $k$, and write $\mathit{VERTEX\_COVER}\coloneqq \set{((Q,F),k)\mid \exists \text{ a vertex cover of size }\leq k}$.

  Suppose $((Q,F),k)$ is a given instance of the vertex cover problem.
  We define a generalized bidding game $(\G = (\V,\VZ,\VO,\VB,\E),\varphi)$ as follows:
  \begin{itemize}
    \item for every vertex $q\in Q$, introduce a bidding vertex $v_q\in \VB$, as well as introduce a pair of terminal vertices $T$ and $L$ that are bidding vertices, i.e., $\VB \coloneqq \set{v_q\mid q\in Q} \cup \set{T,L}$;
    \item for every edge $e\in F$, introduce an Eve-controlled vertex $v_e\in \VZ$, i.e., $\VZ \coloneqq \set{v_e\mid e\in F}$;
    \item the only Adam-controlled vertex is the initial vertex, i.e., $\VO\coloneqq\set{\vin}$;
    \item the set of edges $\E$ contain only the following:
          \begin{itemize}
            \item for every $e\in F$, $(\vin,v_e)\in \E$,
            \item for every $e = (p,q)\in F$, $(v_e,v_p)\in \E$ and $(v_e,v_q)\in \E$, and
            \item for every $q\in Q$, $(v_q,T)\in \E$ and $(v_q,L)\in \E$.
          \end{itemize}
    \item $\varphi$ is the reachability specification with the target set $\set{T}$, i.e., $\varphi = \reach(\set{T})$.
  \end{itemize}
  An example of the above construction is shown in Figure~\ref{fig:repair hardness}.

  \medskip
  \noindent\textbf{Claim (the reduction and its correctness): $((Q,F),k) \in \mathit{VERTEX\_COVER}$ iff $((\G,\varphi,\vin), \VB, k, \frac{1}{4} ) \in \mathit{TH\_REPAIR}$.}
  Since every bidding vertex $v\in \VB$ is connected to $T$ and $L$, the threshold of $v$ in the original graph $\G$ is $\frac{1}{2}$.
  Now suppose $S$ is a vertex cover of $(Q,F)$ of size at most $k$.
  We show how to repair $\G$ to $\G'$ so that $\thresh_{\G',\reach^{\{T\}}}(\vin)$ is below $\frac{1}{4}$: for every $q\in S$, convert the bidding vertex $v_q$ into an Eve-controlled vertex, i.e., $\VZ'=\VZ\cup \set{v_q\in \VB\mid q\in S}$.
  With this, the threshold of every vertex $v_q$ in the repaired graph $\G'$ becomes $0$, because Eve would simply reach $T$ from $v_q$ and win.
  Now, no matter which $v_e$ (recall, $e\in F$ is an edge in the given graph) Adam picks from $\vin$, from $v_e$, Eve can always pick a successor $v_q$ with $q\in S$, which belongs to her and the threshold is $0$.
  Therefore, the threshold of $(\G',\varphi)$ at $\vin$ is $0 \leq  \frac{1}{4}$.

  For the other direction, assume that $((Q,F),k) \notin \mathit{VERTEX\_COVER}$.
  It follows from the construction that, no matter which $k$ bidding vertices are converted to Eve's vertices in the repaired graph $\G'$, Adam is able to pick $v_e$ such that both successors of $v_e$ are bidding vertices in $\G'$, and therefore the threshold of $v_e$ is $\frac{1}{2}$, implying threshold of $\vin$ is also $\frac{1}{2}$.
  Therefore, $((\G,\varphi,\vin), \VB, k, \frac{1}{4} ) \notin \mathit{TH\_REPAIR}$.

  \medskip
  \noindent\textbf{Claim: the reduction can be computed in polynomial time.}
  From the construction, it follows that $|\V| = |Q|+2+|F|+1$, implying that the reduction is polynomial time in the size of the input to the vertex cover problem.
\end{proof}

\subsection{The MILP Encoding of the Repair Problem}
\label{appendix:sec:milp encoding of repair}

We encode the repair problem as a Mixed-Integer Linear Program (MILP), and for simplicity, we only consider the reachability specification, but the same idea can be used to extend it to safety and parity specifications.
Recall the setting: we are given a generalized bidding game
$G = (V, \VZ, \VO, \VB, \E)$ with a reachability specification
$\reach^T$, a constraint set $C \subseteq \VB$ containing all bidding vertices that may be repaired, a repair budget $k \in \mathbb{N}$, and a target threshold
$\gamma \in (0,1)$.
We wish to decide whether there exists a repair set
$S \subseteq C$ with $|S| \leq k$ such that, after converting
every vertex in $S$ to an Eve-controlled vertex, Eve's threshold budget at
$v$ is at most $\gamma$.

By Theorem~\ref{thm:reach:ssg}, this threshold equals $1 - \val_{\SSG(\game),\reach^T}(v)$ in the
corresponding simple stochastic game $\SSG(G)$, where bidding vertices
are treated as random vertices.
Hence, the repair target $\rTh(\vin) \leq \gamma$ is equivalent to
$\val_{\SSG(\G),\reach^T}(v) \geq 1 - \gamma$.

We present two encodings. First, the simpler case when
\emph{every} vertex has out-degree exactly~$2$ (which holds for all bidding
vertices by assumption, and is assumed for control vertices here for
simplicity), and second, the arbitrary out-degree case that lifts the two out-degree case to the general setting.

% =============================================================================
\subsubsection*{The Binary Out-Degree Case}
\label{sec:binary}
% =============================================================================

Assume every vertex $v \in V$ has exactly two successors,
denoted $v^+ $ and $v^{++}$.

\bigskip
\noindent\textbf{VARIABLES:}

\smallskip\noindent\textbf{Repair decisions.}
For each candidate bidding vertex $b \in C$:
\[
  x_b \;\in\; \{0,1\},
  \qquad
  x_b = 1 \iff \text{vertex } b \text{ is converted to Eve control.}
\]
\smallskip\noindent\textbf{SSG value function.}
For each vertex $v \in V$:
\[
  \mu_v \;\in\; [0,1],
  \qquad
  \text{representing } \val(v) \text{ after the repair.}
\]
\smallskip\noindent\textbf{Strategy selectors for Eve.}
For each $v \in \VZ$, and for each $v \in C$ (whose type depends
on $x_v$):
\[
  s_v \;\in\; \{0,1\},
  \qquad
  s_v = 0 \iff \text{Eve chooses successor } v^+;
  \quad
  s_v = 1 \iff \text{Eve chooses successor } v^{++}.
\]
\smallskip\noindent\textbf{Strategy selectors for Adam.}
For each $v \in \VO$:
\[
  r_v \;\in\; \{0,1\},
  \qquad
  r_v = 0 \iff \text{Adam chooses successor } v^+;
  \quad
  r_v = 1 \iff \text{Adam chooses successor } v^{++}.
\]
\smallskip\noindent\textbf{Auxiliary linearisation variables.}
For each $b \in C$, to linearise the product $x_b \cdot e_b$
(where $e_b$ is the Eve-control value at $b$):
\[
  e_b \;\in\; [0,1], \qquad z_b \;\in\; [0,1].
\]

\bigskip
\noindent\textbf{CONSTRAINTS:}

\smallskip\noindent\textbf{Budget.}
\begin{equation}\label{eq:budget}
  \sum_{b \,\in\, C} x_b \;\leq\; k.
\end{equation}
\smallskip\noindent\textbf{Boundary conditions.}
\begin{equation}\label{eq:boundary}
  \mu_v = 1 \;\;\forall v \in T,
  \qquad
  \mu_v = 0 \;\;\forall v \in L,
\end{equation}
where $L$ is the absorbing losing sink (if no explicit losing sink exists, the
constraint on $L$ is omitted and the lower bound $\mu_v \geq 0$ suffices).\\
\smallskip\noindent\textbf{Eve-controlled vertices.} $i \in \VZ$
(Eve maximises, so $\mu_i$ equals the value of the chosen successor):
\begin{align}
  \mu_i & \leq \mu_{i^+}  + s_i,        \label{eq:eve-ub1} \\
  \mu_i & \leq \mu_{i^{++}} + (1-s_i),  \label{eq:eve-ub2} \\
  \mu_i & \geq \mu_{i^+}  - s_i,        \label{eq:eve-lb1} \\
  \mu_i & \geq \mu_{i^{++}} - (1-s_i).  \label{eq:eve-lb2}
\end{align}
Constraints~\eqref{eq:eve-ub1}--\eqref{eq:eve-lb2} together enforce
$\mu_i = \mu_{i^+}$ when $s_i=0$ and $\mu_i = \mu_{i^{++}}$ when $s_i = 1$,
using big-$M = 1$ (valid since all values lie in $[0,1]$).\\
\smallskip\noindent\textbf{Adam-controlled vertices.} $j \in \VO$
(Adam minimises, so $\mu_j$ equals the value of the chosen successor):
\begin{align}
  \mu_j & \leq \mu_{j^+}  + r_j,          \label{eq:adam-ub1} \\
  \mu_j & \leq \mu_{j^{++}} + (1-r_j),    \label{eq:adam-ub2} \\
  \mu_j & \geq \mu_{j^+}  - r_j,          \label{eq:adam-lb1} \\
  \mu_j & \geq \mu_{j^{++}} - (1-r_j).    \label{eq:adam-lb2}
\end{align}
\smallskip\noindent\textbf{Fixed bidding vertices.} $b \in \VB\setminus C$
(remain random regardless; value is the average of successors):
\begin{equation}\label{eq:fixed-bid}
  \mu_b \;=\; \frac{\mu_{b^+} + \mu_{b^{++}}}{2}.
\end{equation}
\smallskip\noindent\textbf{Repairable bidding vertices} $b \in C$.
When $x_b = 0$ the vertex stays random; when $x_b = 1$ it becomes
Eve-controlled.  We handle the two cases separately and combine via
the McCormick linearization.

\medskip
\noindent\textit{Step 1.} Encode the Eve-control value
$e_b = \max(\mu_{b^+}, \mu_{b^{++}})$ using the selector $s_b$:
\begin{align}
  e_b & \leq \mu_{b^+}    + s_b,       \label{eq:eb-ub1} \\
  e_b & \leq \mu_{b^{++}} + (1-s_b),   \label{eq:eb-ub2} \\
  e_b & \geq \mu_{b^+}    - s_b,       \label{eq:eb-lb1} \\
  e_b & \geq \mu_{b^{++}} - (1-s_b).   \label{eq:eb-lb2}
\end{align}

\noindent\textit{Step 2.} Let $g_b = \tfrac{1}{2}(\mu_{b^+} + \mu_{b^{++}})$
be the random-vertex value.
The combined value is
$\mu_b = x_b \cdot e_b + (1-x_b) \cdot g_b$.
Introduce $z_b = x_b \cdot e_b$ and linearise with McCormick envelopes
(using $e_b \in [0,1]$ and $x_b \in \{0,1\}$):
\begin{align}
  z_b & \geq 0,                    \label{eq:mc1} \\
  z_b & \geq e_b + x_b - 1,        \label{eq:mc2} \\
  z_b & \leq e_b,                  \label{eq:mc3} \\
  z_b & \leq x_b.                  \label{eq:mc4}
\end{align}
Then the value constraint becomes:
\begin{equation}\label{eq:bid-value}
  \mu_b \;=\; z_b \;+\; (1 - x_b)\cdot g_b
  \;=\; z_b + g_b - x_b \cdot g_b.
\end{equation}
Since $g_b = \tfrac{1}{2}(\mu_{b^+} + \mu_{b^{++}})$ is a linear expression in
$\mu$, the term $x_b \cdot g_b$ is a bilinear product that must similarly be
linearised.  Introduce $w_b \in [0,1]$ as a proxy for $x_b \cdot g_b$:
\begin{align}
  w_b & \geq 0,                    \label{eq:mc5} \\
  w_b & \geq g_b + x_b - 1,        \label{eq:mc6} \\
  w_b & \leq g_b,                  \label{eq:mc7} \\
  w_b & \leq x_b.                  \label{eq:mc8}
\end{align}
Equation~\eqref{eq:bid-value} then becomes the fully linear constraint:
\begin{equation}\label{eq:bid-value-lin}
  \mu_b \;=\; z_b + g_b - w_b.
\end{equation}
\smallskip\noindent\textbf{Threshold target.}
\begin{equation}\label{eq:target}
  \mu_{\vin} \;\geq\; 1 - \gamma.
\end{equation}

\smallskip
\noindent\textbf{Objective and Summary}

The MILP can be stated as a \emph{feasibility problem}---check whether the
system \eqref{eq:budget}--\eqref{eq:target} is satisfiable---or equivalently
as an optimisation:
\begin{equation}\label{eq:obj}
  \text{maximise}\quad \mu_{\vin}
  \quad \text{subject to constraints \eqref{eq:budget}--\eqref{eq:bid-value-lin},}
\end{equation}
and then verify whether the optimal value is at least $1-\gamma$.

\begin{remark}[Variable count]
  The MILP has $O(|V|)$ continuous variables $(\mu_v, e_b, z_b, w_b, g_b)$,
  $O(|V|)$ binary variables $(x_b, s_v, r_j)$, and $O(|V|)$ constraints,
  all linear.  It is therefore compact and directly encodable in standard
  solvers (Gurobi, CPLEX, SCIP).
\end{remark}

% =============================================================================
\subsubsection*{The General Out-Degree Case}
\label{sec:general}
% =============================================================================

We now drop the binary out-degree assumption for control vertices.  Let
$\out(v) = \{u_1, \ldots, u_{d_v}\}$ denote the successor set of vertex $v$,
where $d_v = |\out(v)|$ may be arbitrary for $v \in \VZ \cup \VO$, and
$d_b = 2$ is retained for all $b \in \VB$ (as required by the model).

\smallskip
\noindent\textbf{MODIFIED VARIABLES}

The repair variables $x_b$, the value variables $\mu_v$, and all auxiliary
variables for bidding vertices ($e_b, s_b, z_b, w_b, g_b$) remain exactly
as in the two out-degree case described earlier.  Only the \emph{strategy selectors} for
$\VZ$ and $\VO$ change.

\smallskip\noindent\textbf{Binary selection vectors for Eve.}
For each $i \in \VZ$ and each successor $u_\ell \in \out(i)$,
$\ell = 1, \ldots, d_i$:
\[
  y_{i,\ell} \;\in\; \{0,1\},
  \qquad
  y_{i,\ell} = 1 \iff \text{Eve chooses } u_\ell \text{ from } i.
\]
(The single binary $s_i$ of the two out-degree setting is the special case
$y_{i,1} = 1 - s_i$, $y_{i,2} = s_i$.)\\
\smallskip\noindent\textbf{Binary selection vectors for Adam.}
For each $j \in \VO$ and each successor $u_\ell \in \out(j)$,
$\ell = 1, \ldots, d_j$:
\[
  z_{j,\ell} \;\in\; \{0,1\},
  \qquad
  z_{j,\ell} = 1 \iff \text{Adam chooses } u_\ell \text{ from } j.
\]

\bigskip
\noindent\textbf{MODIFIED CONSTRAINTS}

The budget~\eqref{eq:budget}, boundary~\eqref{eq:boundary}, fixed bidding
vertices~\eqref{eq:fixed-bid}, repairable bidding vertices
\eqref{eq:eb-ub1}--\eqref{eq:bid-value-lin}, and threshold
target~\eqref{eq:target} are \emph{unchanged}.  Only the control-vertex
constraints are replaced.

\smallskip\noindent\textbf{Eve-controlled vertices} $i \in \VZ$ with successors
$\{u_1, \ldots, u_{d_i}\}$:

\noindent(a) Exactly one successor is chosen:
\begin{equation}\label{eq:eve-sos}
  \sum_{\ell=1}^{d_i} y_{i,\ell} \;=\; 1.
\end{equation}

\noindent(b) The value at $i$ equals the value of the chosen successor
(big-$M = 1$):
\begin{align}
  \mu_i & \leq \mu_{u_\ell} + (1 - y_{i,\ell})
        & \forall\, \ell = 1,\ldots,d_i, \label{eq:eve-gen-ub} \\
  \mu_i & \geq \mu_{u_\ell} - (1 - y_{i,\ell})
        & \forall\, \ell = 1,\ldots,d_i. \label{eq:eve-gen-lb}
\end{align}
Together with~\eqref{eq:eve-sos}, constraints
\eqref{eq:eve-gen-ub}--\eqref{eq:eve-gen-lb} enforce
$\mu_i = \mu_{u_{\ell^*}}$ for the unique $\ell^*$ with $y_{i,\ell^*} = 1$.

\begin{remark}[Valid inequality]
  Adding $\mu_i \geq \mu_{u_\ell}$ for all $\ell$ (which is valid since Eve
  chooses the best successor) strengthens the LP relaxation without changing
  feasibility, and is recommended for solver performance.
\end{remark}

\smallskip\noindent\textbf{Adam-controlled vertices} $j \in \VO$ with successors
$\{u_1, \ldots, u_{d_j}\}$:

\noindent(a) Exactly one successor is chosen:
\begin{equation}\label{eq:adam-sos}
  \sum_{\ell=1}^{d_j} z_{j,\ell} \;=\; 1.
\end{equation}

\noindent(b) The value at $j$ equals the value of the chosen successor:
\begin{align}
  \mu_j & \leq \mu_{u_\ell} + (1 - z_{j,\ell})
        & \forall\, \ell = 1,\ldots,d_j, \label{eq:adam-gen-ub} \\
  \mu_j & \geq \mu_{u_\ell} - (1 - z_{j,\ell})
        & \forall\, \ell = 1,\ldots,d_j. \label{eq:adam-gen-lb}
\end{align}

\begin{remark}[Valid inequality]
  Adding $\mu_j \leq \mu_{u_\ell}$ for all $\ell$ strengthens the LP
  relaxation from Adam's side.
\end{remark}

\bigskip
\noindent\textbf{VARIABLE COUNT}

\begin{center}
  \begin{tabular}{lll}
    \toprule
    Variable type                          & Scope                           & Count                  \\
    \midrule
    $\mu_v$ (continuous, value)            & all $v \in V$                   & $|V|$                  \\
    $x_b$ (binary, repair)                 & $b \in C$                       & $|C|$                  \\
    $y_{i,\ell}$ (binary, Eve strategy)    & $i \in \VZ$, $\ell \in \out(i)$ & $\sum_{i \in \VZ} d_i$ \\
    $z_{j,\ell}$ (binary, Adam strategy)   & $j \in \VO$, $\ell \in \out(j)$ & $\sum_{j \in \VO} d_j$ \\
    $s_b$ (binary, Eve at repaired bid.)   & $b \in C$                       & $|C|$                  \\
    $e_b, z_b, w_b$ (cont., linearisation) & $b \in C$                       & $3|C|$                 \\
    \bottomrule
  \end{tabular}
\end{center}

The total number of binary variables is $O(|E|)$ (linear in the graph size),
and the total number of constraints is likewise $O(|E|)$.  The MILP thus
scales gracefully with the game size, and can be solved directly by
off-the-shelf solvers.

%% file: FIGURES/tic-tac-toe.tex
\definecolor{biddingcol}{HTML}{FFFACD}
\definecolor{adamcol}   {HTML}{CCEEFF}
\definecolor{evecol}  {HTML}{FFCCCC}
\definecolor{termcol}  {HTML}{CCFFCC}

\tikzset{
  bid/.style={rectangle, draw=black, semithick, fill=biddingcol,
              minimum width=0.78cm, minimum height=0.78cm, inner sep=0pt},
  trm/.style={rectangle, draw=black!55, semithick, fill=biddingcol,
              minimum width=0.78cm, minimum height=0.78cm, inner sep=0pt},
  evn/.style={circle, draw=red!60!black, semithick, fill=evecol,
              minimum size=0.78cm, font=\footnotesize\bfseries, inner sep=1pt},
  adn/.style={diamond, draw=blue!60!black, semithick,
              fill=adamcol, minimum size=0.78cm,
              font=\footnotesize\bfseries, inner sep=1pt},
  arr/.style={->, >=Stealth, semithick},
}

%% ── tiny tic-tac-toe helpers (cell spacing 0.18 cm) ──────────────────────
\newcommand{\tboard}{%
  \draw[black!45, very thin]
    (-0.27,-0.09)--( 0.27,-0.09)  (-0.27, 0.09)--( 0.27, 0.09)
    (-0.09,-0.27)--(-0.09, 0.27)  ( 0.09,-0.27)--( 0.09, 0.27);}

\newcommand{\oo}[2]{% O at col #1, row #2  (each in {-1,0,1})
  \draw[red!65!black, very thin] ({#1*.18},{#2*.18}) circle[radius=0.055];}

\newcommand{\xx}[2]{% X
  \draw[blue!65!black, very thin, line cap=round]
    ({#1*.18-.041},{#2*.18+.041})--({#1*.18+.041},{#2*.18-.041})
    ({#1*.18+.041},{#2*.18+.041})--({#1*.18-.041},{#2*.18-.041});}

\newcommand{\wrow}[1]{% horizontal win line at row #1
  \draw[red!45!black, line width=0.9pt, line cap=round, opacity=0.6]
    (-.25,{#1*.18})--(.25,{#1*.18});}

\newcommand{\wcol}[1]{% vertical win line at col #1
  \draw[blue!45!black, line width=0.9pt, line cap=round, opacity=0.6]
    ({#1*.18},-.25)--({#1*.18},.25);}

%% label placed just below the south tip of a diamond node
\newcommand{\dlabel}[2]{%
  \node[font=\scriptsize, anchor=north, inner sep=0pt]
    at ($(#1.south)+(0,-2.5pt)$) {$#2$};}

\scalebox{0.8}{
\begin{tikzpicture}[x=1cm, y=1cm]

%% ── Col 0: start ─────────────────────────────────────────────────────────
\node[bid] (a) at (0,0) {};
\begin{scope}[shift={(0,0)}]\tboard\end{scope}
\dlabel{a}{a}

%% ── Col 1: control nodes ─────────────────────────────────────────────────
\node[evn] (b) at (2.1, 0.8) {$b$};
\node[adn] (c) at (2.1,-0.8) {$c$};
\draw[arr] (a)--(b);
\draw[arr] (a)--(c);

%% ── Col 2: first-move boards ─────────────────────────────────────────────
%% Eve's successors (O placed)
\node[bid] (d) at (3.5, 0.8) {};
\begin{scope}[shift={(3.5, 0.8)}]\tboard\oo{0}{0}\end{scope}
\dlabel{d}{d}

\node[bid] (e) at (5, 0.8) {};
\begin{scope}[shift={(5,0.8)}]\tboard\oo{1}{1}\end{scope}
\dlabel{e}{e}

\draw[arr] (b)--(d);
\draw[arr] (b) to[bend left=35] (e);

%% Adam's successors (X placed)
\node[bid] (f) at (3.5,-0.8) {};
\begin{scope}[shift={(3.5,-0.8)}]\tboard\xx{0}{0}\end{scope}
\dlabel{f}{f}

\node[bid] (g) at (5,-0.8) {};
\begin{scope}[shift={(5,-0.8)}]\tboard\xx{-1}{1}\end{scope}
\dlabel{g}{g}

\draw[arr] (c)--(f);
\draw[arr] (c) to[bend left=35] (g);

%% more-siblings indicator
\node[font=\small] at (4.4, 0) {$\vdots$};

%% ── Continuation ─────────────────────────────────────────────────────────
\node[font=\large] at (6.1, 0) {$\cdots$};

%% ── Col 3: mid-game bidding node ─────────────────────────────────────────
\node[bid] (h) at (7.5, 0) {};
\begin{scope}[shift={(7.5,0)}]
  \tboard\xx{-1}{0}\oo{1}{1}\xx{-1}{-1}\oo{0}{1}
\end{scope}
\dlabel{h}{h}
\node[font=\large, anchor=south] at (h.north) {{\color{red!70!black}$\theta_h=\frac{\theta_i+\theta_j}{2}=0.5$}};

%% ── Col 4: control nodes ─────────────────────────────────────────────────
\node[evn] (i) at (9.4, 0.8) {$i$};
\node[font=\large, anchor=south] at (i.north) {{\color{red!70!black}$\theta_i=\min_{v\in E(i)}\theta_v=0$}};
\node[adn] (j) at (9.4,-0.8) {$j$};
\node[font=\large, anchor=north] at (j.south) {{\color{red!70!black}$\theta_j=\max_{v\in E(j)}\theta_v=1$}};
\draw[arr] (h)--(i);
\draw[arr] (h)--(j);

%% ── Col 5: terminal nodes ────────────────────────────────────────────────
%% k: Eve wins (O completes top row)
\node[trm] (k) at (11.5, 0.8) {};
\path[arr]  (k) [loop right]    edge    ();
\begin{scope}[shift={(11.5,0.8)}]
  \tboard
  \oo{-1}{1}\oo{0}{1}\oo{1}{1}
  \xx{-1}{0}\xx{-1}{-1}
  \wrow{1}
\end{scope}
\dlabel{k}{k}
\node[font=\large, anchor=south west] at (k.north east) {{\color{red!70!black}$\theta_k=0$}};

%% l: Adam wins (X completes left column)
\node[trm] (l) at (11.5,-0.8) {};
\path[arr]  (l) [loop right]    edge    ();
\begin{scope}[shift={(11.5,-0.8)}]
  \tboard
  \xx{-1}{0}\xx{-1}{1}\xx{-1}{-1}
  \oo{0}{1}\oo{1}{1}
  \wcol{-1}
\end{scope}
\dlabel{l}{l}
\node[font=\large, anchor=north west] at (l.south east) {{\color{red!70!black}$\theta_l=1$}};

\draw[arr] (i)--(k);
\draw[arr] (j)--(l);

%% more-siblings indicator (other successors of i and j not shown)
\node[font=\small] at (11.5, 0) {$\vdots$};

% %% losing sinks of the players
% \node[evn]  (se)    at  (6.5,1.5) {$\mathit{losing}_\eve$}; 
% \node[adn]  (sa)    at  (6.5,-1.5) {$\mathit{losing}_\adam$}; 
% \draw[arr]  (b) to[bend left]   (se);
% \draw[arr]  (i)--(se);
% \path[arr]  (se)    edge[loop below]    ();
% \draw[arr]  (c) to[bend right]   (sa);
% \draw[arr]  (j)--(sa);
% \path[arr]  (sa)    edge[loop above]    ();

\end{tikzpicture}
}

%% file: FIGURES/distinction.tex
\definecolor{biddingcol}{HTML}{FFFACD}
\definecolor{adamcol}   {HTML}{CCEEFF}
\definecolor{evecol}  {HTML}{FFCCCC}

\tikzset{
  nd/.style  = {rectangle, draw=black, semithick,  minimum width=0.78cm, minimum height=0.78cm,
                font=\small\itshape, fill=biddingcol, inner sep=0pt},
  evn/.style={circle, draw=red!60!black, semithick, fill=evecol,
              minimum size=0.78cm, font=\footnotesize\bfseries, inner sep=1pt},
  adn/.style={diamond, draw=blue!60!black, semithick,
              fill=adamcol, minimum size=0.78cm,
              font=\footnotesize\bfseries, inner sep=1pt},
  arr/.style = {->, >=Stealth, semithick},
  rfrac/.style = {font=\scriptsize, text=red!70!black},
  bfrac/.style = {font=\scriptsize, text=blue!60!black},
}

\def\nd{1.2}    % <── node distance: adjust this one value
\def\hw{0.39}   % half node width/height (match nd style)

\newcommand{\tikzflame}{%
  %% ── outer red blob ────────────────────────────────────────────────────────
  \fill[red!80!black]
    (0, 0)
    .. controls (-0.38, 0.02) and (-0.44, 0.22) .. (-0.38, 0.42)
    .. controls (-0.34, 0.56) and (-0.26, 0.64) .. (-0.24, 0.78)
    .. controls (-0.20, 0.68) and (-0.22, 0.56) .. (-0.16, 0.50)
    .. controls (-0.10, 0.62) and (-0.06, 0.76) .. ( 0,    0.92)
    .. controls ( 0.06, 0.76) and ( 0.10, 0.62) .. ( 0.16, 0.50)
    .. controls ( 0.22, 0.56) and ( 0.20, 0.68) .. ( 0.24, 0.78)
    .. controls ( 0.26, 0.64) and ( 0.34, 0.56) .. ( 0.38, 0.42)
    .. controls ( 0.44, 0.22) and ( 0.38, 0.02) .. ( 0,    0)
    -- cycle;
  %% ── white cutout left ─────────────────────────────────────────────────────
  \fill[white]
    (-0.16, 0.50)
    .. controls (-0.26, 0.52) and (-0.32, 0.62) .. (-0.26, 0.74)
    .. controls (-0.20, 0.80) and (-0.14, 0.76) .. (-0.12, 0.70)
    .. controls (-0.08, 0.62) and (-0.10, 0.54) .. (-0.16, 0.50)
    -- cycle;
  %% ── white cutout right ────────────────────────────────────────────────────
  \fill[white]
    ( 0.16, 0.50)
    .. controls ( 0.26, 0.52) and ( 0.32, 0.62) .. ( 0.26, 0.74)
    .. controls ( 0.20, 0.80) and ( 0.14, 0.76) .. ( 0.12, 0.70)
    .. controls ( 0.08, 0.62) and ( 0.10, 0.54) .. ( 0.16, 0.50)
    -- cycle;
  %% ── orange inner flame ────────────────────────────────────────────────────
  \fill[orange!85!red]
    ( 0,    0.04)
    .. controls (-0.26, 0.10) and (-0.30, 0.32) .. (-0.22, 0.52)
    .. controls (-0.14, 0.66) and (-0.06, 0.72) .. ( 0,    0.82)
    .. controls ( 0.06, 0.72) and ( 0.14, 0.66) .. ( 0.22, 0.52)
    .. controls ( 0.30, 0.32) and ( 0.26, 0.10) .. ( 0,    0.04)
    -- cycle;
  %% ── yellow core ───────────────────────────────────────────────────────────
  \fill[yellow!85!orange]
    ( 0,    0.12)
    .. controls (-0.14, 0.22) and (-0.16, 0.40) .. (-0.08, 0.54)
    .. controls (-0.03, 0.64) and ( 0,    0.66) .. ( 0,    0.68)
    .. controls ( 0,    0.66) and ( 0.03, 0.64) .. ( 0.08, 0.54)
    .. controls ( 0.16, 0.40) and ( 0.14, 0.22) .. ( 0,    0.12)
    -- cycle;
}

\def\fs{0.7}   % <── flame scale, adjust to taste

\begin{tikzpicture}
    \node[nd] (a) at (     0,  0) {$p_1$};
    \node[evn] (b) at (     1.2*\nd,   \nd ) {$b$};
    \node[adn] (c) at (     1.2*\nd, -\nd) {$c$};
    \node[nd] (d) at (  2.4*\nd, \nd  ) {$p_2$};
    \node[nd] (e) at (  2.4*\nd, -\nd) {$p_1$};

    \begin{scope}[shift={($(e)+(\hw+0.35, -0.35*\fs)$)}, scale=\fs]
      \tikzflame
    \end{scope}

    \draw[arr]  ($(a)+(-1, 0)$) -- (a);
    \draw[arr] (a) -- (b);
    \draw[arr] (a) -- (c);
    \draw[arr] (b) -- (d);
    \draw[arr] (c) -- (d);
    \draw[arr] (c) -- (e);
    \path[arr]  (d) edge[bend right]    (b);
    \path[arr]  (b) edge[bend right]    (a);
    \path[arr]  (e) edge[bend left=3cm]    (a);
\end{tikzpicture}

%% file: FIGURES/repair_hardness_reduction.tex
%% ── Colour palette (matching auction-based-scheduling.tex) ──────────────────
\definecolor{biddingcol}{HTML}{FFFACD}
\definecolor{adamcol}   {HTML}{CCEEFF}
\definecolor{evecol}    {HTML}{FFCCCC}
 
%% ── Node / edge styles ──────────────────────────────────────────────────────
\tikzset{
  %% bidding vertex: yellow rectangle
  nd/.style  = {rectangle, draw=black, semithick,
                minimum width=0.82cm, minimum height=0.82cm,
                font=\small\itshape, fill=biddingcol, inner sep=0pt},
  %% repair candidate highlighted in green border
  ndS/.style = {rectangle, draw=green!60!black, very thick,
                minimum width=0.82cm, minimum height=0.82cm,
                font=\small\itshape, fill=biddingcol, inner sep=0pt},
  %% Eve vertex: red circle
  evn/.style = {circle, draw=red!60!black, semithick, fill=evecol,
                minimum size=0.88cm, font=\scriptsize\bfseries, inner sep=1pt},
  %% Adam vertex: blue diamond
  adn/.style = {diamond, draw=blue!60!black, semithick, fill=adamcol,
                minimum size=1.08cm, font=\scriptsize\bfseries, inner sep=0pt},
  %% plain undirected vertex (left graph)
  hv/.style  = {circle, draw=black!65, fill=gray!10,
                minimum size=0.70cm, font=\small\itshape},
  %% terminal (sink) vertices
  sink/.style = {circle, draw=black!65, semithick, fill=white,
                 minimum size=0.70cm, font=\small\bfseries},
  %% arrows
  arr/.style = {->, >=Stealth, semithick},
  %% undirected edge
  ued/.style = {semithick},
}
 
\begin{tikzpicture}[scale=1.0]
 
% =============================================================================
%  LEFT PANEL – Vertex Cover instance H
% =============================================================================
 
\node[font=\normalsize\bfseries] at (2.4, 7.6)
  {\textsc{Vertex Cover} instance $((Q,F),k=2)$};
 
%% Vertices of H (a path v1–v2–v3–v4)
\node[hv] (u1) at (0.6, 6.2) {$q_1$};
\node[hv] (u2) at (1.8, 6.2) {$q_2$};
\node[hv] (u3) at (3.0, 6.2) {$q_3$};
\node[hv] (u4) at (4.2, 6.2) {$q_4$};
 
%% Undirected edges, labelled
\draw[ued] (u1) -- node[above, font=\scriptsize]{$e_1$} (u2);
\draw[ued] (u2) -- node[above, font=\scriptsize]{$e_2$} (u3);
\draw[ued] (u3) -- node[above, font=\scriptsize]{$e_3$} (u4);
 
%% Parameters annotation
\node[align=center, font=\scriptsize] at (2.4, 5.2)
  {%Budget $k = 2$,\quad threshold target $t < \tfrac{1}{2}$\\[3pt]
   $Q=\set{q_1,q_2,q_3,q_4}$, $F=\set{e_1,e_2,e_3}$,\\[3pt]
   Vertex cover $S$ of size $2$: $\{q_2, q_3\}$};
 
%% Arrow indicating the reduction direction
\draw[->, >=Stealth, very thick, gray!70]
  (4.9, 6.2) -- node[above, font=\small, black]{reduction} (6.0, 6.2);
 
% =============================================================================
%  SEPARATOR
% =============================================================================
\draw[dashed, gray!55, thick] (5.7, -0.8) -- (5.7, 8.1);
 
% =============================================================================
%  RIGHT PANEL – Generalized Bidding Game G
% =============================================================================
 
\node[font=\normalsize\bfseries] at (11.0, 7.6)
  {Generalized Bidding Game $\G$};
 
% ── Adam's initial vertex ────────────────────────────────────────────────────
\node[adn] (init) at (10.75, 6.8) {$\vin$};
 
% ── Eve vertices (one per edge of H) ─────────────────────────────────────────
\node[evn] (e1) at ( 7.75, 5.0) {$v_{e_1}$};
\node[evn] (e2) at (10.75, 5.0) {$v_{e_2}$};
\node[evn] (e3) at (13.75, 5.0) {$v_{e_3}$};
 
% %% Small annotation: which edge each ε represents
% \node[font=\scriptsize, text=gray!70!black] at ( 7.75, 4.28) {$(q_1,q_2)$};
% \node[font=\scriptsize, text=gray!70!black] at (10.75, 4.28) {$(q_2,q_3)$};
% \node[font=\scriptsize, text=gray!70!black] at (13.75, 4.28) {$(q_3,q_4)$};
 
% ── Bidding vertices (one per vertex of H) ───────────────────────────────────
%    b_{v2} and b_{v3} are in the repair set S, drawn with green border
\node[nd]  (b1) at ( 6.5,  2.8) {$v_{q_1}$};
\node[ndS] (b2) at ( 9.25, 2.8) {$v_{q_2}$};   %% ← in repair set S
\node[ndS] (b3) at (12.25, 2.8) {$v_{q_3}$};   %% ← in repair set S
\node[nd]  (b4) at (15.0,  2.8) {$v_{q_4}$};
 
%% Label the repair set
\draw[green!60!black, thick, rounded corners=4pt, dashed]
  (8.7, 2.25) rectangle (12.85, 3.35);
\node[font=\scriptsize, text=green!55!black] at (10.775, 1.95)
  {repair set $S = \{v_{q_2}, v_{q_3}\}$};
 
% ── Sink / terminal vertices ─────────────────────────────────────────────────
\node[sink, double, double distance=1.2pt] (T) at ( 9.25, 0.5) {$T$};
\node[sink]                               (L) at (12.25, 0.5) {$L$};
 
%% Labels
\node[font=\scriptsize, below=4pt of T] {(win)};
\node[font=\scriptsize, below=4pt of L] {(lose)};
 
% =============================================================================
%  EDGES
% =============================================================================
 
%% v_init → Eve vertices
\draw[arr] (init) -- (e1);
\draw[arr] (init) -- (e2);
\draw[arr] (init) -- (e3);
 
%% Eve → Bidding (no crossings: each ε_i flanks its two b_{v_j} targets)
%%  ε_{e1} = (v1,v2)  →  b_{v1}, b_{v2}
\draw[arr] (e1) -- (b1);
\draw[arr] (e1) -- (b2);
 
%%  ε_{e2} = (v2,v3)  →  b_{v2}, b_{v3}
\draw[arr] (e2) -- (b2);
\draw[arr] (e2) -- (b3);
 
%%  ε_{e3} = (v3,v4)  →  b_{v3}, b_{v4}
\draw[arr] (e3) -- (b3);
\draw[arr] (e3) -- (b4);
 
%% Bidding → Sinks  (all b_{vi} go to both T and L)
%%  use gentle bends to reduce visual clutter
\draw[arr] (b1) to[bend right=18] (T);
\draw[arr] (b1) to[bend right=25] (L);
 
\draw[arr] (b2) -- (T);
\draw[arr] (b2) to[bend right=8]  (L);
 
\draw[arr] (b3) to[bend left=8]   (T);
\draw[arr] (b3) -- (L);
 
\draw[arr] (b4) to[bend left=25]  (T);
\draw[arr] (b4) to[bend left=18]  (L);
 
% % =============================================================================
% %  LEGEND  (bottom-right area)
% % =============================================================================
% \begin{scope}[shift={(6.5, -1.1)}, every node/.style={font=\scriptsize}]
%   \node[nd,  label={[font=\scriptsize]right:\ bidding vertex}]   at (0,   0) {};
%   \node[evn, label={[font=\scriptsize]right:\ Eve's vertex}]     at (2.8, 0) {};
%   \node[adn, label={[font=\scriptsize]right:\ Adam's vertex}]    at (5.6, 0) {};
%   \node[ndS, label={[font=\scriptsize]right:\ repair candidate}] at (8.6, 0) {};
% \end{scope}
 
\end{tikzpicture}